\documentclass[aps,twocolumn,floatfix, superscriptaddress,nofootinbib]{revtex4-1} %linenumbers, prb

\usepackage{amsmath,amssymb}
\usepackage{braket}
\usepackage{mathtools}
\usepackage{xcolor, ulem}
\usepackage{siunitx}
\usepackage{CJK} % display Chinese, Korean or Japanese name
\usepackage{lipsum}
\usepackage{graphicx}
\usepackage{footnotebackref}
\usepackage{algorithm}
\usepackage[noend]{algpseudocode}
\usepackage{setspace}

\usepackage{dsfont}

\newcommand{\1}{{\widehat{\mathds{1}}}}
\newcommand{\W}{\widehat{W}}

\newcommand{\Q}{\widehat{Q}}
\renewcommand{\P}{\widehat{P}}
\newcommand{\BB}{\widehat{\v{b}}}%{\widehat{B}}
\newcommand{\CC}{\widehat{\v{c}}}
\newcommand{\DD}{\widehat{d}}
\newcommand{\PP}{\mathbb{P}}
\newcommand{\UCMPO}{\{\W^{(n)}\}}

\newcommand{\Wa}{\widehat{W}^{(1)}}
\newcommand{\Wb}{\widehat{W}^{(2)}}

\newcommand{\Qa}{\widehat{Q}^{(1)}}
\newcommand{\Qb}{\widehat{Q}^{(2)}}

\newcommand{\Pa}{\widehat{P}^{(1)}}
\newcommand{\Pb}{\widehat{P}^{(2)}}

\newcommand{\Z}{\mathbb{Z}}

\newcommand{\N}{\mathbb{N}}

\newcommand{\chiMPO}{D}
\newcommand{\chiMPS}{\chi}

\newcommand{\SM}{$\mathrm{SM}_{y}$}
\newcommand{\tBLG}{tBLG}

\newcommand{\n}[1]{\left| #1 \right|}%%adjustable-height norm shortcut
\newcommand{\dn}[1]{|| #1 ||}%%adjustable-height norm shortcut
\newcommand{\st}[1]{\left\{#1\right\}}%%adjustable-height set notation
\newcommand{\setc}[2]{\{#1\; :  \; #2 \}}

\renewcommand{\v}[1]{\boldsymbol{#1}}%%shortcut to make a vector (overwrites the default command)
\DeclareMathOperator{\Tr}{Tr}
\DeclareMathOperator{\tr}{tr}

\DeclareMathOperator{\diag}{diag}

\def\IP#1{\mathinner{\langle#1\rangle}}

\usepackage{pifont}% http://ctan.org/pkg/pifont
\usepackage{amsthm}
\newtheorem{thm}{Theorem}

\newtheorem{prop}[thm]{Proposition}

\newtheorem{lemma}[thm]{Lemma}

\theoremstyle{definition}

\newtheorem{defn}[thm]{Definition}

\theoremstyle{remark}

\usepackage{hyperref} % should be loaded as late as possible because it overrides commands
\usepackage{float} % float is special. should be loaded after hyperref

\begin{document}

    \begin{CJK*}{UTF8}{min}

	\title{Efficient simulation of moire materials using the density matrix renormalization group}
	
	\author{Tomohiro Soejima (副島智大)}
	\thanks{These two authors contributed equally}

	\affiliation{Department of Physics, University of California, Berkeley, CA 94720, USA}
	\author{Daniel E. Parker}
	\thanks{These two authors contributed equally}
	%\email[]{daniel\_parker@berkeley.edu}
	\affiliation{Department of Physics, University of California, Berkeley, CA 94720, USA}
	\author{Nick Bultinck}
	\affiliation{Department of Physics, University of California, Berkeley, CA 94720, USA}
	\affiliation{Department of Physics, Ghent university, 9000 Ghent, Belgium}
	\author{Johannes Hauschild}
	\affiliation{Department of Physics, University of California, Berkeley, CA 94720, USA}
	\author{Michael P. Zaletel}
		\affiliation{Department of Physics, University of California, Berkeley, CA 94720, USA}
	\affiliation{Materials Sciences Division, Lawrence Berkeley National Laboratory, Berkeley, California 94720, USA
}
	
		\date{\today}
	
	\begin{abstract}
	We present an infinite density-matrix renormalization group (DMRG) study of an interacting continuum model of twisted bilayer graphene (tBLG) near the magic angle. Because of the long-range Coulomb interaction and the large number of orbital degrees of freedom, tBLG is difficult to study with standard DMRG techniques --- even constructing and storing the Hamiltonian already poses a major challenge. To overcome these difficulties, we use a recently developed compression procedure to obtain a matrix product operator representation of the interacting tBLG Hamiltonian which we show is both efficient and accurate even when including the spin, valley and orbital degrees of freedom.
	To benchmark our approach, we focus mainly on the spinless, single-valley version of the problem where, at half-filling, we find that the ground state is a nematic semimetal. 
    Remarkably, we find that the ground state is essentially a $k$-space Slater determinant, so that Hartree-Fock and DMRG give virtually identical results for this problem. Our results show that the effects of long-range interactions in magic angle graphene can be efficiently simulated with DMRG, and opens up a new route for numerically studying strong correlation physics in spinful, two-valley tBLG, and other moire materials, in future work.

	\end{abstract}
	\maketitle
	
	\end{CJK*}

	\tableofcontents
	\section{Introduction}

	Magic angle twisted bilayer graphene (\tBLG{}) hosts a diverse array of correlated insulating and superconducting phases \cite{Cao,Cao2,Yankowitz,Kerelsky,RutgersSTM,efetov,YazdaniSpectroscopic,Efetovscreening,YoungScreening,Choi,yoo2019Atomic,Sharpe,YoungAH,Tomarken,CascadeYazdani,CascadeShahal,Harpreet,YazdaniChern,AndreiChern,Andreananosquid,EfetovFragile,VafekLi,PabloNematicity}. This rich system has inspired intensive theoretical efforts to understand the origin and mechanism(s) behind these phases, and a large number of theories already have been proposed. One way to assess these proposals --- especially when they are not associated to clear experimental signatures --- is numerical calculation. To that end, this work presents a proof-of-concept density matrix renormalization group (DMRG) \cite{white1992density} study of a microscopically realistic, strongly interacting model of \tBLG{}.\cite{KangVafek}

\subsection{Challenges of \tBLG{} Numerics}

	\begin{figure}
	    \centering
	    \includegraphics{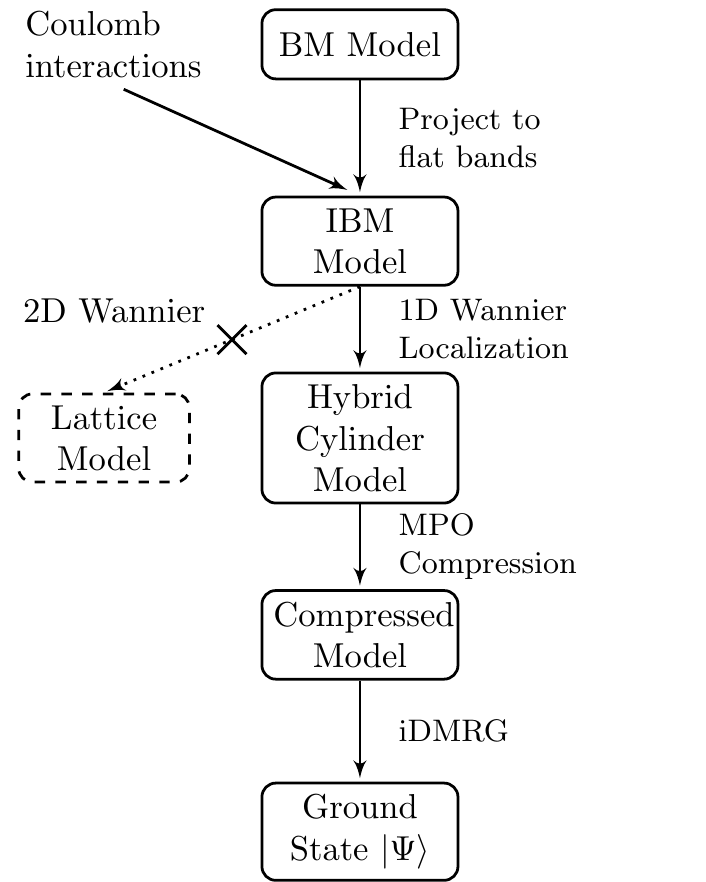}
	    \caption{
	    Flowchart of our approach to DMRG for \tBLG{}. First we start from a continuum BM model, and add Coulomb interactions, projected to the flat bands to reduce the number of degrees of freedom down to a manageable level. Second, we perform hybrid Wannier localization, which maps the model to a cylinder in mixed-$xk$ space, thereby avoiding a topological obstruction and allowing all symmetries to act locally. Third, we use a compression procedure to represent the long-range interactions with a reasonable bond dimension to make DMRG numerically tractable. This allows us to perform DMRG with all $N_B = 8$ components at moderate cylinder radius ($L_y = 6$).
	    }
	    \label{fig:flowchart}
	\end{figure}

	Let us review what makes the \tBLG{} problem so numerically challenging, and identify a viable path around the obstacles. 
    The first obstacle is the separation in scales between the graphene lattice constant $a$ and the moir\'e length scale $L_M$; at the magic angle $L_M/a \sim 1/\theta_M \sim 50$, so the moir\'e unit cell contains over $10,000$ carbon atoms, and consequently the superlattice band structure contains $N_B \sim 10,000$ bands.  
	Fortunately, various treatments of the band structure \cite{bistritzer2011Moire,NamKoshino,CarrKaxiras,Bernevig} (including the Bitzritzer-MacDonald (BM) continuum model \cite{bistritzer2011Moire} used here) reveal that the flat bands of interest  are separated from the tower of ``remote'' bands by gaps of order \SIrange{20}{25}{\milli\electronvolt} (see Fig.~\ref{fig:IBM_model}(c)). 
	Since these gaps are larger than the Coulomb scale  $E_C = \frac{e^2}{4 \pi \epsilon_0\epsilon_r L_M} \sim$ \SIrange{10}{20}{\milli\electronvolt} (using a relative permittivity $\epsilon_r = 12$ -- $6$), it is a reasonable starting point to project the Coulomb interaction $V(r)$ into the flat bands.\footnote{Hartree-Fock studies which include the remote bands do find that they have a quantitative effect (for example, on the magnitudes of the symmetry-broken gaps), but there are some discrepancies regarding their qualitative importance \cite{bultinck2019Ground,XieMacDonald}.} 
	Each spin and valley of the graphene has \emph{two} flat bands, for a total of eight,  winning us a reduction from $N_B = 10,000 \searrow 8$. We refer to this as the ``Interacting Bitzritzer-MacDonald (IBM) model,'' although our method works just as well for improved continuum models of \tBLG{} which take into account effects like lattice relaxation. The touching of these two bands is locally protected by a crucial $C_2 \mathcal{T}$ symmetry (a 180-degree rotation combined with time-reversal), which distinguishes \tBLG{} from other moir\'e materials. The Coulomb scale $E_C$ is much larger than the bandwidth $t \sim $ \SI{5}{\milli\electronvolt}, so \textit{a priori} unbiased, strongly-interacting numerical approaches such as exact diagonalization \cite{repellin2020Ferromagnetism}, determinantal quantum Monte Carlo,  or DMRG are required. 

	Most strongly-interacting approaches proceed from a real space lattice model, so a natural next step is to construct a lattice model via 2D Wannier localization of the continuum Bloch bands.
	In real space, the density of states of the flat bands is predominantly located on the AA-stacking regions of the moir\'e unit cell, which form a triangular lattice (see Fig.~\ref{fig:IBM_model}d). So one might hope that the physics is then well-described by an 8-component triangular lattice Hubbard model.
	However, there is a topological obstruction which complicates this approach: the flat bands possess ``fragile topology'' which makes their Wannier localization very subtle \cite{Po,Po2,ZouPo,Bernevig,Hejazi,LiuDai,ahn2019Failure}. In particular, the presence of $C_2\mathcal{T}$, valley conservation $U_v(1)$ and translation make it impossible to Wannier localize the flat bands in a manner where $U_v(1)$ and $C_2\mathcal{T}$ both act in a strictly local fashion. This is somewhat analogous to the obstruction to finding a Wannier basis for a 2D topological insulator under the requirement that $\mathcal{T}$ acts as a permutation of the orbitals \cite{VanderbiltTI}.
	
	Two resolutions to the Wannier obstruction issue have been proposed in the literature. The conceptually simplest is to include some number of remote bands, at minimum $N_B = 8 \to 20$, which removes the topological obstruction and allows for a local symmetry action \cite{Po2}. But from a DMRG standpoint, a model with 20 orbitals per unit cell, all strongly-interacting, appears to be numerically intractable.
	The other approach is to simply ignore the symmetry considerations and Wannier-localize in a basis which hybridizes different valleys or $C_2\mathcal{T}$ sectors. In this approach, for example, valley number conservation $U_v(1)$ becomes slightly non-local, and the associated charge takes the form $Q_V = \sum_{i,j,m,n} Q^{ij}_{mn} \hat{c}_{m,i}^\dagger \hat{c}_{n,j}$ where the sum runs over all sites $i,j$ and internal degrees of freedom $m,n$. The matrix elements $Q^{ij}_{mn}$ fall off with distance $|r_i - r_j|$ \cite{Po}. Intriguingly, the Wannier orbitals then take the shape of three-lobed ``fidget spinners" connecting three nearby AA regions \cite{Po,KangVafekPRX,KoshinoYuan}. In this basis, the Coulomb interaction is not dominated by a $U \hat{n}^2$ Hubbard interaction, but instead contains a profusion of all allowed $V^{ijk\ell}_{mnop} \hat{c}_{m, i}^\dagger \hat{c}_{n, j}^\dagger \hat{c}_{o, k} \hat{c}_{p,\ell}$ terms which decay exponentially over a few moir\'e sites \cite{KangVafekPRX,KoshinoYuan,KangVafekPRL}. Numerically, however, the interactions must be cut off at some finite range, which will spuriously break either the $U_v(1)$ or $C_2 \mathcal{T}$ symmetry due to the non-local form they take. This runs the risk of biasing the results by explicitly breaking a symmetry which should be preserved, and would require careful extrapolation of the tails to ensure the correct results. While not necessarily unworkable (in particular, see Ref. \cite{WangVafek}), in our estimation this approach makes numerical results  delicate  to interpret.

    Fortunately, DMRG is a 1D algorithm, which allows us to avoid the construction of 2D Wannier orbitals altogether.
    When DMRG is applied to the cylinder geometry, the model must be in a localized basis along the length of the cylinder, to ensure favorable entanglement properties, but it does not need to be in a localized basis around its circumference.
    Therefore, we can consider  ``hybrid'' real-space/momentum-space Wannier states which are maximally localized along the length of the cylinder $x$, but $T_y$-eigenstates around its circumference $y$ (see Fig \ref{fig:wannier_real_space}).
    There is no topological obstruction to the construction of hybrid Wannier states, making them an attractive basis for the flat bands of magic angle graphene, as was also recognized by the authors of Refs. \cite{bultinck2020Mechanism, KangVafek,HejaziHybridWannier,KwanDomainwalls}. 
	Geometrically, this defines a model on a cylinder, with real-space in the $x$ direction and $k$-space around the circumference\cite{motruk2016Density}.
	The hybrid approach allows the  $U_v(1)$, $C_2 \mathcal{T}$, and translation symmetries to all act locally without adding extraneous degrees of freedom. This is exactly the approach used by Kang and Vafek in their recent DMRG study of \tBLG{} \cite{KangVafek}, and it is the approach we take as well. 
	
	The hybrid approach is not without challenge, however, because upon mapping the orbitals to a 1D fermion chain for input into the DMRG, the effective Hamiltonian is quite long-ranged.
	The localization width of the Wannier orbitals is comparable to the moir\'e scale, so all the sites in a single column of the cylinder are strongly overlapping, generating a panoply of couplings $V^{ijk\ell}_{mnop} \hat{c}_{m,i}^\dagger \hat{c}_{n,j}^\dagger \hat{c}_{o,k} \hat{c}_{p,\ell}$. Though these decay exponentially with distance, a cylinder with circumference $L_y = 6$ with both spin and valley has on the order of $850,000$ non-negligible (i.e. above $10^{-2}$ \si{\milli\electronvolt}) matrix elements per unit cell.
	
	A similar problem is encountered in the context of cylinder-DMRG for the fractional quantum Hall effect \cite{ZaletelMulticomponent}, or finite-DMRG simulations for quantum chemistry problems \cite{ChanKeselman}. There, as here, it is essential to use tensor network methods to ``compress'' the $V^{ijkl}$ as a matrix product operator (MPO). To do so, we leverage a recent algorithm for black-box compression of Hamiltonian MPOs with various optimality properties \cite{parker2019local}.
	We find that for a circumference $L_y = 6$ cylinder, the spinless / single-valley problem ($N_B = 2$) requires an MPO bond dimension of $D \sim 100$ for physical observable to obtain a relative precision of $10^{-2}$, while in the spinful / valleyful $N_B = 8$ case, we estimate the required bond dimension to be $D \sim 1000$. While large, these values are tractable, especially when exploiting the charge, spin, valley, and $k_y$ quantum numbers.

	\subsection{Overview of DMRG Results}
	
	After presenting details of the interacting \tBLG{} Hamiltonian and its MPO compression, we apply our approach in detail to a ``toy'' $N_B = 2$ problem in which we keep only valley $K$ and spin $\uparrow$; more physical models are reserved for future work.
    When filling 1 of the 2 bands, this scenario is conceptually similar to fillings $\nu = -3, 3$ of \tBLG{} under the assumption that these fillings are spin and valley polarized. However, we caution the reader that our results are not a quantitative prediction for these fillings because the toy model differs from $|\nu| = 3$ of \tBLG{} by a 4$\times$ difference in the magnitude of the Hartree potential generated relative to neutrality. We've made this choice so that we can quantitatively compare with Refs.~\cite{KangVafek} prior results; the ``physical'' $|\nu| = 3$ result, which differs in some interesting respects, will be presented in a future work.
	
	Following Ref.~\cite{KangVafek}, we fix $\theta = 1.05^\circ$ and vary the ratio of the AA and AB inter-layer tunneling hopping strengths $w_0 / w_1$ from $0$ to $0.9$.
	While physically $w_0 / w_1 \sim 0.8$ \cite{NamKoshino,KoshinoYuan,CarrKaxiras}, the resulting phase diagram is conceptually interesting
	because the dominant effect of $w_0 / w_1$ is to redistribute the Berry curvature of the flatbands, rather than changing their bandwidth, revealing that  the former is crucial to the physics. As a pr\'ecis of our findings, 
	\begin{enumerate}
	\item In agreement with Ref.~\cite{KangVafek}, we find that below a critical value  $w_0 / w_1 \lesssim 0.8 $, the ground state spontaneously breaks the $C_2\mathcal{T}$ symmetry, forming a quantum anomalous Hall state (Chern insulator), with $C = \pm 1$.
	
	\item  In agreement with Ref.~\cite{KangVafek}, above $w_0 / w_1 \gtrsim 0.8 $, the $C_2 \mathcal{T}$ is restored. In this region the DMRG results of Ref.~\cite{KangVafek} did not reliably converge, but their mean-field calculations suggested either a  ``nematic $C_2 \mathcal{T}$ symmetric semimetal'' first proposed in Ref.~\cite{liu2020Nematic}, or a gapped $C_2 \mathcal{T}$-symmetry stripe \cite{KangVafek}.
	Our DMRG numerics reliably converge to a  state in excellent agreement with  the nematic $C_2 \mathcal{T}$-semimetal, with two band touchings near the $\Gamma$ point.
	
    \item  We analyze the $k$-space  electron correlation function $P_{mn}(\mathbf{k}) = \langle c^\dagger_{n, \mathbf{k}} c_{m, \mathbf{k}} \rangle$ of the DMRG ground state, which can be directly compared with Hartree-Fock calculations. We find that both phases are extremely well captured by a single $k$-space Slater determinant (to within $\approx $ 1\%), strongly supporting the validity of recent Hartree-Fock studies \cite{bultinck2019Ground,liu2020Nematic,XieMacDonald,KangVafek,Choi,YazdaniSpectroscopic,Cea,HejaziHybridWannier,KwanDomainwalls,KwanExciton}. 

    \item Finally, we compare the energy of the DMRG ground state with various competing variational ansatz such as the $C_2 \mathcal{T}$-stripe ansatz proposed in Ref.~\cite{KangVafek}. In agreement with their result, we find that the nematic semimetal and the $C_2 \mathcal{T}$ stripe compete at the order of 0.1 meV per unit cell.
	 \end{enumerate}

\begin{figure*}
    \centering
    \includegraphics[width=\textwidth]{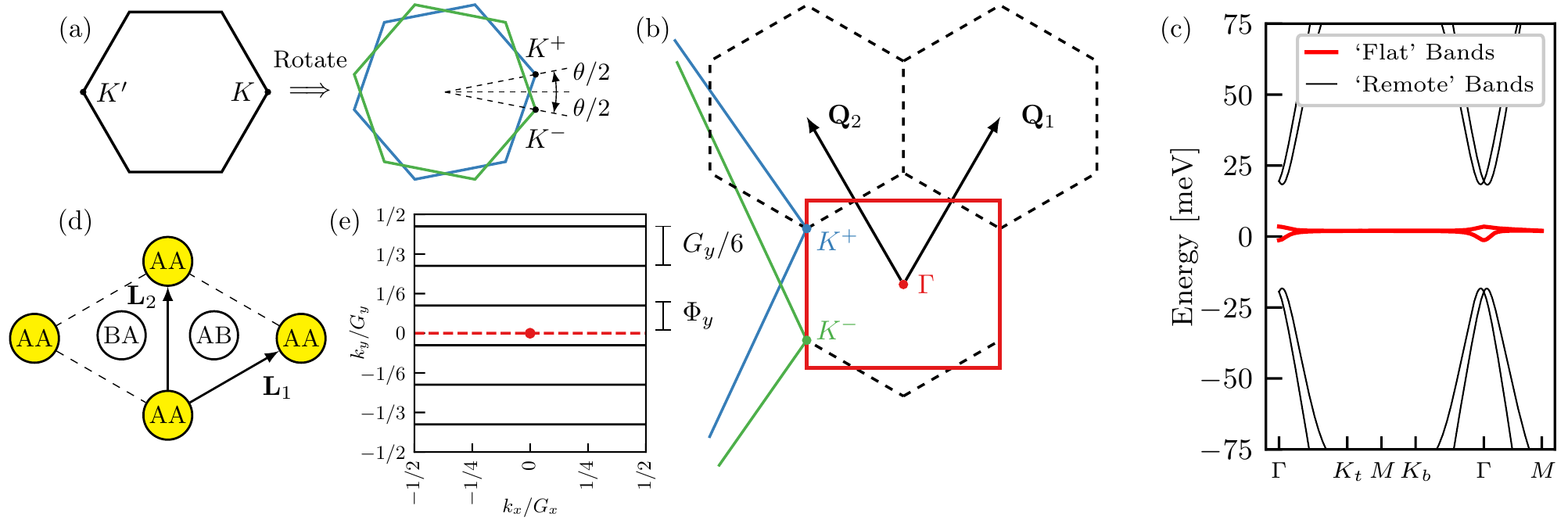}
    \caption{(a) The BM model for bilayer graphene is constructed from two regular graphene Brillouin zones, rotated by $\pm \theta/2$. (b) Zoomed view of (a) showing the mini (or moir\'e) Brillouin zone. We choose a square mBZ (thick red lines) for numerical convenience. (c) The band structure of the BM model over the mBZ showing the flat bands (red lines). The interacting BM model is defined by adding Coulomb interactions to the BM model and projecting to the flat bands. (d) Schematic of the real space moir\'e unit cell with Bravais lattice vectors $\v{L}_{1,2}$. (e) The IBM model on a cylinder has $N_y$ discrete momentum cuts at $k_y$ values given by Eq.~\eqref{eq:ky_cuts}, offset  by $\Phi_y$.
    }
    \label{fig:IBM_model}
\end{figure*}

	The remainder of this work is organized as follows. Section \ref{sec:IBM_model} introduces our model: an interacting Bistritzer-MacDonald model, equipped with long range Coulomb interactions, and projected to the flat bands. Section \ref{sec:MPOs} discusses how the model may be expressed as a Matrix Product Operator and both why and how it must be compressed to perform DMRG. Section \ref{sec:half_filling_physics} provides the results of DMRG calculations, and shows that Hartree-Fock accurately captures the ground state physics in this model. Section \ref{sec:SM_phase} discusses the nature of the nematic $C_2\mathcal{T}$-semimetal. We conclude in Section \ref{sec:conclusion}. Extensive Appendices describe all details needed to reproduce our results. Appendix \ref{app:IBM_model} details the IBM model. Appendix \ref{app:gauge_fixing} deals with the Wannier localization and the gauge choice we make. App.~\ref{app:uncompressed_MPO_construction} constructs the pre-compression MPO: an infinite MPO with arbitrary long range 4-body interactions. App.~\ref{app:UCMPO_compression} provides the algorithm for MPO compression, as well as rigorous error bounds.  Finally, App.~\ref{app:numerical_details} explains the extensive numerical cross-checks we performed to ensure the accuracy of our results.

\section{The IBM Model}
\label{sec:IBM_model}

This section describes the interacting generalization of the Bistritzer-MacDonald (BM) model we use in this work. We first briefly recall the BM model and the geometry of the mini-Brillouin zone (mBZ), then discuss how interactions are added. We then show how the model can be placed on a cylindrical geometry, and conclude with the symmetries of the model.

\subsection{Continuum Model}

Our starting point is the single-particle Bistritzer-MacDonald (BM) model \cite{bistritzer2011Moire}, composed of two layers of graphene, with relative twist angle $\theta$, coupled together by a spatially-varying moir\'e potential. The potential is governed by two parameters, $w_0$ and $w_1$, which specify the AA / AB interlayer tunneling respectively. 
DFT calculations which account for lattice relaxation find that $w_1 = \SI{109}{\milli\electronvolt}$ and $w_0 / w_1 \approx 0.8$ \cite{NamKoshino,KoshinoYuan,CarrKaxiras}, but here we will treat $w_0/w_1$ as an axis of the phase diagram. 
We maximize the ratio of band gap to band width for the flat bands by setting $\theta_{BM} \sim \ang{1.05}$. Figure \ref{fig:IBM_model} details our choice of conventions. In particular, we work with a rectangular mBZ grid for numerical convenience.

We now define an interacting Bistritzer-MacDonald (IBM) model where double-gate screened Coulomb interactions are added to the single-particle model.
As the interactions are much larger than the spectral width of the flat bands, but smaller than the gap to nearby bands, we expect interactions to act quite non-perturbatively inside the flat bands and perturbatively between seperated bands. We therefore project the interactions to the two flat bands, akin to models of the fractional quantum hall effect \cite{qi2011Generic}. Our presentation will focus on a single spin and valley, but their inclusion is conceptually identical: we promote $2 \to 8$. Consider a vector of fermions $\v{f}^{\dagger}_{\v{k}} = \left(f^\dagger_{1,\v{k}},f^\dagger_{2,\v{k}} \right)$ running over the two nearly flat bands. The Hamiltonian is then given by

\begin{equation}\label{eq:IBM}
    \hat{H} = \sum_{\v{k}\in \text{mBZ}} \v{f}^{\dagger}_{\v{k}} h(\v{k}) \v{f}_{\v{k}}  + \frac{1}{2A} \sum_{\v{q}} V_{\v{q}} : \rho_{\v{q}} \rho_{-\v{q}} :\, .
\end{equation}
The single-particle term $h(\v{k})$ contains not only the flat band energies of the BM model, but also band renormalization terms coming from the interaction with the filled remote bands, and a subtraction to avoid double counting of Coulomb interaction effects.
We refer to Appendix \ref{app:IBM_model} for more details.

The second term in Eq.~\eqref{eq:IBM} corresponds to the dual gate-screened Coulomb interaction, with $A$ being the sample area and $V_{\v{q}} = e^2 \tanh(\n{\v{q}} d)/ (2\epsilon_r\epsilon_0 \n{\v{q}})$. The screened Coulomb potential depends on two parameters $\epsilon_r$ and $d$, which respectively are the relative permittivity and the distance between the twisted bilayer graphene device and the metallic gates. 
While the effective dielectric constant of the typical substrate, hBN, is $\epsilon_r \approx 4.4$, here we use $\epsilon_r = 12$ in order to phenomenologically account for screening from the remote bands of the tBLG \ref{app:IBM_model}. This sets the typical interaction energy scale to be several \si{\milli\electronvolt}.
For the gate distance we choose $d = 10$ nm to facilitate comparison with Ref.~\cite{KangVafek}.
The Fourier components of the flat-band projected charge density operator are given by
\begin{equation}
    \rho_{\v{q}} = \sum_{\v{k} \in \mathrm{mBZ}} \v{f}_{\v{k}}^{\dagger} \Lambda_{\v{q}}(\v{k}) \v{f}_{\v{k}+ \v{q}}\, ,
\end{equation}
where the $2\times 2$ form factor matrices $[\Lambda_{\v{q}}(\v{k})]_{ab} =   \braket{\psi_{a,\v{k}}|e^{-i \v{q} \cdot \v{r} }| \psi_{b,\v{k}+\v{q}}}$ are defined in terms of overlaps between the Bloch states of the BM model.

The model enjoys several global symmetries: time reversal followed by in-plane rotation $C_2 \mathcal{T}$, out-of-plane $C_{2x}$ rotation, and $C_3$ rotation. We will describe their action on the basis states explicitly below. In summary, the spinless, single-valley IBM model we have described is a strongly interacting many-body problem defined in momentum space over the mini-Brillouin Zone.

\subsection{Cylinder Model}
\label{sec:cylinder_model}

Our goal is to perform quasi-2D DMRG on Eq.~\eqref{eq:IBM}. To this end, we work in an infinite cylinder geometry of circumference $N_y$ with a mixed real and momentum space representation of the model. In the momentum space, this corresponds to having $N_y$ momentum cuts through the mBZ at
\begin{equation}
    k_y/G_y = \frac{n + \Phi_y/(2\pi)}{N_y} \pmod 1, \quad 0 \le n \le N_y-1
    \label{eq:ky_cuts}
\end{equation}
where $\Phi_y$ is the amount of flux threaded through the cylinder, which offsets the $y$ momentum as shown in Fig.~\ref{fig:IBM_model}. We will Fourier transform each of these momentum cuts in the $x$ direction, such that our basis states are hybrid Wannier orbitals, periodic in $y$ direction and localized in $x$ direction.\footnote{A further advantage of this mixed representation over ``snaking'' around a real space cylinder is that $k_y$ becomes a good quantum number, reducing the resource cost for a given radius\cite{motruk2016Density,ehlers2017hybrid}.}  We will sometimes call this mixed $xk$ representation.

The choice of real-space basis in the $x$ direction is not unique, but we choose the basis of maximally localized Wannier orbitals. Using the maximally-localized orbitals ensures that the interactions are as short-ranged as possible and hence minimizes the range of the interaction terms in the Hamiltonian and the entanglement of the ground state. 
Due to the relation between maximal Wannier localization and the Bloch Berry connection $\mathbf{A}_{\v{k}}$, this basis will also make manifest their topology. We perform the change of basis: 
\begin{equation}
    \begin{split}
        \hat{c}_{\pm,k_x,k_y} ^{\dagger} &:= U_{\pm,b}(\v{k}) \hat{f}_{b, k_x, k_y}^{\dagger},\\
        \hat{c}_{\pm,n,k_y} ^{\dagger} &:= \int \frac{dk_x}{\sqrt{G_x}} e^{i \v{k} \cdot \v{R}_n} \hat{c}_{\pm,k_x,k_y} ^{\dagger},
    \end{split}
    \label{eq:wannier_change_of_basis}
\end{equation}
where $\hat{c}^\dagger_{\pm, n, k_y}$ is the creation operator for the Wannier orbital for unit cell $n$ in the $x$ direction, $\v{R}_n = n \v{L}_1$ with $\v{L}_1$ the Bravais lattice vector, and $U(\v{k})$ is a $2\times 2$ change-of-basis matrix for the internal (band index) degrees of freedom. The non-trivial topology of the \tBLG{} flat bands \cite{Po,Po2,ZouPo,Bernevig,Hejazi,LiuDai,ahn2019Failure} is made explicit in the hybrid Wannier basis by the fact that the states with subscripts $\pm$ are constructed from bands with Chern numbers $\pm 1$. We will explain this in more detail below. We choose the internal rotations $U(\v{k})$ so that the Wannier orbitals are maximally localized (i.e., their spread in the $x$ direction is minimized). Since the problem is effectively 1D for each $k_y$ cut, we can employ a well-known algorithm \cite{marzari1997Maximally} to deterministically calculate the unique $U(\v{k})$ (up to ($k_y$, $\pm$) dependent phases). Fig.~\ref{fig:wannier_real_space} shows examples of the Wannier orbitals. One can see they are localized in the $x$ direction but extended and periodic in $y$. The charge density is also not uniform in the $y$ direction, but is concentrated in certain regions corresponding to the AA region \cite{bistritzer2011Moire}. For later notational convenience, we also define $\ket{w(\pm, n, k_y)} = \hat{c}^\dagger_{\pm, n, k_y}\ket{0}$.
We emphasize that the $\pm$ basis is not the energy eigenbasis of the single-particle Hamiltonian. In the $\pm$ basis with the gauge convention described later,  the $k.p$ expansion of the two band Hamiltonian around $K^\pm$ points takes the form $h(K^\pm +\v{q}) \propto \mp (q_x \cdot \sigma_x + q_y \cdot \sigma_y)$.

\begin{figure}
    \centering
    \includegraphics[width=\linewidth]{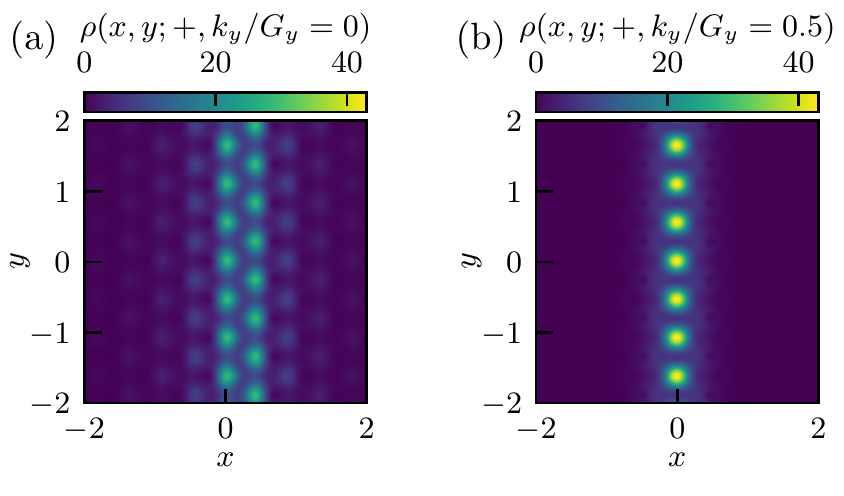}
    \caption{Real space charge density of Wannier orbitals. The orbitals are maximally localized in the $x$ direction and periodic along $y$, with charge densities concentrated in AA stacked regions. The Wannier center and the character changes with $k_y$. }
    \label{fig:wannier_real_space}
\end{figure}

A key physical property of Wannier orbitals is their polarization $P_x(\pm, k_y)$, which can be derived via modern theory of polarization \cite{king-smith1993Theory, resta1993Macroscopic, vanderbilt1993Electric, marzari2012Maximally}. They can be thought of as the center of Wannier orbital inside the zeroth unit cell:  
\begin{equation}
    P_x(\pm,k_y) = \frac{\braket{w(\pm,0,k_y)|\hat{x}|w(\pm,0,k_y)}}{\v{L_1} \cdot \v{e}_1},
\end{equation}
where we normalize the polarization by the $x$ extent of the unit cell.
The polarization is related to the Berry phase along each momentum cut via the Wilson loop $e^{2 \pi i P_x(k_y)} = e^{i  \int A_x(k_x, k_y) dk_x / G_x }$, and is only defined modulo $1$ \cite{marzari2012Maximally}. Therefore $P_x$ returns to itself as $k_y$ sweeps across the  (mini)BZ.
Furthermore, the Chern number of a band is conveniently expressed in terms of the total winding of the polarization,  $C = \int dk_y \frac{dP_x}{d k_y}$. In Fig.~\ref{fig:w0_polarization}, we plot the polarization versus $k_y$ momentum at various different $w_0$. We make two observations: first, we see that the polarization of the plus (minus) band winds from $0$ to $1$ ($0$ to $-1$) as momentum goes from $-0.5$ to $0.5$.
We may therefore identify these bands as having Chern numbers $\pm 1$ --- hence our index convention. On the other hand, the profile of the polarization changes as $w_0/w_1$ increases from $0$ to $0.85$. At $w_0/w_1 = 0$, the slope is constant, and the Wannier orbitals are almost equally spaced in the $x$ direction, reminiscent of the lowest Landau level of a 2D electron gas in a magnetic field. At $w_0/w_1 = 0.85$, however, $P_x$ is constant for most $k_y$ values, and suddenly changes around the $\Gamma$ point.

There is a subtle issue relating our convention for polarization to our choice of gauge for single-particle wavefunctions in the mBZ. Since the polarization increases by $\pm1$ as $k_y$ increases by $2\pi$, we must choose a $k_y$ where the polarization wraps around. We pick the convention that the wrapping $P_x \to P_x \pm 1$ occurs at $k_y =0$, as shown in Fig.~\ref{fig:w0_polarization}. In terms of the Wannier orbitals, this means that their centers of charge move continuously with $k_y$, except at $k_y = 0$ where they ``exit" the unit cell and ``enter" the neighboring unit cell. In terms of the momentum space creation operator $\hat{c}^\dagger_{\pm. \v{k}}$ this corresponds to a a choice of gauge that is smooth in the upper and lower halves of the Brillouin zone, but \textit{discontinuous} across $k_y = 0$. This discontinuity will appear in several figures below.

\begin{figure}
    \centering
    \includegraphics[width = \linewidth]{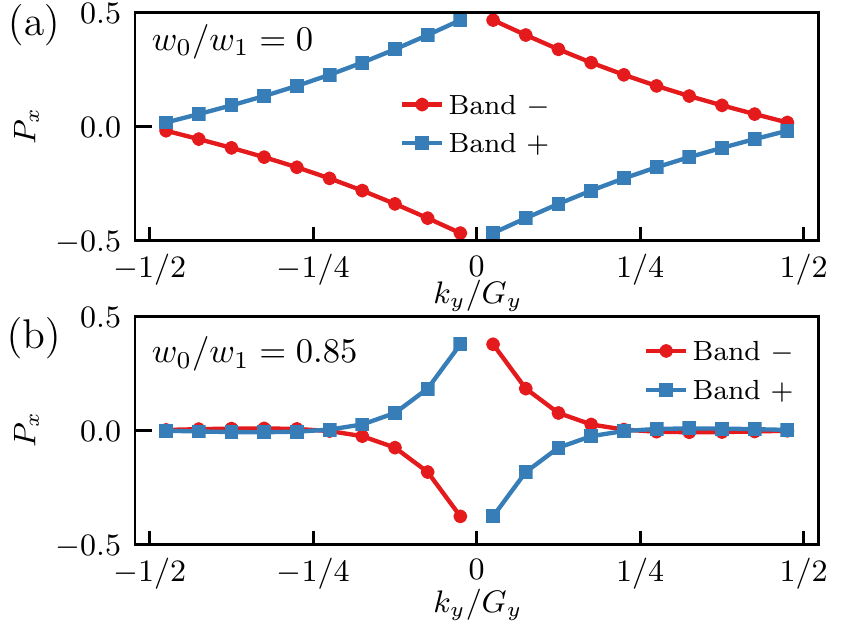}
    \caption{Polarization $P_x$ as a function of $k_y$ and $w_0/w_1$. At each value of $w_0/w_1$, there are two bands with Chern number $\pm 1$, which wraps once as $k_y/G_y$ increases by $1$. This requires a single discontinuity in $P_x$, which we have chosen to place at $k_y = 0$.  One can see that $P_x$ is more linear for $w_0/w_1 = 0$, reflecting flatter Berry curvature.}
    \label{fig:w0_polarization}
\end{figure}

Finally, let us give the explicit action of global symmetries on our basis states. We first note the $C_3$ symmetry of the continuum model is weakly broken by the cylindrical geometry and is no longer an explicit symmetry of the model. We also note that for flux values $\Phi_y \neq 0 \mod \pi$, the $C_{2x}$ symmetry is not present.

Similar to Ref.~\cite{KangVafek}, we partially fix the gauge of the flat band Bloch states such that the symmetries act in a simple way on the hybrid Wannier orbitals:
\begin{equation}
\begin{split}
    T_{L_1}\ket{w(\pm, n, k_y)} & = \ket{w(\pm, n + 1, k_y)}\\
    T_{L_2}\ket{w(\pm, n, k_y)} & = e^{i 2\pi k_y}\ket{w(\pm, n, k_y)}\\
    C_2\mathcal{T} \ket{w(\pm, n, k_y)} & = \ket{w(\mp, -n, k_y)} \\
    C_{2x} \ket{w(\pm, n, k_y)} & = \mp i e^{-i2\pi k_y n}\ket{w(\mp, n, -k_y)}
\end{split}
\label{eq:symmetry_action_cylinder}
\end{equation}
where the last equation holds only at $C_{2x}$ symmetric flux values.
The first two definitions are the consequence of Eq.~\eqref{eq:wannier_change_of_basis}, while the latter two come from demanding the following actions in momentum space:

\begin{equation}
\begin{split}
    C_2\mathcal{T} \hat{c}^\dagger_{\pm, \v{k}} (C_2\mathcal{T})^{-1} &= \sigma^x K \hat{c}^\dagger_{\pm, \v{k}}, \\ 
    C_{2x} \hat{c}^\dagger_{\pm, k_x, k_y} (C_{2x})^{-1} &= \sigma^y \hat{c}^\dagger_{\pm, k_x, -k_y},
\end{split}
\end{equation}
where $\sigma^x$ acts on $\pm$ indices, and $K$ is the complex conjugation operator. This, together with a continuity criterion such that the Wannier functions are smooth function of $k_y$, fixes the phase ambiguity up to an overall minus sign (App. \ref{app:gauge_fixing}).\footnote{In the absence of $C_{2x}$ symmetry, we use a heuristic such that the gauge is continuous as a function of $\Phi_y$.}

Now that we have described the interacting Bistritzer-MacDonald model in detail, we proceed to discuss how we will solve for its ground state using DMRG.

\section{MPO Compression and DMRG}
\label{sec:MPOs}

\begin{figure*}
    \centering
    \includegraphics[width=\linewidth]{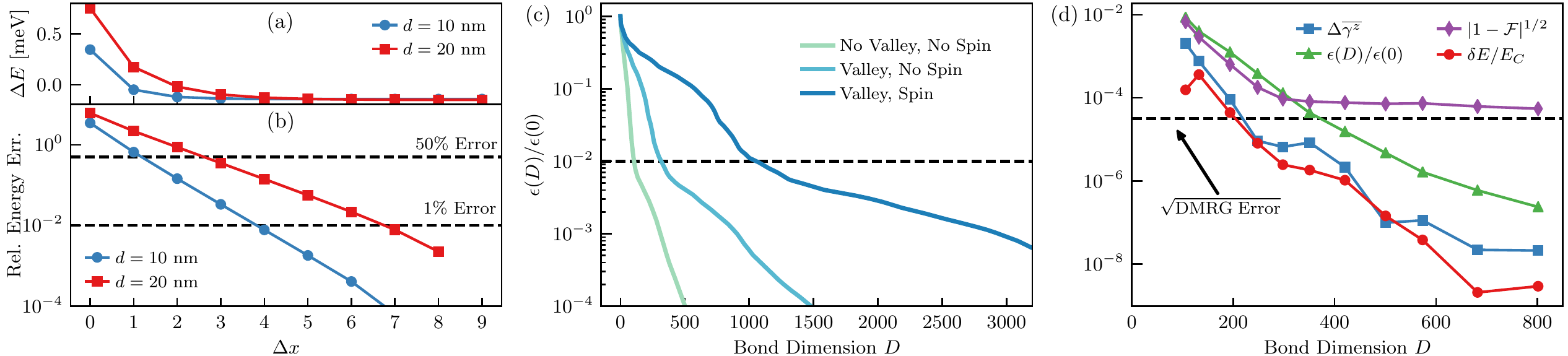}
    \caption{
    (a) Energy difference between the QAH and the SM ansatz at different gate distances at $w_0/w_1 \sim 0.8$.
    (b) The relative error in energy difference between the QAH and SM ansatz. The black dashed lines indicate $50\%$ error and $1\%$ error.
    (c) Precision of the compressed MPO as a function of bond dimension for the IBM model with and without spin and valley degrees of freedom.  The Hamiltonian is the IBM model at the chiral limit $w_0 = 0$ with parameters given by Table~\ref{tab:parameters}. For (b), the cutoff range is reduced to $\Delta x=3$.
    (d) Relative precision of MPO compression as a function of the post-compression bond dimension $D$. The precision is controlled by $\epsilon(D) = (\sum_{a=D+1} s_a^2)^{1/2}$, as described in Eq.~\eqref{eq:compressed_H}. Here $\Delta E(D) = E(D=1000) - E(D)$ is the energy error in the ground state, $\mathcal{F}$ is the fidelity per unit cell between the ground state at $D$ and the ground state at $D=1000$, and $\Delta \overline{\gamma^z}$ is the error in the polarization versus $D=1000$, described in Sec. \ref{sec:half_filling_physics}. One can see that the precision improves roughly in proportion to $\epsilon$, except for $\n{1-\mathcal{F}}^{1/2}$, which is limited by the precision of DMRG (black dashed line).
    }
    \label{fig:MPOfidelity}
\end{figure*}

In this section we consider the practical details of performing infinite DMRG on the IBM model defined in the last section, and the necessity of MPO compression.

To perform infinite DMRG, we must express the Hamiltonian (Eq.~\eqref{eq:IBM}) as an infinite 1D Matrix Product Operator (MPO) whose size $D$ is called the bond dimension\footnote{The MPO bond dimension is always denoted by $\chiMPO$, and $\chiMPS$ is reserved for the MPS bond dimension.} \cite{pirvu2010matrix}. To map from 2D to a 1D chain, we order the Wannier orbitals $ \ket{w(\pm, n, k_y)} $ by the positions $P_x(\pm, k_y) + n$ of their Wannier centers. Translation along $\v{L}_1$ simply increments $\ket{w(\pm,n,k_y)} \to \ket{w(\pm, n+1, k_y)}$, so the 1D chain is periodic with a unit cell of size $N_B N_y$ sites ($N_B = 8$ with spin and valley). Once the MPO is obtained, we can in principle find its ground state with DMRG.

However, the long-range nature of the Coulomb interaction complicates matters. Although the screened Coulomb interaction decays exponentially in real space, truncating it at short range can lead to physically incorrect results. To demonstrate this,  Fig.~\ref{fig:MPOfidelity} (a) examines the energy of two ground state wavefunction ans\"atze ``QAH" and ``\SM{}"  as a function of the truncation distance of the interaction $\Delta x$ \footnote{We define $\Delta x$ as the distance between the first and last field operators along the cylinder} (these physical states are defined and used in Sec.~\ref{sec:SM_phase} below). In particular, we examine the energy difference $\Delta E = E_{QAH} - E_{SM}$ in Panel (a), and the relative energy difference in Panel (b). The true energy difference between these states is $|\Delta E(\Delta x \to \infty)| \approx \SI{0.1}{\milli\electronvolt}$, yet the energy difference achieves \SI{0.1}{\milli\electronvolt} precision only at $\Delta x \ge 3$. Going from $\Delta x = 1 \to 2$, for example, their energies change by almost \SI{0.1}{\milli\electronvolt}.

More importantly, the relative error in energy difference $(\Delta E(\Delta x) - \Delta E(\infty)) / \Delta E(\infty)$ reaches $1\%$ only at the cutoff $\Delta x \ge 4$. This means that in order to resolve closely competing ground state candidates --- which we will encounter in practice in Section \ref{sec:cylinder_model} below --- we require a relatively large cutoff $\Delta x$.

Furthermore, the required cutoff is highly dependent on the model parameters. For example, if we increase the gate distance to \SI{20}{\nano\meter}, then the screening distance is increased, and the relative energy gap does not achieve \emph{$50 \%$} precision until $\Delta x \approx 3$ (Fig. \ref{fig:MPOfidelity} (b)).
\footnote{Note, however, that larger \textit{precompression} MPO bond dimension does not necessarily mean larger \textit{postcompression} MPO bond dimension. The relationship between Hamiltonian parameters and the postcompression bond dimension is a subject of future work.}
Together, these results suggest that premature truncation may lead to physically incorrect results, and we are forced to retain relatively long-range interactions in the Hamiltonian.

After mapping to a 1D chain, this means we must keep track of interactions up to range $R = N_B N_y \Delta x$ orbitals. An exact representation has an optimal bond dimension which scales as $D = O(R^2)$ (App.~\ref{app:uncompressed_MPO_construction}).
However, this still produces an MPO of size $D \sim 1 \times 10^4$ for $\Delta x \le 4$ without spin and valley, and if we were to add in spin and valley it would be $D \sim 1 \times 10^5$. 
As the computational complexity of DMRG increases as $O(D^2)$, and $D$ is usually a few hundred at most, the Hamiltonian for BLG is far too large for DMRG to be practical. 
The DMRG results of Ref.~\cite{KangVafek} considered a single spin and valley with  interactions  truncated at $\Delta x = 2$, resulting in an MPO of $D \sim 2000$ at $L_y = 6$.
But increasing $\Delta x$, or adding  spin and valley, makes the problem impractical.

On a finite system, the MPO can be viewed as a 2-sided MPS and compressed by SVD truncation (this approach is implemented in the \texttt{AutoMPO} feature of the \texttt{iTensor} library \cite{fishman2020ITensor}). However, in the infinite limit we wish to take here, this naive SVD truncation  is unstable and actually destroys the locality of the Hamiltonian. To avoid this, Ref.~\cite{parker2019local}  developed a modification of SVD compression which guarantees that the compressed Hamiltonian remains Hermitian and local in the thermodynamic limit. 

As in finite SVD compression, an intermediate step of the algorithm produces a singular value spectrum $s_a$,  and the bond dimension can be reduced by discarding the lowest values of the spectrum. For an appropriate notion of distance this truncation is optimal, and when applied to a single cut, the discarded weight  $\epsilon(\chiMPO) = \sqrt{ \sum_{i=\chiMPO+1}^{D'} s_a^2 }$  upper bounds the error in $\widehat{H}$ with respect to the Frobenius norm. In the Appendix \ref{app:UCMPO_compression} we present efficient algorithms for finite-length unit cells and derive error bounds for various quantities. When exploiting quantum numbers, the algorithm is capable of compressing MPOs with bond dimensions $5 \times 10^4$ or larger on a cluster node.

With the bond dimension thus reduced to a reasonable value, we may perform DMRG. We use the standard \texttt{TeNPy} library \cite{hauschild2018efficient}, written by one of us, taking full advantage of symmetries. Careful checks  guaranteeing the accuracy and precision of our code,  benchmarks,  and other numerical details are given in App.~\ref{app:numerical_details}. 

Figure \ref{fig:MPOfidelity} showcases the precision of our DMRG results. We performed DMRG at the chiral limit $w_0=0$ and computed the relative error in the ground state energy, ground state fidelity, and expectation values as a function of post-compression bond dimension $D$, relative to $D=1000$. The relative precision $\epsilon(D)/\epsilon(0)$ improves quickly with $D$, dipping below $10^{-6}$ by $D = 800$. In accordance with the error bound on $H$, the ground state energy, wavefunction, and expectation values converge quickly as $\epsilon \to 0$.

As a proof of principle, we also performed MPO compression for the IBM model with spin and valley at $L_y = 6$ and $w_0 = 0$. Due to constraints on the size of the uncompressed MPO we can handle, we chose a cutoff range of $\Delta x = 3$, which resulted in a $D \sim 35,000$ uncompressed MPO.  The singular value spectrum of the MPO is shown in Fig.~\ref{fig:MPOfidelity}. If we define $\mathcal{F}$ as the fidelity per unit cell\footnote{We define the fidelity per unit cell in the thermodynamic limit as $\mathcal{F} = \lim_{N\rightarrow \infty}|\langle\psi_{N,D}|\psi_{N,D=1000}\rangle|^{2/N}$ , where $N$ is the number of unit cells.} between the ground state of the compressed MPO with bond dimension $D$ and the ground state of the MPO with $D=1000$, then we see from Fig. \ref{fig:MPOfidelity} that in the spinless, single-valley calculation, $\epsilon(D)/\epsilon(0)$ and $|1-\mathcal{F}|^{1/2}$ have roughly the same order of magnitude. Using this fact as a guide, we can estimate the bond dimension where $|1-\mathcal{F}|^{1/2} \sim 10^{-2}$ by  looking at the value of $\epsilon(D)/\epsilon(0)$. This gives us bond dimensions $D = \{106, 317, 1057\}$ for spinless/single-valley, spinless/valley, spin/valley MPO. While still relatively large, such bond dimensions are tractable with a standard workstation or cluster node when exploiting quantum numbers.

Of course the IBM model itself is only an approximation to the physical system, neglecting  effects such as lattice relaxation, phonons and twist angle disorder which, though small, are expected to enter at the \SI{1}{\milli\electronvolt} level. This provides a limit on the amount of precision which is physically useful. To be safe, we choose $\Delta x = 10$,
\footnote{This gives us an uncompressed bond dimension of order $5 \times 10^{4}$, close to what would be necessary for spinful/valleyful calculation.}
$\epsilon = 10^{-2}$ \si{\milli\electronvolt}, which results in post-compression bond dimensions of $D \approx 600-1000$, depending on the value of $w_0/w_1$. In conclusion, we have used MPO compression to reduce the Hamiltonian to a computational tractable size, incurring a precision error on the order of $10^{-2}$ \si{\milli\electronvolt} --- three orders of magnitude below the relevant energy scale of the problem. We now discuss the results of DMRG and the implications for the ground state physics of bilayer graphene.

\section{Ground State Physics at Half Filling}
\label{sec:half_filling_physics}

In this section we report the results of our DMRG calculations and discuss the ground state physics of the (spinless, single valley) IBM model at half filling. We will show there is a clear transition from a quantum anomalous Hall state at small $w_0/w_1$ to a nematic semimetallic state at large $w_0/w_1$. Furthermore, we will show that these ground states are almost exactly described by the $k$-space Slater determinants predicted by Hartree-Fock.

\begin{figure}
    \centering
	\includegraphics[width=\linewidth]{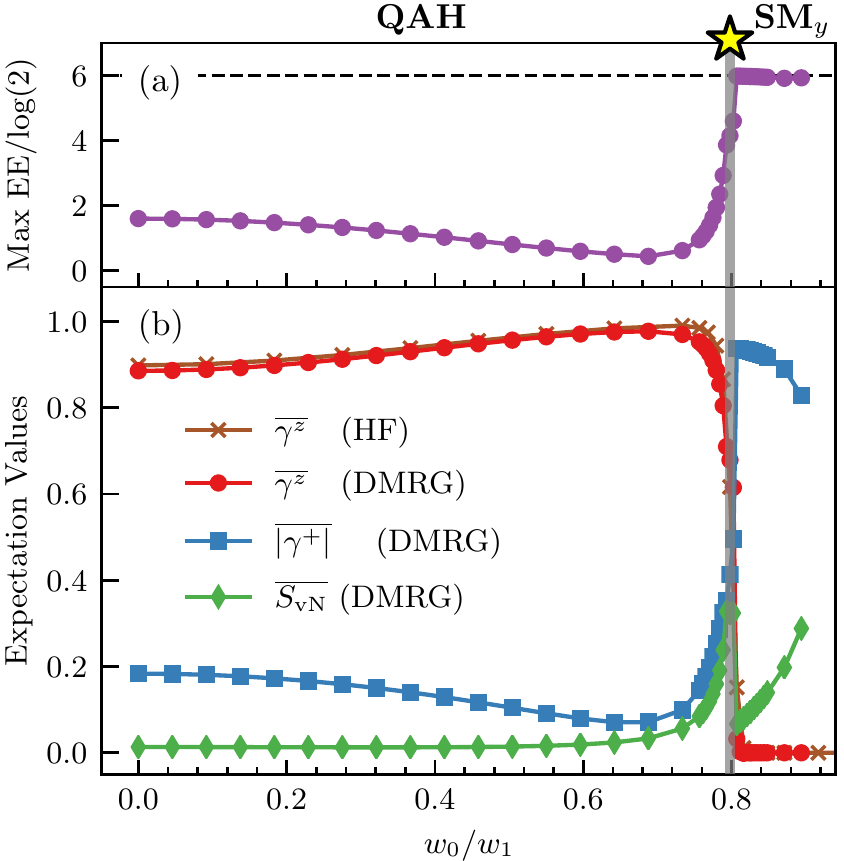}
    \caption{The phase diagram of the IBM at half filling as a function of $w_0/w_1$. There is a transition from a quantum anomalous hall (QAH) phase to a semimetallic (SM) phase at $w_0/w_1 = 0.798 \approx 0.8$, represented by a yellow star. (a) Entanglement entropy of the DMRG ground state with hybrid Wannier orbitals ordered according to their polarization, maximized over all $12$ entanglement cuts dividing the system into left and right halves. (b) Expectation values of various observables (defined in the text) in the DMRG or Hartree-Fock(HF) ground states. The polarization in Chern band space $\overline{\gamma^z}$ is an order parameter for the transition. Its drop across the transition is accompanied by a commensurate increase in $\overline{|\gamma^+|}$, such that the DMRG ground state remains close to a Slater determinant. DMRG is performed at bond dimension $\chiMPS = 1024$, $\epsilon_{\mathrm{MPO}} = 10^{-2}$ \si{\milli\electronvolt} and is convergent away from the transition. (Gray shading indicates where DMRG is not well converged.).
    }
    \label{fig:half_filling_w0_scan}
\end{figure}

\subsection{Single particle projector and order parameter}

We start by defining several crucial observables and order parameters. Because we find the DMRG ground state is translation invariant, all one-body expectation values can be obtained from the correlation matrix
\begin{equation}
    P(\v{k}) :=
    \begin{pmatrix}
\langle c^\dagger_{+, k_x, k_y}c_{+, k_x, k_y} \rangle & \langle c^\dagger_{-, k_x, k_y}c_{+, k_x, k_y} \rangle   \\ 
\langle c^\dagger_{+, k_x, k_y}c_{-, k_x, k_y} \rangle & \langle c^\dagger_{-, k_x, k_y}c_{-, k_x, k_y} \rangle    \\ 
\end{pmatrix}.
\label{eq:single_particle_density_matrix}
\end{equation}
 This matrix is a projector when the expectation values are taken with respect to a Slater determinant, and it is the central variational object for $k$-space Hartree-Fock calculations. For DMRG in mixed-$xk$ space, we calculate $P(\v{k})$ by Fourier transforming two-point correlation functions.
 \footnote{Explicitly, $P(\v{k})$ is defined for $\v{k}$ on a $108 \times L_y$ grid of $\v{k}$ points in the mBZ by computing expectations $\langle c^\dagger_{\pm, 0, k_y}c_{\pm, n, k_y} \rangle$ with respect to the DMRG ground state on the mixed-$xk$ space cylinder for $-53 \le n \le 54$ and performing a discrete Fourier transform with respect to $\mathbf{L}_1$.}
 
The one-body observables are  spanned by the expectation values of Pauli matrices $\sigma$ in the $\pm$ band space,
\begin{equation}
    \gamma^{z}(\v{k}) := \tr[P(\v{k}) \sigma^z],
    \label{eq:_gamma_observable_band_space}
\end{equation}
and similarly for $\gamma^x, \gamma^y, \gamma^+ = \gamma^x + i  \gamma^y$.  We denote mBZ averages by \begin{equation}
    \overline{\gamma^{\alpha}}~:=~\frac{1}{A_{\textrm{mBZ}}}\int_{\mathrm{mBZ}} d^2 \v{k} \; \gamma^{\alpha}(\v{k}).
    \label{eq:amma_observable_BZ_average}
\end{equation}
where $A_{\textrm{mBZ}}$ is the area of the mBZ. We will focus particularly on $\overline{\gamma^z}$ --- which is an order parameter for $C_2 \mathcal{T}$ and $C_{2x}$, as follows from Eq.~\eqref{eq:symmetry_action_cylinder} which implies that $\gamma^z(\v{k}) = 0$ for a $C_2 \mathcal{T}$ symmetric state, and $\gamma^z(C_{2x}\v{k}) = - \gamma^z(\v{k})$ for a $C_{2x}$ symmetric state.

In the case where the state is indeed a momentum-diagonal Slater determinant, $P(\v{k})$ acquires several special properties. In particular, if a momentum mode $\v{k}$ is occupied by one electron, $P(\v{k})$ takes values on the unit sphere and can be parametrized in spherical coordinates as 
\begin{equation}
    P(\v{k}) = \frac{1}{2}(\sigma^0 + \cos \theta_{\v{k}} \sigma^z + \sin\theta_{\v{k}}\cos\varphi_{\v{k}}\sigma^x + \sin\theta_{\v{k}}\sin\varphi_{\v{k}}\sigma^y) \label{eq:projector_parametrization}
\end{equation}
which implies $|\gamma^+|^2 + |\gamma^z|^2 = 1$. If the projector respects $C_2\mathcal{T}$ and $C_{2x}$ symmetries, then respectively $\theta_{\v{k}} = \pi/2$ and $\varphi_{\v{k}} = -\varphi_{C_{2x}\v{k}} + \pi$ at all $\v{k}$. Finally, since $P(\v{k})$ is a projector for momentum-diagonal Slater determinants, it satisfies $S_{\text{vN}}(\v{k}) :=  -\Tr[P(\v{k})\log P(\v{k})] = 0$. In general, then, $S_{\text{vN}}(\v{k}) \ge 0$ measures the deviation of a state from a translationally-invariant Slater determinant.

\subsection{iDMRG details and parameter choices}

Infinite DMRG (iDMRG) calculations were performed using the open source \texttt{TeNPy} package \cite{hauschild2018efficient}.  The numerical parameters and physical energy scales of the problem are summarized in Table~\ref{tab:parameters}. In particular, we take $N_y = 6$ and $\Phi_y = \pi$ as the ``default" values. The MPO bond dimension was compressed down to $600-1000$, such that the expected error is of order $10^{-2}$ \si{\milli\electronvolt}, as described in Sec.~\ref{sec:cylinder_model}. To ensure that iDMRG was converged, we varied the MPS bond dimension $\chiMPS$ between $200$ and $1024$. We found that DMRG converged well even at very low bond dimensions, except near the transition. We also allowed ground states with broken translation invariance with a doubled unit cell, but we found a fully translation invariant ground state for all parameters we tested. 

\begin{table}[]
    \begin{center}
    \begin{ruledtabular}
    \begin{tabular}{ll}
        Parameter & Value(s) \\[0.2em] \colrule
        $\theta_{BM}$ & $\sim{}$\ang{1.05} \\
        $w_1$ &  $\sim{}$\SI{109}{\milli\electronvolt} \\
        $w_0/w_1$ & [0, 1] \\
        Gate distance & \SI{10}{\nano\meter} \\
        Relative permitivity & 12 \\[0.3em] \colrule
        $N_y$ & 6 \\
        $\Phi_y$ & $\pi, \pi/10$\\
        $\chiMPS$ & $\leq 1024$ \\
        $\Delta x$ & 10 \\
        $\epsilon_{\mathrm{MPO}}$ & $< 10^{-2}$ \si{\milli\electronvolt}\\[0.3em] \colrule
        Kinetic energy scale ($t$) & $<$ \SI{1}{\milli\electronvolt}\\
        Interaction energy scale ($V$) & $<$ \SI{10}{\milli\electronvolt}
    \end{tabular}
    \end{ruledtabular}
    \end{center}
    \caption{Parameters of the IBM model, DMRG calculation, and relevant energy scales. See main text for the definition of each entry.}
    \label{tab:parameters}
\end{table}

\subsection{Ground State Transition and the QAH phase}

\begin{figure*}
    \centering
    \includegraphics[width=\textwidth]{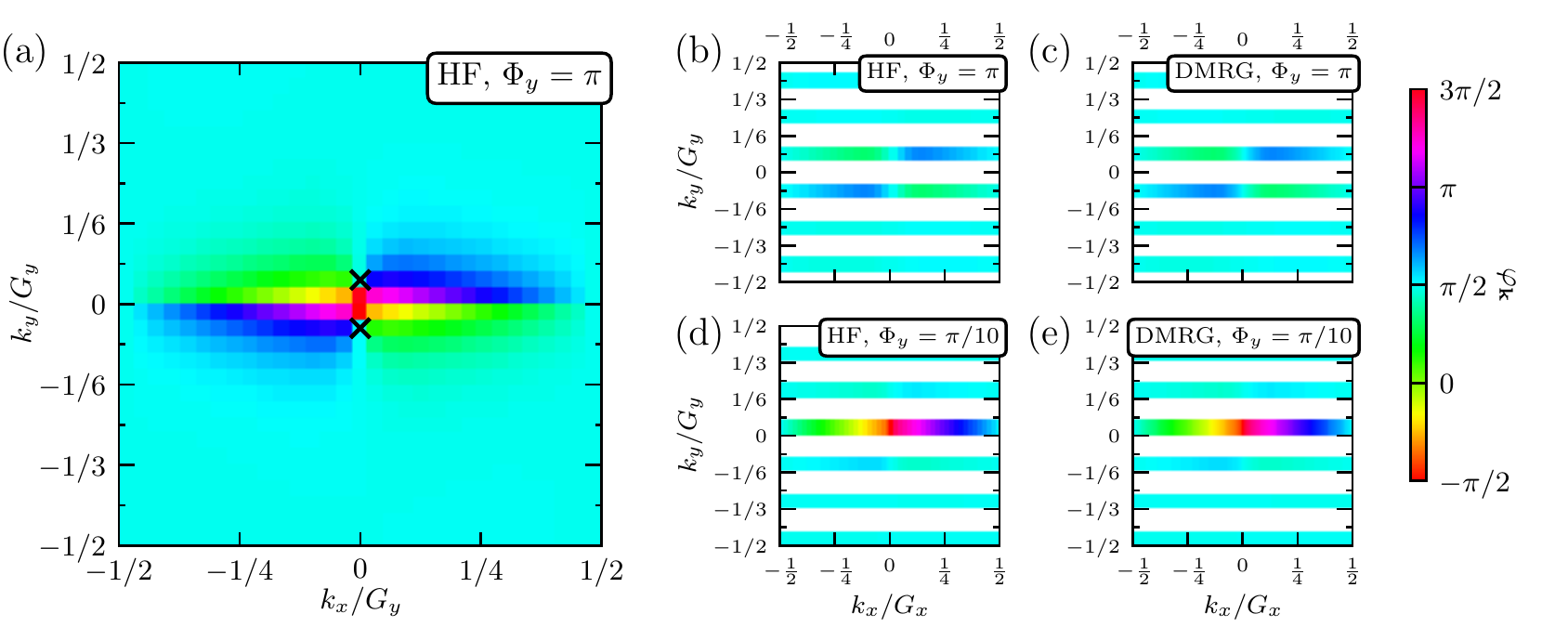}
    \caption{ Comparison of HF and DMRG calculations of $\varphi_{\v{k}} = \operatorname{arg}[\gamma^{+}(\v{k})] \in [-\pi,\pi]$ over the mBZ at $w_0/w_1 = 0.85$.  (a) HF on a $30 \times 29$ grid with $\Phi_y=\pi$. The crosses represent the approximate location of Dirac points. (b) HF on a $30 \times 6$ grid at $\Phi_y = \pi$. (c) DMRG correlation function on a $108 \times 6$ grid at $\Phi_y=\pi$. (d) HF on a $30\times 6$ grid at $\Phi_y=\pi/10$. (e) DMRG correlation function on a $108\times6$ grid at $\Phi_y=\pi/10$. The horizontal bands in (b) -- (d) are centered on the $k_y$ cuts used, given by Eq.~\eqref{eq:ky_cuts}. One can see that the DMRG and HF calculations are virtually identical. DMRG are performed using $\chiMPS=1024$. The discontinuity at $k_y = 0^+$ is a gauge choice, described in \ref{sec:cylinder_model}.
    }
    \label{fig:DMRG_HF_comparison}
\end{figure*}

We performed iDMRG at $44$ values of $w_0/w_1$ in the range $[0,1]$. Fig.~\ref{fig:half_filling_w0_scan} (b) shows that the order parameter $\overline{\gamma^z}$ is non-zero for $w_0/w_1 \le 0.8$, and vanishes  for larger values of $w_0/w_1$, signaling a transition from a $C_2\mathcal{T}$ and $C_{2x}$ broken phase to a $C_{2}\mathcal{T}$ and $C_{2x}$ symmetric phase.

For low $w_0/w_1$, not only is $C_2 \mathcal{T}$ broken, but the state is almost perfectly polarized, with $\overline{\gamma^z} \approx 1$. This implies that the state has a large overlap with the product state in which all ``$+$'' orbitals are occupied:
\begin{equation}
\ket{\mathrm{QAH}} \approx    \prod_{x,k_y} \hat{c}_{+,x,k_y}^\dagger \ket{0}.
\label{eq:approx_QAH_wavefunction}
\end{equation}
Since the $\pm$ bands carry Chern number $C = \pm 1$, this state is a quantum anomalous hall (QAH) insulator \cite{XieMacDonald,bultinck2020Mechanism,YaHuiChern}. This approximation is quite good: the QAH state is well described by a product state plus small corrections, $\n{\braket{\Psi_0|\Psi_{QAH}}} \approx 0.846$ per unit cell at $w_0 = 0$. Consequently, the QAH state has low entanglement entropy (Fig.~\ref{fig:half_filling_w0_scan}) and DMRG converges at quite moderate bond dimensions. 

Above $w_0/w_1  \approx 0.8$, $\overline{\gamma^z} = 0$ and the state instead develops a large expectation value for $\gamma^{+} = \gamma^x + i \gamma^y$. Section \ref{sec:SM_phase} below is devoted to the large $w_0/w_1$ phase, and we will see that it is a  nematic semimetal \cite{liu2020Nematic}, which we refer to as ``\SM{}'', in reference to the ordering in the $x/y$ plane. First, however, we analyze a surprising structure in the ground state correlations. 

\subsection{The remarkable accuracy of Hartree-Fock}

The ground states of the strongly interacting IBM model are -- quite surprisingly -- very well described by $k$-space Slater determinants. For all values of $w_0/w_1$ away from the transition, the difference between the ground state and a Slater determinant as quantified by $S_{\text{vN}}(\v{k})$ is small. In particular, Fig.~\ref{fig:half_filling_w0_scan} shows that $\overline{S}_{\text{vN}}$ is low in the QAH phase, increases or diverges near the transition, and is relatively small but growing in the \SM{} phase. In the QAH phase this behavior is expected due to the large overlap with the simple Chern band polarized Slater determinant Eq.~\eqref{eq:approx_QAH_wavefunction}.

To provide further evidence that the ground state is essentially a Slater determinant, we compare DMRG results with Hartree-Fock (HF) calculations. Hartree-Fock determines an optimal Slater determinant approximation to the ground state of a many-body problem through a self consistent equation.
Computationally, HF scales only polynomially in the number of $k_y$ cuts (rather than exponentially for DMRG), so it provides a much cheaper alternative --- when it is applicable.
When the  ground state of the IBM model is close to a Slater determinant, the HF ground state should be quite accurate and would have high overlap with the true ground state. We performed HF calculations on a $N_x \times N_y$ grid in the mBZ; numerical details of our HF calculations have been reported elsewhere \cite{bultinck2019Ground}. 

We find that HF and DMRG results are nearly identical. The $C_2\mathcal{T}$ order parameter $\gamma^z$ differs by around 2\% (Fig.~\ref{fig:half_filling_w0_scan}).
Fig.~\ref{fig:DMRG_HF_comparison} shows a side-by-side comparison of the DMRG and HF predictions for $\varphi_{\v{k}} = \operatorname{arg}[\gamma^+(\v{k})] $ in the \SM{} phase, where it completely specifies the Slater determinant because $\theta \equiv \pi/2$ is fixed by $C_2\mathcal{T}$ symmetry (See Eq~\eqref{eq:projector_parametrization}). Panel (a) shows a high-resolution HF calculation with $N_x = 30, N_y = 29$, which shows that $\varphi_{\v{k}} \approx \pi/2$ over most of the mBZ, but winds through $2\pi$ for $k_y$ cuts that go near the $\Gamma$ point. Panels (b) -- (d) demonstrate that the same pattern appears with only $N_y = 6$ discrete momentum cuts. Both HF and DMRG produce a $C_{2x}$ symmetric $\varphi_{\v{k}}$ and the results obtained from both methods are almost indistinguishable. Other observables are similarly accurate in HF. We may therefore use HF to study large system sizes or observables that are not easily accessible in DMRG. For instance, Koopman's Theorem implies that the energies of single-fermion excited states are given by the self-consistent Hartree-Fock spectrum. 
We found the Hartree-Fock spectrum at the chiral limit has a gap of order \SI{20}{\milli\electronvolt}, showing the QAH state is gapped.  %This immediately implies that the QAH phase has a gap to electronic excitations.
We conclude that DMRG and HF agree to a remarkable degree and may be used almost interchangeably in this regime.

\section{The Nematic Semimetal}
\label{sec:SM_phase}

 \begin{figure*}
    \centering
    \includegraphics[width=\linewidth]{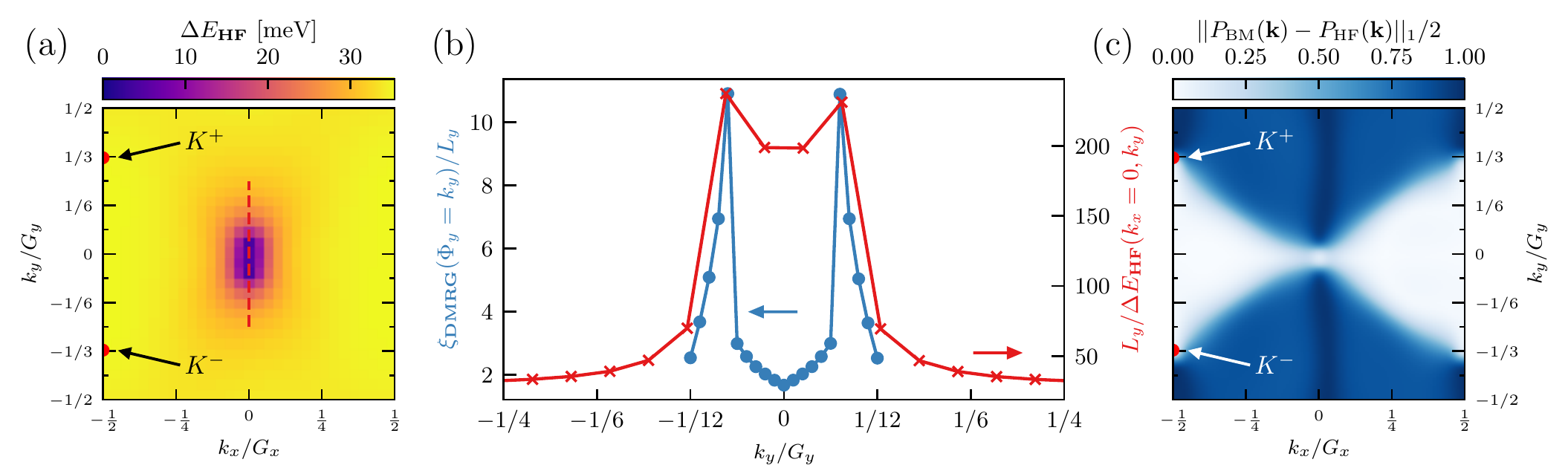}
	\caption{Result of HF calculation at $N_x \times N_y = 30 \times 29$ .  (a) Energy gap between the top band and the bottom band. We see there are two Dirac points near the gamma point. (b) Correlation length obtained from thin cylinder DMRG at different flux values and the inverse HF gap size as a function of $k_y$. We see the correlation length becomes large
	near the Dirac points found in HF calculation.
	(c) The trace distance between the density matrices for the BM ground state and Hartree-Fock ground state. One can see that they are unrelated over much of the mBZ.
	}
    \label{fig:HF_gap}
\end{figure*} 

We now show that the large-$w_0/w_1$ phase is a nematic semimetal, first described in Ref.~\cite{liu2020Nematic}, with energetics  governed by the Berry curvature of the flat bands. This is an altogether different state than the Dirac semimetal which appears in the non-interacting BM model. Our analysis is based on combination of DMRG
(at $L_y = 6$) and HF (at $L_y \sim 30$), which agree wherever they can be compared. After establishing the nature of the nematic semimetal, we make contact with recent ideas in the literature \cite{bultinck2020Mechanism, liu2020Nematic,KangVafek}. Namely, we explain how the Ginzburg-Landau-like functional for the interband coherence $\varphi$ proposed in Ref.~\cite{liu2020Nematic}   provides an intuitive description of the nematic state and the transition, and also confirm that the stripe state proposed in Ref.~\cite{KangVafek} is extremely competitive, with an energy only $0.2$ meV / electron  above the DMRG ground state.

\subsection{The Large $w_0/w_1$ Phase is Nematic}

The large $w_0/w_1$ phase is a nematic state which breaks $C_3$ but preserves $C_2\mathcal{T}$. 
$C_2 \mathcal{T}$ requires $\gamma^z = 0$, but allows for finite $\gamma^{x/y}$, so within the spherical coordinate description of Eq.~\eqref{eq:projector_parametrization} the state is characterized by $\theta = \pi / 2$ and an azimuthal angle $\varphi_{\mathbf{k}}$.
This is clearly visible in Fig.~\ref{fig:DMRG_HF_comparison} (a): throughout \textit{most} of the mBZ, including at the mini-$K^{\pm}$ points, the state is  $\gamma^y \approx 1$  ($\varphi \sim \pi/2$). 
A state with finite $\gamma^x + i \gamma^y$ at the $K^{\pm}$ points breaks $C_3$ symmetry (nematicity) because $C_3$ acts as $C_3 = e^{i \pi \sigma^z/3}$ there.
This is in contrast to the BM ground state, which has  Dirac nodes at $K^{\pm}$: the BM Dirac structure $h(K^{\pm} + \v{q}) \propto q_x \sigma_x + q_y \sigma_y$ instead causes $\gamma^{x} + i \gamma^y$ to wind by $+2\pi$. 

Consequently we denote the large $w_0/w_1$ state ``$\mathrm{SM}_y$.'' Presumably the other $C_3$-rotated versions are not found  because the cylinder geometry weakly breaks the $C_3$ symmetry for finite $N_y$.

While the \SM{} phase is close to a Slater determinant, it is not a small perturbation to the non-interacting ground state. To quantify this, Fig.~\ref{fig:HF_gap}(c) shows the trace distance between the BM ground state projector and the \SM{} projector over the mBZ.
Around the $K^{\pm}$ points, the trace distance rotates between complete agreement and orthogonality, consistent with the winding of $\gamma^+$ in the BM state versus the fixed $\gamma^y \approx 1$ in $\textrm{SM}_y$. The trace distance also provides us with a gauge invariant way to identify the nematicity of the phase: the BM and $\textrm{SM}_y$ projectors achieve near complete agreement along the $y$-axis.

\subsection{$C_2\mathcal{T}$ protected Dirac points}
\label{sec:c2t_dirac}

Near the $\Gamma$ point, however $\varphi_{\mathbf{k}}$ deviates from $\varphi = \pi / 2$. We now show that this is because the \SM{} phase features two Dirac points in the vicinity of $\Gamma$, which cause $\varphi_{\mathbf{k}}$ to wind there. 
This behaviour is in fact enforced by topological properties \cite{Po, bultinck2020Mechanism}. For a generic two band problem in the presence of $C_2\mathcal{T}$, any Wilson loop is quantized: $W(C) := i \int_C \v{A} \cdot d\mathbf{k} = n\pi$ with $n\in \mathbb{Z}$ \cite{kim2015Dirac}.\footnote{The Berry phase is computed for the filled band only.}
In particular, $W(C) = (2n+1)\pi$ if and only if it encloses a Dirac cone.

This is the well-known topological protection of Dirac cones. In the case of the single valley BM model, Dirac cones at the mini $K^\pm$ points have the \textit{same} chirality $W(C) = +\pi$.
Therefore, not only are the Dirac cones locally protected, but even if they move away from $K^{\pm}$ they cannot meet and annihilate, enforcing the existence of either a pair of Dirac points or quadratic band touching.\footnote{Note however that this ``global'' protection implicitly assumes translation symmetry: if the unit cell doubles, the bands fold and the band count doubles. Beyond two bands, $W(C)$ is only defined modulo $\pi$, so Dirac points can meet and annihilate. This mechanism underlies the 
$C_2 \mathcal{T}$ stripe phase\cite{KangVafek}.}

The semimetallic nature of \SM{} is borne out in both HF and DMRG numerics.
The spectrum of the self-consistent HF Hamiltonian, shown in Fig.~\ref{fig:HF_gap} (a), has a large ($\Delta \approx $ \SI{30}{\milli\electronvolt}) gap across most of the mBZ, except near the $\Gamma$ point.
The structure of $\Delta$ near the $\Gamma$ point is consistent with two Dirac points at $\v{k}_{\pm} := (k_x, k_y) \approx (0, \pm0.05 G_y)$ (Fig.~\ref{fig:HF_gap}(b)).
In contrast, the BM model is gapless at the $K^{\pm}$ points, but gapped near $\Gamma$. 

DMRG numerics can also detect these nodes, but special care is required. This is because the allowed momentum cuts, Eq.~\eqref{eq:ky_cuts}, generically avoid $\v{k}_{\pm}$. 
To confirm their existence, we continuously adjust the flux $\Phi_y$ through the cylinder (see Fig.~\ref{fig:IBM_model} (e)) and monitor the behavior of the DMRG ground state as the allowed momenta pass through the putative Dirac points. 
Fig.~\ref{fig:HF_gap}(b) shows that the DMRG correlation length appears to diverge right as the allowed momenta pass through the location of the Dirac points $\mathbf{k}^\pm$ found in HF, consistent with the gap closing. 

We conclude that the large $w_0/w_1$ phase is a nematic semimetal: \SM{} preserves $C_2 \mathcal{T}$, breaks $C_3$, and has two Dirac nodes on the $y$ axis near $\Gamma$.

\subsection{Ginzburg-Landau-like Description of the \SM{} Phase}
There is a very appealing Hartree-Fock picture for why the Coulomb interactions reconstruct the single particle Dirac semimetal into the nematic \SM{} semimetal.
When $C_2 \mathcal{T}$ is preserved, the state is specified entirely by the phase of the inter-band coherence  $\varphi_\mathbf{k}$, so the HF energy is a functional $E_{\mathrm{HF}}[\varphi_\mathbf{k}]$.
Ref.~\cite{liu2020Nematic} analytically computed this functional for the IBM model, Eq.~\eqref{eq:IBM}, and found that the dominant  contribution takes the form
     \begin{align}
      E_{\mathrm{ HF}}[\varphi_\mathbf{k}] &= E^{\mathrm{QAH}}_{\mathrm{ HF}} + \frac{1}{2} \int g_{\mathbf{k}} ( \nabla_{\mathbf{k}} \varphi_{\mathbf{k}}- 2 \v{a}_{\v{k}})^2  d^2\v{k} + \cdots  \notag \\
       g_{\mathbf{k}} &= \frac{1}{A} \sum_{\mathbf{q}} V_{\mathbf{q}}  \, \mathbf{q}^2 |\Lambda_{\mathbf{q}}( \mathbf{k} )|^2 
       \label{eq:sm_ehf}
      \end{align} 
Here $g_{\mathbf{k}}$ is a $E_C$-scale function independent of $\varphi_{\v{k}}$, and $E^{\mathrm{QAH}}_{\mathrm{ HF}}$ is the Coulomb energy of the QAH state.
Finally, $\v{a}_{\v{k}}$ is a U(1) vector potential which encodes the band geometry: due to the $C_2 \mathcal{T} = \sigma^x K$ symmetry,  the SU(2) Berry connection of the Bloch states is constrained to take the diagonal form $\mathbf{A}(\mathbf{k}) = \sigma^z \v{a}_{\v{k}}$,  reducing it to a $U(1)$ connection.
\footnote{Note that $\v{a}_{\v{k}}$ here does not have a quantized Wilson loop, unlike $\v{A}_{\v{k}}$ discussed in Sec.~\ref{sec:c2t_dirac}.
This is because $\v{a}_{\v{k}}$ is defined in terms of Chern bands, which are not invariant under $C_2\mathcal{T}$.}

We see that the energy is similar to the Ginzburg-Landau functional for a superconductor in a magnetic field $ F_{\mathbf{k}} = \mathrm{d} \v{a}_{\v{k}}$.
This isn't a coincidence: $\nabla_{\mathbf{k}} \varphi_{\mathbf{k}}$ can only appear via a gauge-covariant derivative because $\varphi_{\mathbf{k}}, \v{a}_{\v{k}}$ transform as a gauge pair under a $C_2 \mathcal{T}$-preserving phase redefinition of the Bloch states, $\hat{c}_{\pm, \mathbf{k}} \to e^{ \pm i \phi_\mathbf{k}} \hat{c}_{\pm, \mathbf{k}}$.
However, there is no exact U(1) symmetry, so the small ``$\cdots$'' terms we neglect (for example the dispersion $h$) do couple directly to $\varphi$.

The superconducting analogy can be made more concrete by applying a particle-hole transformation to only the $C=-1$ band,  so that the coherence $\varphi$ between $C= \pm 1$ bands maps to ``superconducting pairing'' between two $C=1$ bands \cite{bultinck2020Mechanism}. The Berry curvature $F_{\mathbf{k}}$ then appears with the same form as a magnetic field, albeit in $k$-space, similar to how Berry curvature manifests as a ``$k$-space magnetic field'' in the semiclassical equations of motion for Bloch electrons \cite{PhysRevB.59.14915}. 

If we treat the mBZ as the unit cell, the Chern number $\frac{1}{2 \pi} \int F_{\mathbf{k}} \, d^2 \v{k} = 1$ implies that there is one flux quantum per unit cell.
Just like the vortex lattice of a superconductor in a magnetic field $F_{\mathbf{k}}$, this forces $\varphi$ to have two vortices per unit cell.
Each vortex ($+2 \pi$ winding) is  equivalent to a Dirac point, so this recovers the topological protection of the Dirac points discussed earlier.  
In the BM ground state, the two vortices are pinned to the $K^{\pm}$ points, while in the \SM{} state they lie near $\Gamma$ (Fig.~\ref{fig:DMRG_HF_comparison} (a)).

The vortices lead to an energy penalty relative to the QAH state, explaining Eq.~\eqref{eq:sm_ehf}.
However, in our case the Berry curvature $F_{\mathbf{k}}$ is not uniform: instead, $F_{\mathbf{k}}$ is concentrated near the $\Gamma$ point (Fig.~\ref{fig:berry_curvature}(a)).
By  analogy to a superconductor in a non-uniform field, the lowest energy configuration of Eq.~\eqref{eq:sm_ehf} will place the vortices in the region of concentrated $F_{\mathbf{k}}$, explaining their shift from $K^{\pm} \to \Gamma$.
In Fig.~\ref{fig:berry_curvature}(b), we confirm that $ \nabla_{\mathbf{k}} \varphi_{\mathbf{k}}- 2 \v{a}_{\v{k}} \approx 0$ in the region where $F_{\mathbf{k}}$ is small, but is finite near $\Gamma$ where $F_{\mathbf{k}}$ is concentrated.
Accordingly, most of the energy penalty comes from near the $\Gamma$ point. Increasing $w_0 / w_1$ makes $F_{\mathbf{k}}$ increasingly concentrated, reducing the Coulomb penalty of \SM{} relative to QAH. 

The final ingredient driving the finite $w_0/w_1$ transition are the small terms like the dispersion $h$ hidden in ``$\cdots$'', which  slightly prefer the \SM{} phase.\footnote{ For example, in the \SM{} phase, $\varphi_{\mathbf{k}}$ can perturbatively deform to follow the dispersion $h$, particularly in the vicinity of $\Gamma$.} 
As $w_0/w_1$ increases, the Coulomb penalty for the \SM{} phase decreases due to  the Berry curvature concentration, and these subleading  terms win out. 
Consequently, while increasing  $w_0/w_1$ does slightly increase the bandwidth, its primary effect is actually via the redistribution of Berry curvature, which enters at the Coulomb scale ($g_{\mathbf{k}} \sim E_C$).

 \begin{figure}
 \begin{center}
     \includegraphics[width=\linewidth]{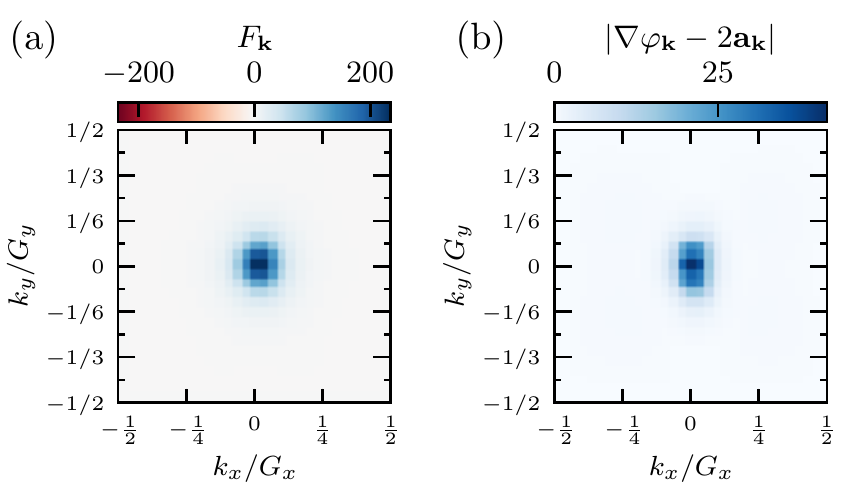}
 \end{center}
	\caption{(a) The Berry curvature $F_{\v{k}}$ for $w_0/w_1 = 0.85$ shows concentration at the $\Gamma$-point. 
	(b) The $k$-space ``supercurrent'' $|\nabla \varphi_{\v{k}} - 2\v{a}_{\v{k}}|$ is concentrated in the same region as $F_{\v{k}}$.
	Both (a) and (b) are calculated in the units $G_x = G_y = 1$, $A_{BZ} = 1$.
	}
    \label{fig:berry_curvature}
\end{figure} 

\subsection{``Thin cylinder''  DMRG Analysis}\label{sec:thin_cylinder}

The DMRG ground state at $N_y = 6$ can be approximated by a particularly simple ``thin cylinder'' ansatz, provided none of the momentum cuts cross the region of large Berry curvature near $\Gamma$. We ensure this by taking $\Phi_y = \pi$ and focus on $w_0/w_1 \approx 0.85$. In this case, the DMRG ground state has a relatively small particle number fluctuation in the unit cell. This is because away from $\Gamma$, $P(\v{k})$
varies slowly with $k_x$, so in the xk-space the correlations are local (intra-unit cell).

Table~\ref{tab:occupation_number} shows the probability for each momentum mode in the unit cell to be occupied by $0$, $1$, or $2$ electrons. All modes have $p(N_{k_y} = 1)$ larger than 0.9, and many of them larger than 0.95. This suggests there is a simple ``thin-cylinder'' ansatz with no particle number fluctuation per momentum mode: following Ref.~\cite{KangVafek}, we define
\begin{equation}
    \begin{aligned}
    \ket{\Psi_{\textrm{TC}}} &= \prod_{x, k} \left(\cos\frac{\theta(x, k_y)}{2} \hat{c}^\dagger_{+, x, k_y} \right.\\
    &\hspace{3em} +\left.\sin\frac{\theta(x, k_y)}{2} e^{i\varphi(x, k_y)}\hat{c}^\dagger_{-, x, k_y}\right) \ket{0}.
    \end{aligned}
    \label{eq:SP_ansatz}
\end{equation}
It is easy to check that if $\theta$ and $\varphi$ are independent of $x$, these $\theta$ and $\varphi$ corresponds to those in Eq.~\eqref{eq:projector_parametrization}. By inspecting Fig.~\ref{fig:DMRG_HF_comparison} (b), we find $\theta = \pi/2$ and $\varphi = \pi/2$ is a good approximation for the \SM{} state\footnote{Taking our different gauge convention into account, this is in agreement with \cite{KangVafek}.}
so long as the momentum cuts are not close to $k_y = 0$.
The fidelity per unit cell of the resulting ansatz $\ket{\Psi_{\textrm{SM}}}$ and the DMRG ground state is $\mathcal{F} \sim{} 0.98$ --- remarkably accurate for such a simple ansatz. This ansatz also captures well the behavior of entanglement entropy in Fig.~\ref{fig:half_filling_w0_scan}: each momentum mode contributes a ``Bell pair'' to entanglement, making the maximum entanglement $6\log 2$ with our choice of orbital ordering.

\begin{table}
    \begin{center}
    \begin{ruledtabular}
    \begin{tabular}{r*{6}{c}}
    $k_y$           & $-\frac{5}{12}$ & $-\frac{1}{4}$ & $-\frac{1}{12}$ & $\frac{1}{12}$ & $\frac{1}{4}$ & $\frac{5}{12}$ \\[0.3em] \colrule
    $p(N_{k_y} = 0)$ & 0.021 & 0.021 &0.049& 0.049& 0.021& 0.021   \\
    $p(N_{k_y} = 1)$ & 0.958 & 0.959 & 0.902& 0.902& 0.959& 0.958  \\
    $p(N_{k_y} = 2)$ & 0.021 & 0.021 & 0.049 & 0.049 & 0.021& 0.021 \\[0.2em]
    \end{tabular}
    \end{ruledtabular}
    \end{center}
    \caption{Probability that some states are occupied by 0, 1, or 2 electrons in a given unit cell at $w_0/w_1 = 0.85$, $\Phi_y = \pi$, and $\chiMPS = 1024$. $N_{k_y}$ is the number of electrons with momentum $k_y$ in the unit cell.
    }
    \label{tab:occupation_number}
\end{table}

We can further motivate the thin-cylinder ansatz from the polarization $P_x(\pm, k_y)$. (Fig.~\ref{fig:w0_polarization}). At large $w_0/w_1$, $P_x$ is close to 0 for most $k_y$ values, and orbitals in one unit cell are well separated from those in neighboring unit cells. The interaction therefore strongly couples modes within the same unit cell, resulting in vanishing number fluctuation per unit cell.

We also see why the thin-cylinder ansatz breaks down near the $\Gamma$ point: $P(\v{k})$ changes rapidly there (Fig.~\ref{fig:DMRG_HF_comparison} (e)), leading to inter-unit cell correlations.
We stress that the ansatz is thus a crude approximation to the true HF ground state, since it fails to capture the complex Berry curvature contribution to the energy that is dominant near the $\Gamma$ point.
As $N_y \to \infty$, the momentum cuts unavoidably approach $\Gamma$, and the thin-cylinder ansatz will break down.

Finally, we comment on the energy competition between different candidate ground states. The simple form of the ansatz enabled Kang and Vafek \cite{KangVafek} to put forward another candidate for the ground state, which breaks translation symmetry in the $\v{L}_1$ direction in favor of a period-2 stripe state with screw symmetry $C_{2x}T_{L_1}$. To crude approximation,  this state corresponds to $\theta(n, k_y) = \pi/2$ and $\varphi(n, k_y) = (-1)^n \pi/2 + \pi/2$ in the parametrization of Eq.~\eqref{eq:SP_ansatz}. 
In order to test this ansatz, we computed the energy of the SM ansatz and  the stripe ansatz and compared it to the DMRG ground state energy (Table~\ref{tab:energy_comparison}) and the QAH ansatz ($\theta = 0$). We also computed the  energy of the ground state of the BM model with respect to the IBM model to establish the relevant Coulomb energy scale. We find that the energy of the QAH, \SM{}, and the stripe ansatz are within $1/3$ \si{\milli\electronvolt} of the DMRG energy.
While the stripe is not the true ground state for the parameter values studied here,\footnote{We verify this by initializing the DMRG using the doubled unit cell stripe ansatz, and find the stripe reverts to the \SM{} phase at convergence.}
 this confirms the assertion in Ref.~\cite{KangVafek} that the stripe is a viable candidate in the wider phase diagram.

\begin{table}
    \begin{center}
    \begin{ruledtabular}
    \begin{tabular}{lr}
    State & Energy [\si{\milli\electronvolt}]\\[0.2em]\colrule
    DMRG Ground State (\SM{}) & $-28.24$\\
    QAH Ansatz (Eq.~\ref{eq:SP_ansatz}) & $-28.04$\\
    \SM{} Ansatz (Eq.~\ref{eq:SP_ansatz}) & $-27.92$ \\
    $C_2 \mathcal{T}$ - Stripe Ansatz (Eq.~\ref{eq:SP_ansatz}) & $-28.08$ \\
    Dirac (BM Ground State) & $-20.62$\\
    \end{tabular}
    \end{ruledtabular}
    \end{center}
    \caption{Energy per electron of various trial states, evaluated with respect to the IBM Hamiltonian at $w0/w_1 = 0.85$.  Here ``Dirac" refers to ground state of the single-particle BM Hamiltonian, and the parameters for SM and Stripe Ans\"atze are described in the text. We see that the Stripe is very close to the DMRG ground state, and the Dirac state is well-separated from the rest, reflecting the dominant importance of the Coulomb interactions. Note that the energies are negative because we have subtracted off the $q=0$ part of the Coulomb interaction.}

    \label{tab:energy_comparison}
\end{table}

\section{Conclusions}\label{sec:conclusion}

In this work we have introduced a method for studying \tBLG{} (and any other moir\'e material) using DMRG.
Our method (Fig.~\ref{fig:flowchart}) starts with the BM model, adds interactions, and compresses the resulting MPO down to a reasonable size so that the DMRG is computationally tractable. We carefully verified the correctness of our approach and showed that it is sufficiently precise to capture the ground state physics of the IBM model. To benchmark our approach we focused mostly on the spinless, single-valley case. However, we showed in Section \ref{sec:MPOs} that our method can be extended to the spinful, two-valley case with only moderately greater computational resources. Therefore we have identified a method to study the ground state physics of \tBLG{} using DMRG.

Even though the spinless, single-valley model is not strictly physical, our results  have several important conceptual implications for the study of \tBLG{}. Remarkably, we found that the DMRG ground state is well-approximated as a $k$-space Slater determinant for all $w_0/w_1$. However, we stress that the ground state nevertheless has \textit{no} relation to the ground state of the single-particle BM model: it mixes states very far from the Fermi surface of the BM Hamiltonian (Fig.~\ref{fig:HF_gap} (c)). At least near the magic angle this suggests that weak-coupling approaches to \tBLG{}, which rely on various details of the BM Fermi surface, will miss the essential physics. Instead, the energetics are dominated by the exchange physics of the Coulomb interaction. As a result, the effective band structure (as would be computed from the self-consistent Hartree-Fock Hamiltonian) is entirely different than that of the BM model, with a width set by $E_C$ (Fig.~\ref{fig:HF_gap}(a)). This is true even in the large $w_0 / w_1$ nematic semimetal phase \SM{}, which may have some relation to the  $\nu=0$ semimetallic resistance peak found in experiment.

Furthermore, it is subtle to describe the observed ground states within a 2D Wannier-localized ``Mott insulating'' picture.
For  small $w_0 / w_1$ we have a QAH phase: the filled states have net Chern number, so the projector onto the filled states cannot be 2D Wannier localized. This is not to say that it is impossible to find these states within a  numerical approach which starts from 2D Wannier orbitals (which is just a change of basis), but rather that in such a basis the order would manifest as a set of coherences $\langle c^\dagger_i c_j \rangle$ between sites rather than an onsite order parameter. Presumably this would complicate any mean-field approach which depends on a site-local self energy.

Taken together, this supports the point of view that \tBLG{} is  more closely related to quantum Hall ferromagnetism, where symmetry breaking is driven by the combination of band topology and Coulomb exchange, than it is to  the Mott insulating physics of the Hubbard model. 
 But of course in contrast to quantum Hall systems, \tBLG{} comes with time-reversal symmetry, making it amenable to superconductivity. 
Future work will explore the physics of \tBLG{} upon restoring the spin and valley degrees of freedom.
\\ \\
\begin{acknowledgments}
    We thank Xiangyu Cao and Eslam Khalaf for guiding our understanding of MPO compression and the nematic semimetal, as well as discussions with Shubhayu Chatterjee, Frank Pollmann, Senthil Todadri and Ashvin Vishwanath.
    We also thank Jian Kang and Oskar Vafek for carefully explaining their DMRG results to us.

    DEP  acknowledges  support  from  the  NSF  Graduate  Research Fellowship Program Grant No. NSF DGE 1752814.
    MPZ was supported by the Director, Office of Science, Office of Basic Energy Sciences, Materials Sciences and Engineering Division of the U.S. Department of Energy under contract no. DE-AC02-05-CH11231 (van der Waals heterostructures program, KCWF16). 
    JH was funded by the U.S. Department of Energy, Office of Science, Office of Basic Energy Sciences, Materials Sciences and Engineering Division under Contract No. DE-AC02-05- CH11231 through the Scientific Discovery through Advanced Computing (SciDAC) program (KC23DAC Topological and Correlated Matter via Tensor Networks and Quantum Monte Carlo).
    This research used the Savio computational cluster resource provided by the Berkeley Research Computing program at the University of California, Berkeley (supported by the UC Berkeley Chancellor, Vice Chancellor for Research, and Chief Information Officer).
\end{acknowledgments}

\bibliographystyle{unsrtnat}
\bibliography{references}

\begin{thebibliography}{70}
\providecommand{\natexlab}[1]{#1}
\providecommand{\url}[1]{\texttt{#1}}
\expandafter\ifx\csname urlstyle\endcsname\relax
  \providecommand{\doi}[1]{doi: #1}\else
  \providecommand{\doi}{doi: \begingroup \urlstyle{rm}\Url}\fi

\bibitem[Cao et~al.(2018{\natexlab{a}})Cao, Fatemi, Demir, Fang, Tomarken, Luo,
  Sanchez-Yamagishi, Watanabe, Taniguchi, Kaxiras, Ashoori, and
  Jarillo-Herrero]{Cao}
Yuan Cao, Valla Fatemi, Ahmet Demir, Shiang Fang, Spencer~L. Tomarken, Jason~Y.
  Luo, Javier~D. Sanchez-Yamagishi, Kenji Watanabe, Takashi Taniguchi,
  Efthimios Kaxiras, Ray~C. Ashoori, and Pablo Jarillo-Herrero.
\newblock Correlated insulator behaviour at half-filling in magic-angle
  graphene superlattices.
\newblock \emph{Nature}, 556:\penalty0 80 EP --, 03 2018{\natexlab{a}}.
\newblock URL \url{https://doi.org/10.1038/nature26154}.

\bibitem[Cao et~al.(2018{\natexlab{b}})Cao, Fatemi, Fang, Watanabe, Taniguchi,
  Kaxiras, and Jarillo-Herrero]{Cao2}
Yuan Cao, Valla Fatemi, Shiang Fang, Kenji Watanabe, Takashi Taniguchi,
  Efthimios Kaxiras, and Pablo Jarillo-Herrero.
\newblock Unconventional superconductivity in magic-angle graphene
  superlattices.
\newblock \emph{Nature}, 556:\penalty0 43 EP --, 03 2018{\natexlab{b}}.
\newblock URL \url{https://doi.org/10.1038/nature26160}.

\bibitem[Yankowitz et~al.(2019)Yankowitz, Chen, Polshyn, Zhang, Watanabe,
  Taniguchi, Graf, Young, and Dean]{Yankowitz}
Matthew Yankowitz, Shaowen Chen, Hryhoriy Polshyn, Yuxuan Zhang, K.~Watanabe,
  T.~Taniguchi, David Graf, Andrea~F. Young, and Cory~R. Dean.
\newblock Tuning superconductivity in twisted bilayer graphene.
\newblock \emph{Science}, 363\penalty0 (6431):\penalty0 1059--1064, 2019.
\newblock ISSN 0036-8075.
\newblock \doi{10.1126/science.aav1910}.
\newblock URL \url{https://science.sciencemag.org/content/363/6431/1059}.

\bibitem[Kerelsky et~al.(2019)Kerelsky, McGilly, Kennes, Xian, Yankowitz, Chen,
  Watanabe, Taniguchi, Hone, Dean, Rubio, and Pasupathy]{Kerelsky}
Alexander Kerelsky, Leo~J. McGilly, Dante~M. Kennes, Lede Xian, Matthew
  Yankowitz, Shaowen Chen, K.~Watanabe, T.~Taniguchi, James Hone, Cory Dean,
  Angel Rubio, and Abhay~N. Pasupathy.
\newblock Maximized electron interactions at the magic angle in twisted bilayer
  graphene.
\newblock \emph{Nature}, 572\penalty0 (7767):\penalty0 95--100, August 2019.
\newblock ISSN 1476-4687.
\newblock \doi{10.1038/s41586-019-1431-9}.

\bibitem[Jiang et~al.(2019)Jiang, Lai, Watanabe, Taniguchi, Haule, Mao, and
  Andrei]{RutgersSTM}
Yuhang Jiang, Xinyuan Lai, Kenji Watanabe, Takashi Taniguchi, Kristjan Haule,
  Jinhai Mao, and Eva~Y. Andrei.
\newblock Charge order and broken rotational symmetry in magic-angle twisted
  bilayer graphene.
\newblock \emph{Nature}, 573\penalty0 (7772):\penalty0 91--95, 2019.
\newblock ISSN 1476-4687.
\newblock \doi{10.1038/s41586-019-1460-4}.
\newblock URL \url{https://doi.org/10.1038/s41586-019-1460-4}.

\bibitem[Lu et~al.(2019)Lu, Stepanov, Yang, Xie, Aamir, Das, Urgell, Watanabe,
  Taniguchi, Zhang, et~al.]{efetov}
Xiaobo Lu, Petr Stepanov, Wei Yang, Ming Xie, Mohammed~Ali Aamir, Ipsita Das,
  Carles Urgell, Kenji Watanabe, Takashi Taniguchi, Guangyu Zhang, et~al.
\newblock Superconductors, orbital magnets and correlated states in magic-angle
  bilayer graphene.
\newblock \emph{Nature}, 574\penalty0 (7780):\penalty0 653--657, 2019.

\bibitem[{Xie} et~al.(2019){Xie}, {Lian}, {J{\"a}ck}, {Liu}, {Chiu},
  {Watanabe}, {Taniguchi}, {Bernevig}, and {Yazdani}]{YazdaniSpectroscopic}
Yonglong {Xie}, Biao {Lian}, Berthold {J{\"a}ck}, Xiaomeng {Liu}, Cheng-Li
  {Chiu}, Kenji {Watanabe}, Takashi {Taniguchi}, B.~Andrei {Bernevig}, and Ali
  {Yazdani}.
\newblock {Spectroscopic signatures of many-body correlations in magic-angle
  twisted bilayer graphene}.
\newblock \emph{\nat}, 572\penalty0 (7767):\penalty0 101--105, July 2019.
\newblock \doi{10.1038/s41586-019-1422-x}.

\bibitem[Stepanov et~al.(2019)Stepanov, Das, Lu, Fahimniya, Watanabe,
  Taniguchi, Koppens, Lischner, Levitov, and Efetov]{Efetovscreening}
Petr Stepanov, Ipsita Das, Xiaobo Lu, Ali Fahimniya, Kenji Watanabe, Takashi
  Taniguchi, Frank~HL Koppens, Johannes Lischner, Leonid Levitov, and Dmitri~K
  Efetov.
\newblock The interplay of insulating and superconducting orders in magic-angle
  graphene bilayers.
\newblock \emph{arXiv preprint arXiv:1911.09198}, 2019.

\bibitem[Saito et~al.(2020)Saito, Ge, Watanabe, Taniguchi, and
  Young]{YoungScreening}
Yu~Saito, Jingyuan Ge, Kenji Watanabe, Takashi Taniguchi, and Andrea~F. Young.
\newblock Independent superconductors and correlated insulators in twisted
  bilayer graphene.
\newblock \emph{Nature Physics}, pages 1--5, June 2020.
\newblock ISSN 1745-2481.
\newblock \doi{10.1038/s41567-020-0928-3}.

\bibitem[{Choi} et~al.(2019){Choi}, {Kemmer}, {Peng}, {Thomson}, {Arora},
  {Polski}, {Zhang}, {Ren}, {Alicea}, {Refael}, {von Oppen}, {Watanabe},
  {Taniguchi}, and {Nadj-Perge}]{Choi}
Youngjoon {Choi}, Jeannette {Kemmer}, Yang {Peng}, Alex {Thomson}, Harpreet
  {Arora}, Robert {Polski}, Yiran {Zhang}, Hechen {Ren}, Jason {Alicea}, Gil
  {Refael}, Felix {von Oppen}, Kenji {Watanabe}, Takashi {Taniguchi}, and
  Stevan {Nadj-Perge}.
\newblock {Electronic correlations in twisted bilayer graphene near the magic
  angle}.
\newblock \emph{Nature Physics}, 15\penalty0 (11):\penalty0 1174--1180, August
  2019.
\newblock \doi{10.1038/s41567-019-0606-5}.

\bibitem[Yoo et~al.(2019)Yoo, Engelke, Carr, Fang, Zhang, Cazeaux, Sung,
  Hovden, Tsen, Taniguchi, Watanabe, Yi, Kim, Luskin, Tadmor, Kaxiras, and
  Kim]{yoo2019Atomic}
Hyobin Yoo, Rebecca Engelke, Stephen Carr, Shiang Fang, Kuan Zhang, Paul
  Cazeaux, Suk~Hyun Sung, Robert Hovden, Adam~W. Tsen, Takashi Taniguchi, Kenji
  Watanabe, Gyu-Chul Yi, Miyoung Kim, Mitchell Luskin, Ellad~B. Tadmor,
  Efthimios Kaxiras, and Philip Kim.
\newblock Atomic and electronic reconstruction at the van der {{Waals}}
  interface in twisted bilayer graphene.
\newblock \emph{Nature Materials}, 18\penalty0 (5):\penalty0 448--453, May
  2019.
\newblock ISSN 1476-4660.
\newblock \doi{10.1038/s41563-019-0346-z}.

\bibitem[{Sharpe} et~al.(2019){Sharpe}, {Fox}, {Barnard}, {Finney}, {Watanabe},
  {Taniguchi}, {Kastner}, and {Goldhaber-Gordon}]{Sharpe}
Aaron~L. {Sharpe}, Eli~J. {Fox}, Arthur~W. {Barnard}, Joe {Finney}, Kenji
  {Watanabe}, Takashi {Taniguchi}, M.~A. {Kastner}, and David
  {Goldhaber-Gordon}.
\newblock {Emergent ferromagnetism near three-quarters filling in twisted
  bilayer graphene}.
\newblock \emph{Science}, 365\penalty0 (6453):\penalty0 605--608, August 2019.
\newblock \doi{10.1126/science.aaw3780}.

\bibitem[Serlin et~al.(2020)Serlin, Tschirhart, Polshyn, Zhang, Zhu, Watanabe,
  Taniguchi, Balents, and Young]{YoungAH}
M.~Serlin, C.~L. Tschirhart, H.~Polshyn, Y.~Zhang, J.~Zhu, K.~Watanabe,
  T.~Taniguchi, L.~Balents, and A.~F. Young.
\newblock Intrinsic quantized anomalous hall effect in a moir{\'e}
  heterostructure.
\newblock \emph{Science}, 367\penalty0 (6480):\penalty0 900--903, 2020.
\newblock ISSN 0036-8075.
\newblock \doi{10.1126/science.aay5533}.
\newblock URL \url{https://science.sciencemag.org/content/367/6480/900}.

\bibitem[Tomarken et~al.(2019)Tomarken, Cao, Demir, Watanabe, Taniguchi,
  Jarillo-Herrero, and Ashoori]{Tomarken}
S.~L. Tomarken, Y.~Cao, A.~Demir, K.~Watanabe, T.~Taniguchi,
  P.~Jarillo-Herrero, and R.~C. Ashoori.
\newblock Electronic compressibility of magic-angle graphene superlattices.
\newblock \emph{Phys. Rev. Lett.}, 123:\penalty0 046601, Jul 2019.
\newblock \doi{10.1103/PhysRevLett.123.046601}.
\newblock URL \url{https://link.aps.org/doi/10.1103/PhysRevLett.123.046601}.

\bibitem[{Wong} et~al.(2020){Wong}, {Nuckolls}, {Oh}, {Lian}, {Xie}, {Jeon},
  {Watanabe}, {Taniguchi}, {Bernevig}, and {Yazdani}]{CascadeYazdani}
Dillon {Wong}, Kevin~P. {Nuckolls}, Myungchul {Oh}, Biao {Lian}, Yonglong
  {Xie}, Sangjun {Jeon}, Kenji {Watanabe}, Takashi {Taniguchi}, B.~Andrei
  {Bernevig}, and Ali {Yazdani}.
\newblock {Cascade of electronic transitions in magic-angle twisted bilayer
  graphene}.
\newblock \emph{\nat}, 582\penalty0 (7811):\penalty0 198--202, June 2020.
\newblock \doi{10.1038/s41586-020-2339-0}.

\bibitem[Zondiner et~al.(2020)Zondiner, Rozen, {Rodan-Legrain}, Cao, Queiroz,
  Taniguchi, Watanabe, Oreg, {von Oppen}, Stern, Berg, {Jarillo-Herrero}, and
  Ilani]{CascadeShahal}
U.~Zondiner, A.~Rozen, D.~{Rodan-Legrain}, Y.~Cao, R.~Queiroz, T.~Taniguchi,
  K.~Watanabe, Y.~Oreg, F.~{von Oppen}, Ady Stern, E.~Berg,
  P.~{Jarillo-Herrero}, and S.~Ilani.
\newblock Cascade of phase transitions and {{Dirac}} revivals in magic-angle
  graphene.
\newblock \emph{Nature}, 582\penalty0 (7811):\penalty0 203--208, June 2020.
\newblock ISSN 1476-4687.
\newblock \doi{10.1038/s41586-020-2373-y}.

\bibitem[Arora et~al.(2020)Arora, Polski, Zhang, Thomson, Choi, Kim, Lin,
  Wilson, Xu, Chu, Watanabe, Taniguchi, Alicea, and {Nadj-Perge}]{Harpreet}
Harpreet~Singh Arora, Robert Polski, Yiran Zhang, Alex Thomson, Youngjoon Choi,
  Hyunjin Kim, Zhong Lin, Ilham~Zaky Wilson, Xiaodong Xu, Jiun-Haw Chu, Kenji
  Watanabe, Takashi Taniguchi, Jason Alicea, and Stevan {Nadj-Perge}.
\newblock Superconductivity in metallic twisted bilayer graphene stabilized by
  {{WSe}} 2.
\newblock \emph{Nature}, 583\penalty0 (7816):\penalty0 379--384, July 2020.
\newblock ISSN 1476-4687.
\newblock \doi{10.1038/s41586-020-2473-8}.

\bibitem[{Nuckolls} et~al.(2020){Nuckolls}, {Oh}, {Wong}, {Lian}, {Watanabe},
  {Taniguchi}, {Bernevig}, and {Yazdani}]{YazdaniChern}
Kevin~P. {Nuckolls}, Myungchul {Oh}, Dillon {Wong}, Biao {Lian}, Kenji
  {Watanabe}, Takashi {Taniguchi}, B.~Andrei {Bernevig}, and Ali {Yazdani}.
\newblock {Strongly Correlated Chern Insulators in Magic-Angle Twisted Bilayer
  Graphene}.
\newblock \emph{arXiv e-prints}, art. arXiv:2007.03810, July 2020.

\bibitem[{Wu} et~al.(2020){Wu}, {Zhang}, {Watanabe}, {Taniguchi}, and
  {Andrei}]{AndreiChern}
Shuang {Wu}, Zhenyuan {Zhang}, K.~{Watanabe}, T.~{Taniguchi}, and Eva~Y.
  {Andrei}.
\newblock {Chern Insulators and Topological Flat-bands in Magic-angle Twisted
  Bilayer Graphene}.
\newblock \emph{arXiv e-prints}, art. arXiv:2007.03735, July 2020.

\bibitem[{Tschirhart} et~al.(2020){Tschirhart}, {Serlin}, {Polshyn}, {Shragai},
  {Xia}, {Zhu}, {Zhang}, {Watanabe}, {Taniguchi}, {Huber}, and
  {Young}]{Andreananosquid}
C.~L. {Tschirhart}, M.~{Serlin}, H.~{Polshyn}, A.~{Shragai}, Z.~{Xia},
  J.~{Zhu}, Y.~{Zhang}, K.~{Watanabe}, T.~{Taniguchi}, M.~E. {Huber}, and A.~F.
  {Young}.
\newblock {Imaging orbital ferromagnetism in a moir{\'e} Chern insulator}.
\newblock \emph{arXiv e-prints}, art. arXiv:2006.08053, June 2020.

\bibitem[{Lu} et~al.(2020){Lu}, {Lian}, {Chaudhary}, {Piot}, {Romagnoli},
  {Watanabe}, {Taniguchi}, {Poggio}, {MacDonald}, {Bernevig}, and
  {Efetov}]{EfetovFragile}
Xiaobo {Lu}, Biao {Lian}, Gaurav {Chaudhary}, Benjamin~A. {Piot}, Giulio
  {Romagnoli}, Kenji {Watanabe}, Takashi {Taniguchi}, Martino {Poggio},
  Allan~H. {MacDonald}, B.~Andrei {Bernevig}, and Dmitri~K. {Efetov}.
\newblock {Fingerprints of Fragile Topology in the Hofstadter spectrum of
  Twisted Bilayer Graphene Close to the Second Magic Angle}.
\newblock \emph{arXiv e-prints}, art. arXiv:2006.13963, June 2020.

\bibitem[{Liu} et~al.(2020){Liu}, {Wang}, {Watanabe}, {Taniguchi}, {Vafek}, and
  {Li}]{VafekLi}
Xiaoxue {Liu}, Zhi {Wang}, K.~{Watanabe}, T.~{Taniguchi}, Oskar {Vafek}, and
  J.~I.~A. {Li}.
\newblock {Tuning electron correlation in magic-angle twisted bilayer graphene
  using Coulomb screening}.
\newblock \emph{arXiv e-prints}, art. arXiv:2003.11072, March 2020.

\bibitem[{Cao} et~al.(2020){Cao}, {Rodan-Legrain}, {Park}, {Noah Yuan},
  {Watanabe}, {Taniguchi}, {Fernandes}, {Fu}, and
  {Jarillo-Herrero}]{PabloNematicity}
Yuan {Cao}, Daniel {Rodan-Legrain}, Jeong~Min {Park}, Fanqi {Noah Yuan}, Kenji
  {Watanabe}, Takashi {Taniguchi}, Rafael~M. {Fernandes}, Liang {Fu}, and Pablo
  {Jarillo-Herrero}.
\newblock {Nematicity and Competing Orders in Superconducting Magic-Angle
  Graphene}.
\newblock \emph{arXiv e-prints}, art. arXiv:2004.04148, April 2020.

\bibitem[White(1992)]{white1992density}
Steven~R White.
\newblock Density matrix formulation for quantum renormalization groups.
\newblock \emph{Physical review letters}, 69\penalty0 (19):\penalty0 2863,
  1992.

\bibitem[Kang and Vafek(2020)]{KangVafek}
Jian Kang and Oskar Vafek.
\newblock Non-abelian dirac node braiding and near-degeneracy of correlated
  phases at odd integer filling in magic-angle twisted bilayer graphene.
\newblock \emph{Phys. Rev. B}, 102:\penalty0 035161, Jul 2020.
\newblock \doi{10.1103/PhysRevB.102.035161}.
\newblock URL \url{https://link.aps.org/doi/10.1103/PhysRevB.102.035161}.

\bibitem[Bistritzer and MacDonald(2011)]{bistritzer2011Moire}
Rafi Bistritzer and Allan~H. MacDonald.
\newblock Moir\'e bands in twisted double-layer graphene.
\newblock \emph{Proceedings of the National Academy of Sciences}, 108\penalty0
  (30):\penalty0 12233--12237, July 2011.
\newblock ISSN 0027-8424, 1091-6490.
\newblock \doi{10.1073/pnas.1108174108}.

\bibitem[Nam and Koshino(2017)]{NamKoshino}
Nguyen N.~T. Nam and Mikito Koshino.
\newblock Lattice relaxation and energy band modulation in twisted bilayer
  graphene.
\newblock \emph{Phys. Rev. B}, 96:\penalty0 075311, Aug 2017.
\newblock \doi{10.1103/PhysRevB.96.075311}.
\newblock URL \url{https://link.aps.org/doi/10.1103/PhysRevB.96.075311}.

\bibitem[Carr et~al.(2019)Carr, Fang, Zhu, and Kaxiras]{CarrKaxiras}
Stephen Carr, Shiang Fang, Ziyan Zhu, and Efthimios Kaxiras.
\newblock Exact continuum model for low-energy electronic states of twisted
  bilayer graphene.
\newblock \emph{Phys. Rev. Research}, 1:\penalty0 013001, Aug 2019.
\newblock \doi{10.1103/PhysRevResearch.1.013001}.
\newblock URL \url{https://link.aps.org/doi/10.1103/PhysRevResearch.1.013001}.

\bibitem[Song et~al.(2019)Song, Wang, Shi, Li, Fang, and Bernevig]{Bernevig}
Zhida Song, Zhijun Wang, Wujun Shi, Gang Li, Chen Fang, and B.~Andrei Bernevig.
\newblock All magic angles in twisted bilayer graphene are topological.
\newblock \emph{Phys. Rev. Lett.}, 123:\penalty0 036401, Jul 2019.
\newblock \doi{10.1103/PhysRevLett.123.036401}.
\newblock URL \url{https://link.aps.org/doi/10.1103/PhysRevLett.123.036401}.

\bibitem[Bultinck et~al.(2020{\natexlab{a}})Bultinck, Khalaf, Liu, Chatterjee,
  Vishwanath, and Zaletel]{bultinck2019Ground}
Nick Bultinck, Eslam Khalaf, Shang Liu, Shubhayu Chatterjee, Ashvin Vishwanath,
  and Michael~P. Zaletel.
\newblock Ground state and hidden symmetry of magic-angle graphene at even
  integer filling.
\newblock \emph{Phys. Rev. X}, 10:\penalty0 031034, Aug 2020{\natexlab{a}}.
\newblock \doi{10.1103/PhysRevX.10.031034}.
\newblock URL \url{https://link.aps.org/doi/10.1103/PhysRevX.10.031034}.

\bibitem[Xie and MacDonald(2020)]{XieMacDonald}
Ming Xie and A.~H. MacDonald.
\newblock Nature of the correlated insulator states in twisted bilayer
  graphene.
\newblock \emph{Phys. Rev. Lett.}, 124:\penalty0 097601, Mar 2020.
\newblock \doi{10.1103/PhysRevLett.124.097601}.
\newblock URL \url{https://link.aps.org/doi/10.1103/PhysRevLett.124.097601}.

\bibitem[Repellin et~al.(2020)Repellin, Dong, Zhang, and
  Senthil]{repellin2020Ferromagnetism}
C{\'e}cile Repellin, Zhihuan Dong, Ya-Hui Zhang, and T.~Senthil.
\newblock Ferromagnetism in narrow bands of moir\textbackslash 'e
  superlattices.
\newblock \emph{Physical Review Letters}, 124\penalty0 (18):\penalty0 187601,
  May 2020.
\newblock ISSN 0031-9007, 1079-7114.
\newblock \doi{10.1103/PhysRevLett.124.187601}.

\bibitem[Po et~al.(2018)Po, Zou, Vishwanath, and Senthil]{Po}
Hoi~Chun Po, Liujun Zou, Ashvin Vishwanath, and T.~Senthil.
\newblock Origin of mott insulating behavior and superconductivity in twisted
  bilayer graphene.
\newblock \emph{Phys. Rev. X}, 8:\penalty0 031089, Sep 2018.
\newblock \doi{10.1103/PhysRevX.8.031089}.
\newblock URL \url{https://link.aps.org/doi/10.1103/PhysRevX.8.031089}.

\bibitem[Po et~al.(2019)Po, Zou, Senthil, and Vishwanath]{Po2}
Hoi~Chun Po, Liujun Zou, T.~Senthil, and Ashvin Vishwanath.
\newblock Faithful tight-binding models and fragile topology of magic-angle
  bilayer graphene.
\newblock \emph{Phys. Rev. B}, 99:\penalty0 195455, May 2019.
\newblock \doi{10.1103/PhysRevB.99.195455}.
\newblock URL \url{https://link.aps.org/doi/10.1103/PhysRevB.99.195455}.

\bibitem[Zou et~al.(2018)Zou, Po, Vishwanath, and Senthil]{ZouPo}
Liujun Zou, Hoi~Chun Po, Ashvin Vishwanath, and T.~Senthil.
\newblock Band structure of twisted bilayer graphene: Emergent symmetries,
  commensurate approximants, and wannier obstructions.
\newblock \emph{Phys. Rev. B}, 98:\penalty0 085435, Aug 2018.
\newblock \doi{10.1103/PhysRevB.98.085435}.
\newblock URL \url{https://link.aps.org/doi/10.1103/PhysRevB.98.085435}.

\bibitem[Hejazi et~al.(2019)Hejazi, Liu, Shapourian, Chen, and Balents]{Hejazi}
Kasra Hejazi, Chunxiao Liu, Hassan Shapourian, Xiao Chen, and Leon Balents.
\newblock Multiple topological transitions in twisted bilayer graphene near the
  first magic angle.
\newblock \emph{Phys. Rev. B}, 99:\penalty0 035111, Jan 2019.
\newblock \doi{10.1103/PhysRevB.99.035111}.
\newblock URL \url{https://link.aps.org/doi/10.1103/PhysRevB.99.035111}.

\bibitem[Liu et~al.(2019)Liu, Liu, and Dai]{LiuDai}
Jianpeng Liu, Junwei Liu, and Xi~Dai.
\newblock Pseudo landau level representation of twisted bilayer graphene: Band
  topology and implications on the correlated insulating phase.
\newblock \emph{Phys. Rev. B}, 99:\penalty0 155415, Apr 2019.
\newblock \doi{10.1103/PhysRevB.99.155415}.
\newblock URL \url{https://link.aps.org/doi/10.1103/PhysRevB.99.155415}.

\bibitem[Ahn et~al.(2019)Ahn, Park, and Yang]{ahn2019Failure}
Junyeong Ahn, Sungjoon Park, and Bohm-Jung Yang.
\newblock Failure of {{Nielsen}}-{{Ninomiya Theorem}} and {{Fragile Topology}}
  in {{Two}}-{{Dimensional Systems}} with {{Space}}-{{Time Inversion
  Symmetry}}: {{Application}} to {{Twisted Bilayer Graphene}} at {{Magic
  Angle}}.
\newblock \emph{Physical Review X}, 9\penalty0 (2):\penalty0 021013, April
  2019.
\newblock \doi{10.1103/PhysRevX.9.021013}.

\bibitem[Soluyanov and Vanderbilt(2011)]{VanderbiltTI}
Alexey~A. Soluyanov and David Vanderbilt.
\newblock Wannier representation of ${\mathbb{z}}_{2}$ topological insulators.
\newblock \emph{Phys. Rev. B}, 83:\penalty0 035108, Jan 2011.
\newblock \doi{10.1103/PhysRevB.83.035108}.
\newblock URL \url{https://link.aps.org/doi/10.1103/PhysRevB.83.035108}.

\bibitem[Kang and Vafek(2018)]{KangVafekPRX}
Jian Kang and Oskar Vafek.
\newblock Symmetry, maximally localized wannier states, and a low-energy model
  for twisted bilayer graphene narrow bands.
\newblock \emph{Phys. Rev. X}, 8:\penalty0 031088, Sep 2018.
\newblock \doi{10.1103/PhysRevX.8.031088}.
\newblock URL \url{https://link.aps.org/doi/10.1103/PhysRevX.8.031088}.

\bibitem[Koshino et~al.(2018)Koshino, Yuan, Koretsune, Ochi, Kuroki, and
  Fu]{KoshinoYuan}
Mikito Koshino, Noah F.~Q. Yuan, Takashi Koretsune, Masayuki Ochi, Kazuhiko
  Kuroki, and Liang Fu.
\newblock Maximally localized wannier orbitals and the extended hubbard model
  for twisted bilayer graphene.
\newblock \emph{Phys. Rev. X}, 8:\penalty0 031087, Sep 2018.
\newblock \doi{10.1103/PhysRevX.8.031087}.
\newblock URL \url{https://link.aps.org/doi/10.1103/PhysRevX.8.031087}.

\bibitem[Kang and Vafek(2019)]{KangVafekPRL}
Jian Kang and Oskar Vafek.
\newblock Strong coupling phases of partially filled twisted bilayer graphene
  narrow bands.
\newblock \emph{Phys. Rev. Lett.}, 122:\penalty0 246401, Jun 2019.
\newblock \doi{10.1103/PhysRevLett.122.246401}.
\newblock URL \url{https://link.aps.org/doi/10.1103/PhysRevLett.122.246401}.

\bibitem[Wang and Vafek(2020)]{WangVafek}
Xiaoyu Wang and Oskar Vafek.
\newblock Diagnosis of explicit symmetry breaking in the tight-binding
  constructions for symmetry-protected topological systems.
\newblock \emph{Phys. Rev. B}, 102:\penalty0 075142, Aug 2020.
\newblock \doi{10.1103/PhysRevB.102.075142}.
\newblock URL \url{https://link.aps.org/doi/10.1103/PhysRevB.102.075142}.

\bibitem[Bultinck et~al.(2020{\natexlab{b}})Bultinck, Chatterjee, and
  Zaletel]{bultinck2020Mechanism}
Nick Bultinck, Shubhayu Chatterjee, and Michael~P. Zaletel.
\newblock A mechanism for anomalous {{Hall}} ferromagnetism in twisted bilayer
  graphene.
\newblock \emph{Physical Review Letters}, 124\penalty0 (16):\penalty0 166601,
  April 2020{\natexlab{b}}.
\newblock ISSN 0031-9007, 1079-7114.
\newblock \doi{10.1103/PhysRevLett.124.166601}.

\bibitem[{Hejazi} et~al.(2020){Hejazi}, {Chen}, and
  {Balents}]{HejaziHybridWannier}
Kasra {Hejazi}, Xiao {Chen}, and Leon {Balents}.
\newblock {Hybrid Wannier Chern bands in magic angle twisted bilayer graphene
  and the quantized anomalous Hall effect}.
\newblock \emph{arXiv e-prints}, art. arXiv:2007.00134, June 2020.

\bibitem[{Kwan} et~al.(2020{\natexlab{a}}){Kwan}, {Wagner}, {Chakraborty},
  {Simon}, and {Parameswaran}]{KwanDomainwalls}
Yves~H. {Kwan}, Glenn {Wagner}, Nilotpal {Chakraborty}, Steven~H. {Simon}, and
  S.~A. {Parameswaran}.
\newblock {Orbital Chern insulator domain walls and chiral modes in twisted
  bilayer graphene}.
\newblock \emph{arXiv e-prints}, art. arXiv:2007.07903, July
  2020{\natexlab{a}}.

\bibitem[Motruk et~al.(2016)Motruk, Zaletel, Mong, and
  Pollmann]{motruk2016Density}
Johannes Motruk, Michael~P. Zaletel, Roger S.~K. Mong, and Frank Pollmann.
\newblock Density matrix renormalization group on a cylinder in mixed real and
  momentum space.
\newblock \emph{Physical Review B}, 93\penalty0 (15):\penalty0 155139, April
  2016.
\newblock ISSN 2469-9950, 2469-9969.
\newblock \doi{10.1103/PhysRevB.93.155139}.

\bibitem[Zaletel et~al.(2015)Zaletel, Mong, Pollmann, and
  Rezayi]{ZaletelMulticomponent}
Michael~P. Zaletel, Roger S.~K. Mong, Frank Pollmann, and Edward~H. Rezayi.
\newblock Infinite density matrix renormalization group for multicomponent
  quantum hall systems.
\newblock \emph{Phys. Rev. B}, 91:\penalty0 045115, Jan 2015.
\newblock \doi{10.1103/PhysRevB.91.045115}.
\newblock URL \url{https://link.aps.org/doi/10.1103/PhysRevB.91.045115}.

\bibitem[Chan et~al.(2016)Chan, Keselman, Nakatani, Li, and
  White]{ChanKeselman}
Garnet Kin-Lic Chan, Anna Keselman, Naoki Nakatani, Zhendong Li, and Steven~R.
  White.
\newblock Matrix product operators, matrix product states, and ab initio
  density matrix renormalization group algorithms.
\newblock \emph{The Journal of Chemical Physics}, 145\penalty0 (1):\penalty0
  014102, 2016.
\newblock \doi{10.1063/1.4955108}.
\newblock URL \url{https://doi.org/10.1063/1.4955108}.

\bibitem[Parker et~al.(2020)Parker, Cao, and Zaletel]{parker2019local}
Daniel~E. Parker, Xiangyu Cao, and Michael~P. Zaletel.
\newblock Local matrix product operators: Canonical form, compression, and
  control theory.
\newblock \emph{Phys. Rev. B}, 102:\penalty0 035147, Jul 2020.
\newblock \doi{10.1103/PhysRevB.102.035147}.
\newblock URL \url{https://link.aps.org/doi/10.1103/PhysRevB.102.035147}.

\bibitem[Liu et~al.(2020)Liu, Khalaf, Lee, and Vishwanath]{liu2020Nematic}
Shang Liu, Eslam Khalaf, Jong~Yeon Lee, and Ashvin Vishwanath.
\newblock Nematic topological semimetal and insulator in magic angle bilayer
  graphene at charge neutrality.
\newblock \emph{arXiv:1905.07409 [cond-mat]}, April 2020.

\bibitem[{Cea} and {Guinea}(2020)]{Cea}
Tommaso {Cea} and Francisco {Guinea}.
\newblock {Band structure and insulating states driven by Coulomb interaction
  in twisted bilayer graphene}.
\newblock \emph{\prb}, 102\penalty0 (4):\penalty0 045107, July 2020.
\newblock \doi{10.1103/PhysRevB.102.045107}.

\bibitem[{Kwan} et~al.(2020{\natexlab{b}}){Kwan}, {Hu}, {Simon}, and
  {Parameswaran}]{KwanExciton}
Yves~H. {Kwan}, Yichen {Hu}, Steven~H. {Simon}, and S.~A. {Parameswaran}.
\newblock {Exciton band topology in spontaneous quantum anomalous Hall
  insulators: applications to twisted bilayer graphene}.
\newblock \emph{arXiv e-prints}, art. arXiv:2003.11560, March
  2020{\natexlab{b}}.

\bibitem[Qi(2011)]{qi2011Generic}
Xiao-Liang Qi.
\newblock Generic {{Wave}}-{{Function Description}} of {{Fractional Quantum
  Anomalous Hall States}} and {{Fractional Topological Insulators}}.
\newblock \emph{Physical Review Letters}, 107\penalty0 (12):\penalty0 126803,
  September 2011.
\newblock \doi{10.1103/PhysRevLett.107.126803}.

\bibitem[Ehlers et~al.(2017)Ehlers, White, and Noack]{ehlers2017hybrid}
G~Ehlers, SR~White, and RM~Noack.
\newblock Hybrid-space density matrix renormalization group study of the doped
  two-dimensional hubbard model.
\newblock \emph{Physical Review B}, 95\penalty0 (12):\penalty0 125125, 2017.

\bibitem[Marzari and Vanderbilt(1997)]{marzari1997Maximally}
Nicola Marzari and David Vanderbilt.
\newblock Maximally localized generalized {{Wannier}} functions for composite
  energy bands.
\newblock \emph{Physical Review B}, 56\penalty0 (20):\penalty0 12847--12865,
  November 1997.
\newblock \doi{10.1103/PhysRevB.56.12847}.

\bibitem[{King-Smith} and Vanderbilt(1993)]{king-smith1993Theory}
R.~D. {King-Smith} and David Vanderbilt.
\newblock Theory of polarization of crystalline solids.
\newblock \emph{Physical Review B}, 47\penalty0 (3):\penalty0 1651--1654,
  January 1993.
\newblock \doi{10.1103/PhysRevB.47.1651}.

\bibitem[Resta(1993)]{resta1993Macroscopic}
R.~Resta.
\newblock Macroscopic {{Electric Polarization}} as a {{Geometric Quantum
  Phase}}.
\newblock \emph{Europhysics Letters (EPL)}, 22\penalty0 (2):\penalty0 133--138,
  April 1993.
\newblock ISSN 0295-5075.
\newblock \doi{10.1209/0295-5075/22/2/010}.

\bibitem[Vanderbilt and {King-Smith}(1993)]{vanderbilt1993Electric}
David Vanderbilt and R.~D. {King-Smith}.
\newblock Electric polarization as a bulk quantity and its relation to surface
  charge.
\newblock \emph{Physical Review B}, 48\penalty0 (7):\penalty0 4442--4455,
  August 1993.
\newblock \doi{10.1103/PhysRevB.48.4442}.

\bibitem[Marzari et~al.(2012)Marzari, Mostofi, Yates, Souza, and
  Vanderbilt]{marzari2012Maximally}
Nicola Marzari, Arash~A. Mostofi, Jonathan~R. Yates, Ivo Souza, and David
  Vanderbilt.
\newblock Maximally localized {{Wannier}} functions: {{Theory}} and
  applications.
\newblock \emph{Reviews of Modern Physics}, 84\penalty0 (4):\penalty0
  1419--1475, October 2012.
\newblock ISSN 0034-6861, 1539-0756.
\newblock \doi{10.1103/RevModPhys.84.1419}.

\bibitem[Pirvu et~al.(2010)Pirvu, Murg, Cirac, and Verstraete]{pirvu2010matrix}
Bogdan Pirvu, Valentin Murg, J~Ignacio Cirac, and Frank Verstraete.
\newblock Matrix product operator representations.
\newblock \emph{New Journal of Physics}, 12\penalty0 (2):\penalty0 025012,
  2010.

\bibitem[Fishman et~al.(2020)Fishman, White, and
  Stoudenmire]{fishman2020ITensor}
Matthew Fishman, Steven~R. White, and E.~Miles Stoudenmire.
\newblock The {{ITensor Software Library}} for {{Tensor Network Calculations}}.
\newblock \emph{arXiv:2007.14822 [cond-mat, physics:physics]}, July 2020.

\bibitem[Hauschild and Pollmann(2018)]{hauschild2018efficient}
Johannes Hauschild and Frank Pollmann.
\newblock Efficient numerical simulations with tensor networks: Tensor network
  python (tenpy).
\newblock \emph{SciPost Phys. Lect. Notes}, page~5, 2018.
\newblock \doi{10.21468/SciPostPhysLectNotes.5}.
\newblock URL \url{https://scipost.org/10.21468/SciPostPhysLectNotes.5}.

\bibitem[Zhang et~al.(2019)Zhang, Mao, and Senthil]{YaHuiChern}
Ya-Hui Zhang, Dan Mao, and T.~Senthil.
\newblock Twisted bilayer graphene aligned with hexagonal boron nitride:
  Anomalous hall effect and a lattice model.
\newblock \emph{Phys. Rev. Research}, 1:\penalty0 033126, Nov 2019.
\newblock \doi{10.1103/PhysRevResearch.1.033126}.
\newblock URL \url{https://link.aps.org/doi/10.1103/PhysRevResearch.1.033126}.

\bibitem[Kim et~al.(2015)Kim, Wieder, Kane, and Rappe]{kim2015Dirac}
Youngkuk Kim, Benjamin~J. Wieder, C.~L. Kane, and Andrew~M. Rappe.
\newblock Dirac {{Line Nodes}} in {{Inversion}}-{{Symmetric Crystals}}.
\newblock \emph{Physical Review Letters}, 115\penalty0 (3):\penalty0 036806,
  July 2015.
\newblock \doi{10.1103/PhysRevLett.115.036806}.

\bibitem[Sundaram and Niu(1999)]{PhysRevB.59.14915}
Ganesh Sundaram and Qian Niu.
\newblock Wave-packet dynamics in slowly perturbed crystals: Gradient
  corrections and berry-phase effects.
\newblock \emph{Phys. Rev. B}, 59:\penalty0 14915--14925, Jun 1999.
\newblock \doi{10.1103/PhysRevB.59.14915}.
\newblock URL \url{https://link.aps.org/doi/10.1103/PhysRevB.59.14915}.

\bibitem[Hwang and Das~Sarma(2007)]{HwangDasSarma}
E.~H. Hwang and S.~Das~Sarma.
\newblock Dielectric function, screening, and plasmons in two-dimensional
  graphene.
\newblock \emph{Phys. Rev. B}, 75:\penalty0 205418, May 2007.
\newblock \doi{10.1103/PhysRevB.75.205418}.
\newblock URL \url{https://link.aps.org/doi/10.1103/PhysRevB.75.205418}.

\bibitem[Liang and Pang(1994)]{liang1994approximate}
Shoudan Liang and Hanbin Pang.
\newblock Approximate diagonalization using the density matrix
  renormalization-group method: A two-dimensional-systems perspective.
\newblock \emph{Physical Review B}, 49\penalty0 (13):\penalty0 9214, 1994.

\bibitem[Schollw{\"o}ck(2011)]{schollwock2011density}
Ulrich Schollw{\"o}ck.
\newblock The density-matrix renormalization group in the age of matrix product
  states.
\newblock \emph{Annals of physics}, 326\penalty0 (1):\penalty0 96--192, 2011.

\bibitem[Crosswhite and Bacon(2008)]{crosswhite2008finite}
Gregory~M. Crosswhite and Dave Bacon.
\newblock Finite automata for caching in matrix product algorithms.
\newblock \emph{Phys. Rev. A}, 78\penalty0 (1):\penalty0 012356, jul 2008.
\newblock ISSN 1050-2947.
\newblock \doi{10.1103/PhysRevA.78.012356}.
\newblock URL \url{http://arxiv.org/abs/0708.1221
  http://dx.doi.org/10.1103/PhysRevA.78.012356
  https://link.aps.org/doi/10.1103/PhysRevA.78.012356}.

\end{thebibliography}

\newpage
    
\appendix\textbf{}

\section{Interacting Bistritzer-MacDonald Model}
\label{app:IBM_model}

In this appendix, we review the interacting BM model projected into the flat bands. We use the conventions and definition of $\v{h}(\v{k})$ from Supp. Mat. I of \ref{bultinck2019Ground}.

Let us first consider the Coulomb interaction

\begin{eqnarray}
    \hat{H}_C & = & \frac{1}{2}\int \mathrm{d}\v{r}\int \mathrm{d}\v{r}'\, V(\v{r}-\v{r}') \psi^\dagger_\alpha(\v{r})\psi^\dagger_\beta(\v{r}')\psi_\beta(\v{r}')\psi_\alpha(\v{r}) \nonumber \\
    & = & \frac{1}{2A}\sum_{\v{k},\v{k}',\v{q}}V_{\v{q}}\, \psi^\dagger_{\alpha,\v{k}+\v{q}} \psi^\dagger_{\beta,\v{k}'-\v{q}} \psi_{\beta,\v{k}'} \psi_{\alpha,\v{k}}\, ,
\end{eqnarray}
where $A$ is the sample area, $V_{\v{q}} = \int \mathrm{d}\v{r}\,V(\v{r}) e^{i\v{q}\cdot\v{r}} $ and $\alpha,\beta$ are combined layer-sublattice indices. Summation over repeated indices is implicit. The Fourier components of the Fermi operators are defined to satisfy the canonical anti-commutation relations:

\begin{equation}
  \{\psi^\dagger_{\alpha,\v{k}},\psi_{\beta,\v{k}'}\} = \delta_{\alpha,\beta}\delta_{\v{k},\v{k}'}
\end{equation}
Next, we relabel the sums over the momenta $\v{k}$ and $\v{k}'$ as

\begin{equation}
    \sum_{\v{k}} \rightarrow \sum_{\v{k}\in \text{mBZ}}\sum_{\tau}\sum_{\v{G}}\, ,
\end{equation}
where $\tau = \pm$ is a valley label, and $\v{G}$ are the moir\'e reciprocal lattice vectors. We can now approximate the Coulomb interaction as

\begin{eqnarray}\label{eq:Coulombappr}
    \hat{H}_C & = & \frac{1}{2A}\sum_{\tau,\tau'}\sum_{\v{q}}\sum_{\v{k},\v{k}'\in \mathrm{mBZ}} \sum_{\v{G},\v{G}'} V_{\v{q}} \times \\  
    & & \psi^\dagger_{\alpha,\tau,\v{G}}(\v{k}+\v{q}) \psi^\dagger_{\beta,\tau',\v{G}'}(\v{k}'-\v{q}) \psi_{\beta,\tau',\v{G}'}(\v{k}') \psi_{\alpha,\tau,\v{G}}(\v{k})\, ,\nonumber
\end{eqnarray}
where $\psi^\dagger_{\alpha,\tau,\v{G}}(\v{k}) = \psi^\dagger_{\alpha,\v{k}+\tau \v{K}_{\Gamma} + \v{G}} $ and $\v{K}_{\Gamma}$ denotes the $\Gamma$ point of the mBZ centered at the $K$ points of the graphene layers. Note that by definition, $\psi^\dagger_{\alpha,\tau,\v{G}}(\v{k}+\v{G}') = \psi^\dagger_{\alpha,\tau,\v{G}+\v{G}'}(\v{k})$. Eq.~\eqref{eq:Coulombappr} is only an approximation to the complete Coulomb interaction, as inter-valley scattering terms have been neglected. This can be justified because of the long-range nature of the interaction, which suppresses inter-valley scattering by a factor of order $V_{2\v{K}_\Gamma}/ V_0$.

Next, we perform a unitary transformation to the BM band basis and define

\begin{equation}
    f^\dagger_{m,\tau,\v{k}} = \sum_{\alpha,\v{G}} u_{m,\tau;\alpha,\v{G}}(\v{k}) \psi^\dagger_{\alpha,\tau,\v{G}}(\v{k})\, ,
\end{equation}
where $m$ labels the bands of the single-valley BM model, and $u_{m,\tau;\alpha,\v{G}}(\v{k})$ are the periodic part of the Bloch states of the BM Hamiltonian. Note that $f^\dagger_{m,\tau,\v{k}+\v{G}'} = f^\dagger_{m,\tau,\v{k}}$ because the BM Bloch states satisfy $u_{m,\tau;\alpha,\v{G}}(\v{k}+\v{G}') = u_{m,\tau;\alpha,\v{G}+\v{G}'}(\v{k})$. With this definition, Eq.~\eqref{eq:Coulombappr} takes the following form in the BM band basis:

\begin{eqnarray}\label{eq:BMbasis}
    \hat{H}_C & = & \frac{1}{2A}\sum_{\tau,\tau'}\sum_{\v{q}}\sum_{\v{k},\v{k}'\in \mathrm{mBZ}} V_{\v{q}} \left[\Lambda^{\tau}_{\v{q}}(\v{k}) \right]_{mn} \left[\Lambda^{\tau'}_{-\v{q}}(\v{k}') \right]_{m'n'} \nonumber\\  
    & & \times f^\dagger_{m,\tau,\v{k}+\v{q}} f^\dagger_{m',\tau',\v{k}'-\v{q}} f_{n',\tau',\v{k}'} f_{n,\tau,\v{k}}\, ,
\end{eqnarray}
where the sums over band indices are implicit, and the form factors are given by

\begin{equation}
    \left[\Lambda^{\tau}_{\v{q}}(\v{k}) \right]_{mn} = \sum_{\alpha,\v{G}} u^*_{m,\tau;\alpha,\v{G}}(\v{k}+\v{q}) u_{n,\tau;\alpha,\v{G}}(\v{k})\, .
\end{equation}
In this work, we consider the single-valley model, which means that we fix all valley labels, i.e. $\tau = +$ everywhere. The single-valley Coulomb interaction is then given by

\begin{equation}
    \hat{H}_{C,sv} = \frac{1}{2A} \sum_{\v{q},\v{k},\v{k}'} V_{\v{q}} :\left[ \v{f}^\dagger_{\v{k}+\v{q}} \Lambda_{\v{q}}(\v{k}) \v{f}_{\v{k}}\right]
    \left[ \v{f}^\dagger_{\v{k}'-\v{q}} \Lambda_{-\v{q}}(\v{k}') \v{f}_{\v{k}'}\right]:
\end{equation}
where $\Lambda_{\v{q}}(\v{k}) = \Lambda^+_{\v{q}}(\v{k})$ and $\v{f}^\dagger_{\v{k}}=\v{f}^\dagger_{+,\v{k}}$ is a vector of creation operators running over the BM bands.

As a final step, we now project $\hat{H}_C$ into the subspace where all remote valence bands are occupied, and all remote conduction bands are empty. To do that, we first define following Hartree Hamiltonian functional:

\begin{align}
    \hat{H}_{h}[P(\v{k})] = \frac{V_0}{A}\sum_{\v{G}} \left[\sum_{\v{k'}} \text{tr} \left( P(\v{k}') \Lambda_{\v{G}}(\v{k}') \right) \right] \nonumber\\ 
    \times\sum_{\v{k}}\v{f}^\dagger_{\v{k}} \Lambda_{-\v{G}}(\v{k}) \v{f}_{\v{k}}\, ,
\end{align}
where the fermion operators are restricted to the flat bands, and which depends on a general Slater determinant correlation matrix $\langle f^\dagger_{m,\v{k}}f_{n,\v{k}'}\rangle = \sum_{\v{G}}\delta_{\v{k}+\v{G},\v{k}'} \left[P(\v{k})\right]_{nm}$. We also similarly define a Fock Hamiltonian functional:

\begin{equation}
    \hat{H}_f[P(\v{k})] = -\frac{1}{A}\sum_{\v{q},\v{k}} V_{\v{q}}\; \v{f}^\dagger_{\v{k}}\Lambda_{\v{q}}(\v{k}-\v{q})P(\v{k}-\v{q})\Lambda_{-\v{q}}(\v{k}) \v{f}_{\v{k}} \, ,
\end{equation}
where again the fermion operators are restricted to the flat bands. With these definitions, one can write the flat-band projected Coulomb interaction as

\begin{equation}
    \hat{H}_{C,sv}\Big|_{FB} = \tilde{H}_{C,sv} + \hat{H}_h[P_r(\v{k})] + \hat{H}_f[P_r(\v{k})]\, ,
\end{equation}
where $\tilde{H}_{C,sv}$ is obtained from $\hat{H}_{C,sv}$ by simply restricting all band indices to the flat bands, and $P_r(\v{k})$ is the correlation matrix of the Slater determinant where only the remote valence bands are filled.

Having obtained the flat-band projected single-valley Coulomb interaction, we now have to be careful not to double count certain interaction effects. In particular, the value of the hopping parameter in the tight-binding model of mono-layer graphene is chosen to best reproduce the experimentally observed Dirac velocity. Importantly, this Dirac velocity is already renormalized by the Coulomb interaction. So if we want to explicitly add back the complete Coulomb interaction, we must make sure not to forget to subtract off the renormalization of the dispersion. In practice, this means that we have to subtract off the following Hartree-Fock Hamiltonian:

\begin{equation}
    \hat{H}_{sub} = \hat{H}_h[P_0(\v{k})] + \hat{H}_f[P_0(\v{k})]\, ,
\end{equation}
where $P_0(\v{k})$ is the correlation matrix of the charge-neutrality Slater determinant of two \emph{decoupled} graphene layers \cite{XieMacDonald} restricted single spin and valley, expressed in the BM band basis. The complete projected single-valley BM model thus takes the form

\begin{equation}\label{eq:complete}
    \hat{H} = \hat{H}_{BM,sv} + \hat{H}_{C,sv}\Big|_{FB} - \hat{H}_{sub}\,.
\end{equation}
Because the inter-layer tunneling is only a small perturbation compared to the intra-layer hopping, it does not significantly change the remote bands. It thus holds to a very good approximation that

\begin{equation}
    \left[P_r(\v{k}) \right]_{mn} =  \left[P_0(\v{k}) \right]_{mn} \text{ for } m,n \in \text{remote bands}\, .
\end{equation}
Now combining all the single-particle terms in Eq.~\eqref{eq:complete}, one obtains the matrix $h(\v{k})$ defined in Eq.~\eqref{eq:IBM}. The remaining interaction Hamiltonian of Eq.~\eqref{eq:complete} is then exactly the second term of Eq.~\eqref{eq:IBM}.

Finally, let us also comment on the value of the dielectric constant $\epsilon_r$ that we use for the projected model. The dielectric constant of the unprojected model is determined by the insulating hexagonal Boron-Nitride (hBN) substrate. In particular, it is given by the geometric mean of the dielectric constants of hBN parallel and orthogonal to the atomic plane: $\epsilon_{\textrm{hBN}} = \sqrt{\epsilon_\perp \epsilon_\parallel} \approx 4.4$. When we restrict to the flat bands we have to take into account that the filled remote bands which have been projected out act as another insulating background, whose polarization will also contribute to the dielectric constant. We phenomologically incorporate this effect by screening the Coulomb potential with the static RPA polarization of eight Dirac cones at neutrality. In doing so, we ignore the small inter-layer tunneling terms (whose effect on the remote bands is expected to be small), and we overcount a small contribution coming from the states in the flat bands themselves. The polarization bubble of Dirac fermions in mono-layer graphene was calculated in Ref.~\cite{HwangDasSarma}. Using the results from that paper, we obtain a dielectric constant $\epsilon_r = \epsilon_{\textrm{hBN}} + N_D*0.73 \approx 10.3$, where $N_D = 8$ is the number of Dirac cones. This gives an order-of-magnitude estimate for the dielectric constant of the projected model. Given the many approximations used in deriving this estimate, it is safest to assume that the true dielectric constant lies somewhere in the interval $\epsilon_r \sim 6 - 12$.

We also note that the effective gate distance $d$ appearing in $V_\mathbf{q} = \frac{e^2}{4 \pi \epsilon q} \tanh(q d)$ is modified by the anisostropy of the substrate dielectric constant. Specifically, if  $d_g$ is the physical distance to the gate, then $d = \sqrt{\frac{\epsilon_\perp}{\epsilon_\parallel} } d_g$, due to the bending of field lines in an anisotropic substrate.

\section{Wannier Localization, Gauge fixing, and Symmetrization}\label{app:gauge_fixing}
In this section, we list several invariants enforced by the maximal localization of Wannier orbitals and gauge fixing. We also comment on the different between our gauge choice and the gauge choice in \cite{KangVafek}, and how we can map states from one gauge to another.

We denote the periodic part of momentum space orbitals $c^\dagger_{\pm, \v{k}}$ by $\ket{u_\pm(k_x, k_y)}$. We define $2\times 2$ overlap matrix as

\begin{equation}
    O_{\alpha\beta}(\v{k}, \v{k}') = \braket{u_\alpha(\v{k})|u_\beta(\v{k}')}.
\end{equation}
For each overlap matrix $O$, we define unitary overlap matrix $\tilde{O}$ by the unitary part of the polar decomposition of $O$.
Wannier localization \cite{marzari1997Maximally} guarantees that the unitary overlap matrix in the $k_x$ direction is always given by

\begin{equation}
    \tilde{O}_{\alpha\beta}(\v{k}, \v{k} + \Delta k_x) = \delta_{\alpha\beta} e^{iP_x(\alpha, k_y)},
\end{equation}
where $\Delta k_x$ is the unit of discretization in the $x$ direction. Physically, this corresponds to choosing constant $A_x = P_x/G_x$. 

Wannier localization fixes relative phases within each $k_y$ mode. To fix relative phases between different $k_y$ modes, we demand a continuity criterion in the $k_y$ direction. One natural choice is to demand the unitary overlap matrix in the $y$ direction  be the identity matrix:
\begin{equation}
    \tilde{O}_{\alpha\beta}(\v{k}, \v{k} + \Delta k_y) = \delta_{\alpha\beta},
\end{equation}
where the $x$ component of $\v{k}$ is some fixed $k_{x0}$. The condition notably does \textit{not} wraparound the mBZ: the mode near $k_y/G_y = 0.5$ and the mode near $k_y/G_y = -0.5$ are not subject to the continuity condition with each other.

This leaves us with a global $U(1) \times U(1)$ phase ambiguity. Fortunately, the Wannier localization and continuity conditions are compatible with the symmetry requirements

\begin{equation}
    \begin{split}
        C_2\mathcal{T}\ket{u_\pm(\v{k})} &= \ket{u_\mp(\v{k})},\\
        C_{2x}\ket{u_\pm(\v{k})} &= \mp i\ket{u_\mp(-\v{k})},
    \end{split}
\end{equation}
which fixes the gauge up to an overall minus sign.

We note that the gauge fixing condition is different from \cite{KangVafek} on two accounts: we use a different continuity criterion, and we fix $C_{2x}$ to act as $\sigma_y$ rather than $\sigma_x$. The latter is easy to account for by a $e^{i\pi\sigma_x/4}$ rotation, which maps $\sigma_y$ to $\sigma_x$, but there is no guarantee that the gauge is the same even after such rotation.

Luckily, in the case of $L_y = 6, \Phi_y = \pi, w_0=0.185$, $e^{i\pi\sigma_x/4}$ rotation maps our ground state ansatz in Sec.~\ref{sec:thin_cylinder} to their ground state ansatz, suggesting our gauge choice is the same up to the rotation. In particular, this means we can go from their parametrization to our parametrization simply by changing $\varphi \to \varphi + \pi/2$. We make use of this fact when we write down the Stripe ansatz in our gauge.

	\section{An Uncompressed MPO for BLG}
	\label{app:uncompressed_MPO_construction}

This Appendix details the construction of the uncompressed infinite MPO for bilayer graphene. More concretely, we show how to construct an iMPO which encodes arbitrary 4-fermion interactions up to range $R$.
\footnote{For the construction of \textit{finite} MPO for interacting fermions, as is relevant to quantum chemistry applications, we refer the reader to \cite{ChanKeselman}.}
Schematically, what is the iMPO for the Hamiltonian
\begin{equation}
	\widehat{H} = \sum_{ij} t_{ij} c^{\dagger}_{i} c_{j} + \sum V_{ijk\ell} c_{i}^{\dagger} c_{j}^{\dagger} c_k c_\ell,
\end{equation}
for arbitrary interactions $t$ and $V$?

We proceed in several stages. First we provide a straightforward construction of such an MPO of size $D = O(R^3)$. However, this gives an iMPO which is too large to even begin compressing ($D \approx 8,000,000$ for our standard Hamiltonian parameters). Second, therefore, we provide a more efficient MPO which represents the same operator at size $D = O(R^2)$, which will give $D \approx 100,000$ --- small enough to be compressed.

\subsection{A Straightforward MPO Construction}

This section will give a relatively straightforward way to construct an MPO for a  long-range 4-fermion Hamiltonian. For simplicity, we restrict ourselves to a 1D spin chain of spinless fermions and no internal degrees of freedom and construct a iMPO representation for a class of ``toy" 4-body Hamiltonians
\begin{equation}
	\widehat{H}_{\text{simple}} = \sum_{i<j<k<\ell} V_{ijk\ell} c_{i}^{\dagger} c_j^{\dagger} c_k c_{\ell}
	\label{eq:H_4_bdy_simple}
\end{equation}
with  $\n{i-j}, \n{j-k},\n{k-\ell} \le R$. \footnote{This is slightly different from how we implemented $\Delta x$ cutoff for BLG, where we demanded $ \ell - i < R$. This will only change the complexity by a constant factor, so we stick to the simpler cutoff} In particular, we have imposed the artificial properties that $H$ contains only terms of the form $c^{\dagger} c^{\dagger} c c$ (which makes this non-Hermitian), is completely translationally invariant, and has $i < j < k < \ell$ with strict inequalities. These restrictions will be lifted below.

\begin{figure}[h]
	\includegraphics{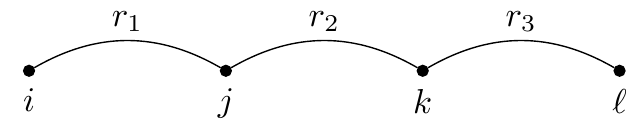}
\end{figure}

\begin{figure}
	\includegraphics[width=\linewidth]{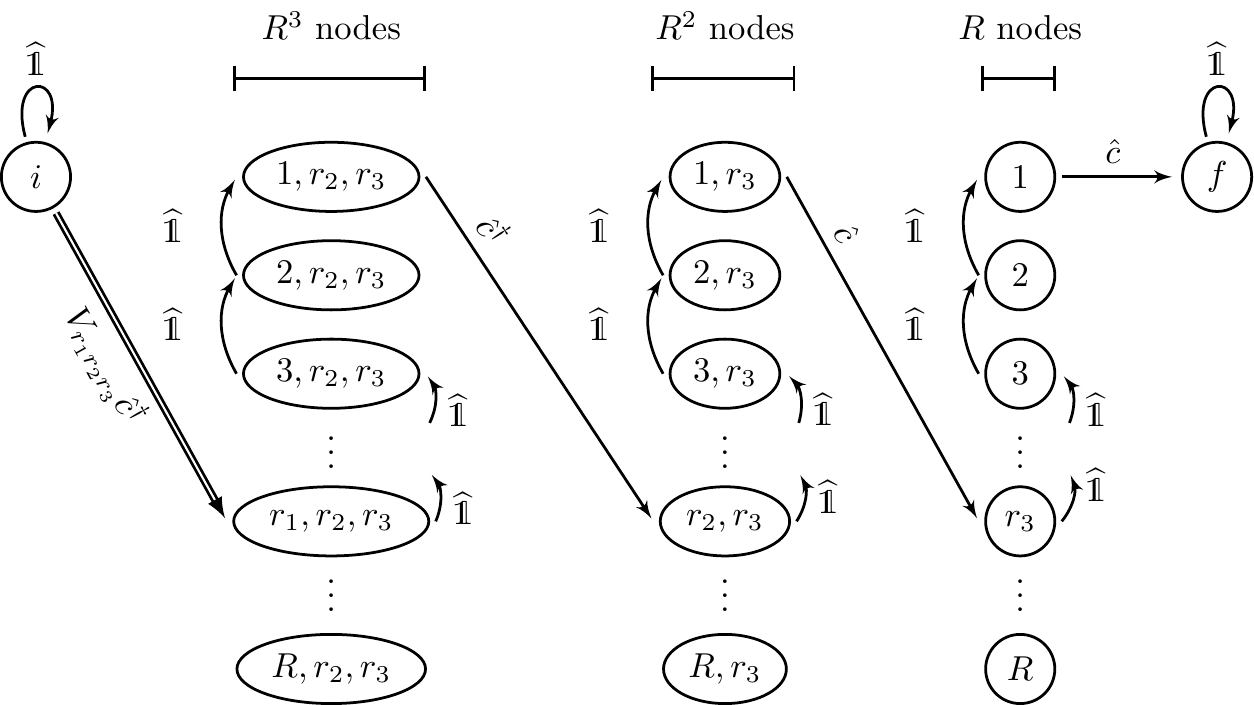}
	\caption{A straightforward but rather inefficient MPO for $\widehat{H}_{\text{simple}}$.}
	\label{fig:straightforward_4_bdy_MPO}
\end{figure}

Let us work in the finite state machine picture for the iMPO. For convenience, define $r_1 = j-i$, $r_2 = k-j$, $r_3 = \ell-k$. As $0 < r_1, r_2, r_3 \le R$, our operator can have up to $R^3$ terms. For each one, we must first place a $c^{\dagger}$, then place $r_1 -1$ identity operators $\1$, then place another $c^{\dagger}$, and so on. To encode these into the finite state machine, we make a unique path for each term from the initial to final nodes, as shown in Fig.~\ref{fig:straightforward_4_bdy_MPO}. For each term in the Hamiltonian, there is an edge $V_{r_1 r_2 r_3}$ from the initial node $(i)$ to node $(r_1,r_2,r_3)$. After that, there is a unique path from node $(r_1,r_2,r_3)$ to node $(f)$. The nodes are labeled by the distances to non-identity nodes that have yet to be placed. In each column, the path simply ``counts down" in the first index, until it reaches 1. At that site a non-identity on-site operator is placed, and the path goes on to the next column. As a matrix, the operation of counting down is encoded by a ``skipping matrix", with identities on the first superdiagonal:

\[
	\widehat{\mathds{S}} = \begin{pmatrix} 
0 & \1 & & & &\\
& 0 & \1 & & & \\
& & \ddots & \ddots & \\
& & & 0 & \1\\
& & &  & 0\\
	\end{pmatrix}. 
\]
We therefore adopt a concise diagrammatic notation where \textit{rectangular} nodes represent nodes that count down, as shown in Fig.~\ref{fig:straightforward_4_bdy_MPO_notation}. In this notation, $\widehat{\mathds{S}}_i$ acts on a node as\footnote{This is slight abuse of notation, since it doesn't convey that $\widehat{\mathds{S}}_i$ places an identity operator upon the transition.}

\begin{equation}
    \widehat{\mathds{S}}_i(\widehat{O}[r_i...]) = (\widehat{O}[r_i-1...]).
\end{equation}

We will use this notation below. We also note that some transitions are deterministic, in the sense that the next node is fully specified by the current node (e.g.~$(2, r_2, r_3) \rightarrow (1, r_2, r_3)$). On the other hand, the transition from the initial node is highly branching, giving rise to the $R^3$ different terms in the Hamiltonian. We show such highly branching transitions in the figure by double arrows, while deterministic transitions are shown by single solid arrows.

\begin{figure*}
	\includegraphics{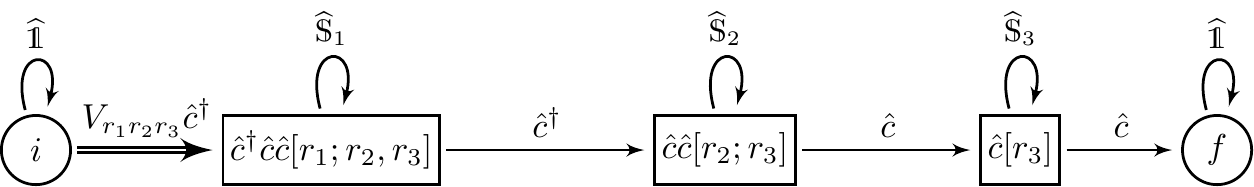}\\
	\vspace{2em}
	\includegraphics{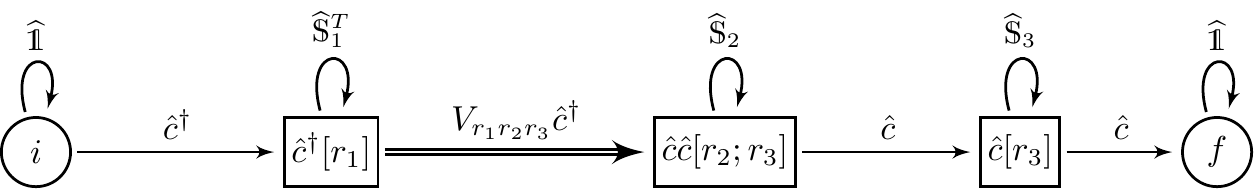}
	\caption{(Top) Definition of the compact state machine notation. This is the same state machine as in Fig.~\ref{fig:straightforward_4_bdy_MPO}, but where columns have been replaced by rectangular nodes. Rectangular nodes label the on-site operators that are yet to be placed, as well as the distances to them. The first rectangle contains $R^3$ nodes, the second contains $R^2$ and the third contains $R$. The self-loop $\widehat{\mathds{S}}$ means one should place an identity $\1$ and decrement the first index of the node. (Bottom) Another MPO which represents the same Hamiltonian with only $D = O(R^2)$ nodes instead of $O(R^3)$. Rectangular nodes to the right of the branching arrow are labeled similarly to the (Top), while the rectangular node to the left of the branching arrow is labeled by which operator it has already picked up, and the distance to it.}
	\label{fig:straightforward_4_bdy_MPO_notation}
\end{figure*}

Overall, this construction requires $C_1\left(R \right) := R^3 + R^2 + R + 2$ nodes. This is already somewhat efficient, as multiple terms will partially reuse the same paths. For instance, the paths through the state machine for $V_{5, r_2, r_3}$ and $V_{6,r_2,r_3}$ will be the same after the first few edges. To add other types of terms, such as $c c^{\dagger} c c^{\dagger}$ or $c c c^{\dagger} c^{\dagger}$, one must duplicate all these nodes, giving  $4C_1(R)$ nodes. For bilayer graphene, we use a unit cell of size 12 and interaction cutoff of 10 unit cells, giving $D = 6,970,088$ just for the four-body terms. This makes the $O(D^3)$ compression algorithm impractical, so we must seek a more efficient way to encode the MPO.

\subsection{An Efficient MPO Construction}

We now describe a more efficient way to encode the uncompressed MPO, and give explicit state machines for 2-body, 3-body, and 4-body interaction terms. The key idea is to place the coefficient $V_{r_1r_2r_3}$ in the \textit{middle} of the path rather than at the beginning.

Fig.~\ref{fig:straightforward_4_bdy_MPO_notation} (Bottom) shows the more efficient construction of the MPO for the same operator as before, Eq.~\eqref{eq:H_4_bdy_simple}. Let us unpack how it works. For each term $V_{r_1 r_2 r_3} c^{\dagger} c^{\dagger} c c$, we start at node $(i)$, jump to the node ($\hat{c}^\dagger[1]$) by placing $c^\dagger$, then ``count up'' to $r_1$, at which point we place the on-site operator $V_{r_1 r_2 r_3} c^{\dagger}$. The state machine then ``counts down" for distance $r_2$, places a $c$ operator, counts down for distance $r_3$, places a second $c$, and reaches the final node.

To distinguish ``count up" skipping matrix and "count down" skipping matrix, we introduced ``transposed skipping matrix'' $\widehat{\mathds{S}}^T_i$ such that

\begin{equation}
    \widehat{\mathds{S}}^T_i (\widehat{O}[r_i...]) = (\widehat{O}[r_i+1...]).
\end{equation}

The advantage over the previous method is that, instead of having a unique path for each term in the Hamiltonian, many of the paths are partially shared. For example the terms $V_{r_1 5 r_3}$ and $V_{r_1 7 r_3}$ will share the same first $r_1$ and last $r_3$ steps in their path through the state machine. This means that instead of requiring columns of sizes $R^3, R^2$, and $R$, we instead only need $R, R^2$, and $R$, which is vastly more efficient. Another way to see this is to observe when the highly-branching transitions occur. In this more efficient algorithm, branching occurs between $\hat{c}^\dagger$ to $\hat{c}^\dagger\hat{c}^\dagger$), thereby reducing the number of paths to keep track of to $R^2$. 
Of course, one could continue to further optimize the layout of the state machine, but we do not need a generic or completely optimal solution to this problem, only one that will render the BLG Hamiltonian small enough to fit in memory to be compressed.

Now that we have described the technique for writing down a sufficient efficient MPO construction, we will relax the artificial assumptions. Let us first relax the assumption that only $c^{\dagger} c^{\dagger} c c$ terms appear.  In general, we will have a Hamiltonian of the form
\begin{align}
	\label{eq:full_H_234_bdy}
    H \ &=\ H_{\text{hop}} + H_{\text{int}} = H_{\text{hop}} + H_2 + H_3 + H_4\\
	\nonumber
    H_{\text{hop}} &= \sum_{i<j} \tilde{V}^{\dot{c}c}_{ij} c^\dagger_i c_j + \sum_{i<j} \tilde{V}^{c\dot{c}}_{ij}c_i c^\dagger_j\\ \nonumber
    H_2 &=  \sum_{i<j} \tilde{V}^{nn}_{ij} n_i n_j\\
	\nonumber
	H_3 &= \sum_{i<j<k} 	\tilde{V}^{\dot{c}nc}_{ijk} c^\dagger_i n_j c_k
	+ \tilde{V}^{n\dot{c}c}_{ijk} n_i c^\dagger_j c_k
	+ \tilde{V}^{\dot{c}cn}_{ijk} c^\dagger_i c_j n_k\\
	\nonumber
	\ &\hspace{3em}+ \tilde{V}^{nc\dot{c}}_{ijk} n_i c_j c^\dagger_k
	+ \tilde{V}^{cn\dot{c}}_{ijk} c_i n_j c^\dagger_k
	+ \tilde{V}^{c\dot{c}n}_{ijk} c_i c^\dagger_j n_k\\
	\nonumber
    H_4 &= \sum_{i<j<k<\ell}
    \tilde{V}^{\dot{c}\dot{c}cc}_{ijk\ell} c^\dagger_i c^\dagger_j c_k c_\ell
    + \tilde{V}^{\dot{c}c\dot{c}c}_{ijk\ell} c^\dagger_i c_j c^\dagger_k c_\ell\\
	\nonumber
	&\hspace{3em}+ \tilde{V}^{\dot{c}cc\dot{c}}_{ijk\ell}c^\dagger_i c_j c_k c_\ell^\dagger
    + \tilde{V}^{cc\dot{c}\dot{c}}_{ijk\ell}c_i c_j c_k^\dagger c_\ell^\dagger\\
	\nonumber
	&\hspace{3em}+ \tilde{V}^{c\dot{c}c\dot{c}}_{ijk\ell}c_i c_j^\dagger c_k c_\ell^\dagger
    + \tilde{V}^{c\dot{c}\dot{c}c}_{ijk\ell}c_i c_j^\dagger c_k^\dagger c_\ell,
\end{align}
with all $2$-body, $3$-body, and $4$-body interactions involving $c^{\dagger}, c$, and $n$ operators, and $c^\dagger$ is represented by $\dot{c}$ for concision. We can make this representation more amenable to finite state machine description by mapping $\tilde{V}_{ijkl}$ to $V_{r_0r_1r_2r_3}$ such that $V_{i(j-i)(k-j)(\ell-k)} = \tilde{V}_{ijkl}$(See Fig.~\ref{fig:4-body-efficient-real}). By combining all of these, one can generate the full Hamiltonian Eq.~\eqref{eq:full_H_234_bdy} with $D(R) = 4R^2 + 6R+2$ nodes. In practice, then,  the spinless, single-valley BLG Hamiltonian on a cylinder with 6 momentum cuts (i.e. a unit cell of 12 sites) and interactions of range 10 unit cells is size $D \approx 58,000$ before compression. 

\begin{figure}
    \centering
    \includegraphics[width=\linewidth]{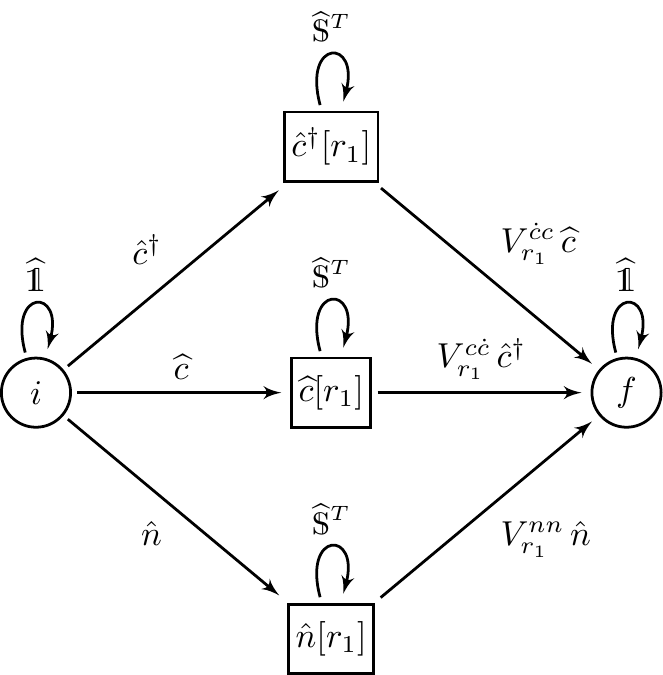}
    \includegraphics[width=\linewidth]{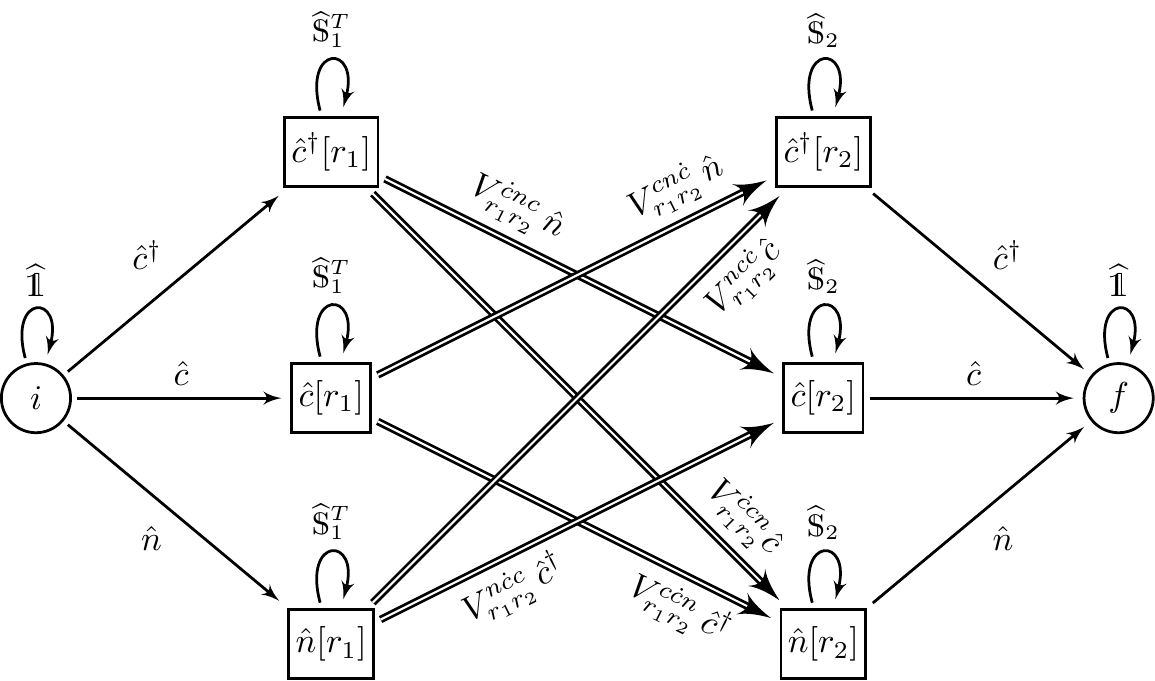}
    \includegraphics[width=\linewidth]{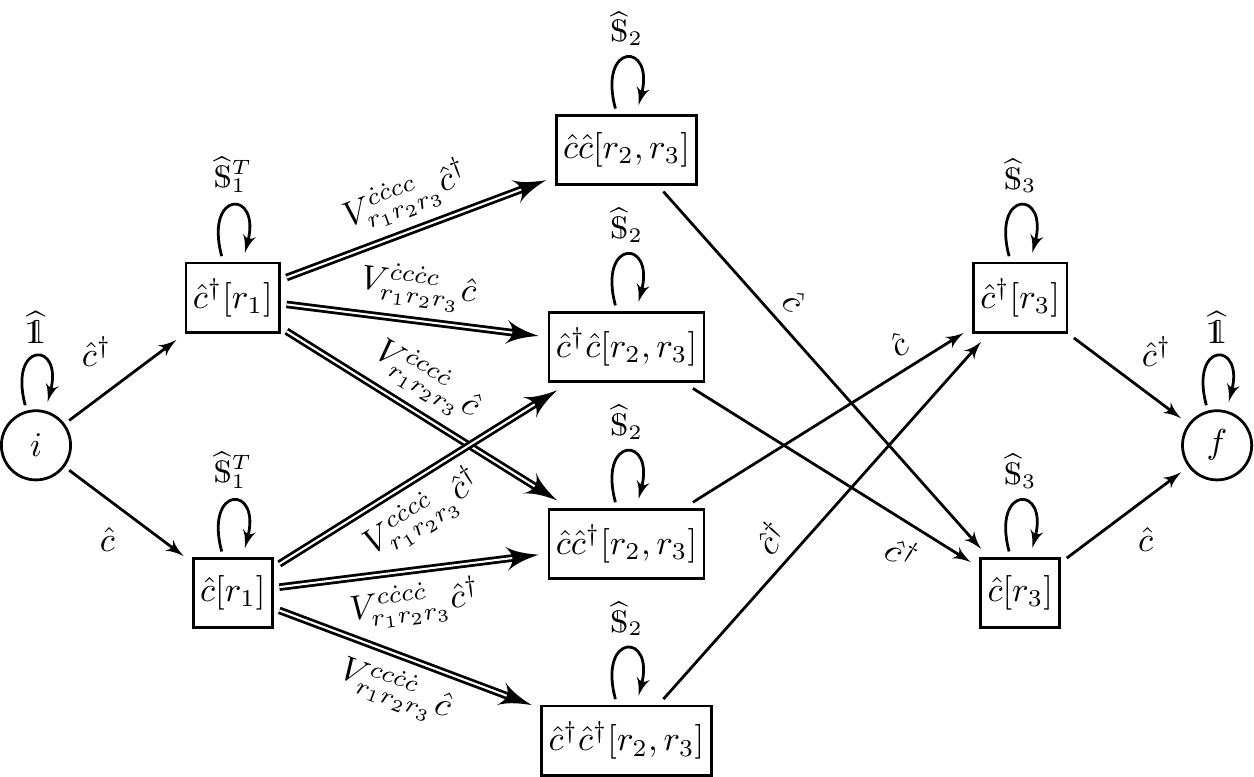}
	\caption{Finite state machines for constructing the 2-body, 3-body, and 4-body interactions of Eq.~\eqref{eq:full_H_234_bdy}. The notation is the same as in Fig.~\ref{fig:straightforward_4_bdy_MPO_notation}: boxes stand for collections of many nodes with $1 \le r_i \le R$. For concision, the operator $c^{\dagger}$ is represented by $\dot{c}$ in superscripts. Note that some of the nodes can be reused among $2$, $3$, and $4$-body paths. We have suppressed the $r_0$ index for clarity on the initial node and for the $V$ coefficients; this shows the case of unit cell of size $N=1$.}
    \label{fig:4-body-efficient-real}
\end{figure}

\section{Compression of infinite MPOs with general unit cells}
\label{app:UCMPO_compression}

This Appendix will present an algorithm for compressing infinite Matrix Product Operators (iMPOs) with non-trivial unit cells, a generalization of Ref. \cite{parker2019local}. The main application for these algorithms to compress the Hamiltonian for bilayer graphene down to a sufficiently small bond dimension that DMRG may be performed easily, while retaining sufficient precision to determine its physical properties. However, as the technique may be of independent interest, we will keep our discussion sufficiently general that the results apply to any iMPO that shows up in 2D DMRG.

Let us briefly review the context in which this algorithm is useful. Suppose that you have a Hamiltonian $\widehat{H}$ for a 2D system, and wish to find the ground state with DMRG \cite{white1992density,liang1994approximate}. A standard technique is to use a `thin cylinder' geometry of circumference $N_y$ sites and length $N_x \to \infty$. One then chooses a linear (1D) ordering for the sites on the cylinder by wrapping around in a helical pattern. This effectively reduces the problem to a 1D chain, but at a cost: interactions at distance $r$ in 2D can be as far as $N_y \times r$ in 1D. Furthermore, the resource cost grows hugely, as the matrix product state (MPS) bond dimension needed to accurately capture a 2D area law state grows as $\chiMPS \sim e^{N_y}$ \cite{liang1994approximate}. In practice, therefore, 2D DMRG is often limited to around $N_y=6-12$, even with bond dimensions of $\chiMPS \sim 10,000$ or more. For sufficiently long-range interactions in 2D, however, the bottleneck is not the MPS bond dimension, but rather the MPO bond dimension needed to encode the Hamiltonian. For example, long-range 4-body interactions of range $R$ result in iMPOs of bond dimension $D \sim R^2$ (see App.~\ref{app:uncompressed_MPO_construction}), and hence DMRG scales as $D^2 \sim R^4$, which becomes quickly impractical. The algorithms given below allow one to proceed by finding the best approximation for the iMPO of bond dimension $D' < D$. For many physical Hamiltonians, this compression incurs only a minor penalty (say, $10^{-4}$) in the precision of the eventual ground state. In the case of the single-valley IBM model we have used in this paper, the bond dimension may be reduced by a factor of $10^3$, vastly improving the speed of DMRG.

The rest of this Appendix is organized as follows. We first give an overview of iMPO compression in the case without unit cells to set notation. We then discuss how an iMPO may be ``topologically ordered", how this vastly speeds up compression, and present a practical compression algorithm. Afterwards we examine the properties of compressed Hamiltonians: under reasonable assumptions (1) Hermiticity is retained, (2) the fidelity of the compressed ground state versus the true ground state is high, and (3) their ground state energy and expectations values are accurate.

	\subsection{Lightning Review of iMPO Compression}

	We briefly review the notion of iMPO compression from \cite{parker2019local}. It is well-known that the optimal way to compress a 1D matrix product \textit{state} is to perform a Schmidt decomposition and drop the smallest singular values (see \cite{schollwock2011density} or \cite{hauschild2018efficient} for a review). For iMPOs, we employ the same basic technique with a few modifications to preserve the locality of the operator.

	We first present the algorithm in terms of (non-matrix product) local operators. Consider a local operator $\widehat{H}$ on an infinite 1D chain. We can split $\widehat{H}$ into left and right halves at some bond:
	\begin{equation}
		\widehat{H} = \widehat{H}_L \1_R + \1_L \widehat{H}_R + \sum_{a,b=1}^D  \widehat{h}_{L}^a \mathsf{M}_{ab} \widehat{h}_R^b
		\label{eq:H_bipartite_form}
	\end{equation}
	where a subscript $L$ or $R$  means the operator is supported entirely on the left or right half, respectively, and the matrix $\mathsf{M}$ keeps track of the terms which straddle the cut. This decomposition is not unique; we have the freedom to apply arbitary unitaries on the left and right. We can take advantage of this freedom to put $\widehat{H}$ into \textbf{almost-Schmidt form}
	\begin{equation}
		\widehat{H} = \widehat{H}_L \1_R + \1_L \widehat{H}_R + \sum_{a=1}^\chiMPO  \widehat{O}_{L}^a s_{a} \widehat{O}_R^a
		\label{eq:H_almost_Schmid_form}	
	\end{equation}
	where both $\{\widehat{O}_L^a\}$ and $\{\widehat{O}_R^a\}$ are orthonormal collections under the scaled Frobenius norm
	\begin{equation}
	    \IP{A,B} := \Tr[\widehat{A}^\dagger \widehat{B}]/\Tr[I],
	    \label{eq:scaled_Frobenius_norm}
	\end{equation}
	and where $s_1 \ge s_2 \ge \cdots \ge s_\chiMPO \ge 0$ are referred to as the singular values.\footnote{We note that the singular values resulting from an almost-Schmidt decomposition and a true Schmidt decomposition are slightly different. Their relation is described in Section 8 of \cite{parker2019local}.}
	We further require that the operators $\widehat{O}_{L, R}^a$ are \textbf{identity free}, i.e. $\langle \1_{L, R}, \widehat{O}_{L, R}^a\rangle = 0 \; \forall a$.  The \textbf{compressed operator} is then simply the truncation of the sum to the largest $\chiMPO' < \chiMPO$ singular values
	\begin{equation}
		\widehat{H}' := \widehat{H}_L \1_R + \1_L \widehat{H}_R + \sum_{a=1}^{\chiMPO'}  \widehat{O}_{L}^a s_{a} \widehat{O}_R^a.
		\label{eq:H_compressed}	
	\end{equation}
	One can show that $\widehat{H}'$ is the best approximation to $\widehat{H}$ with only $\chiMPO'$ `terms' straddling the cut \cite{parker2019local}. Furthermore, as we shall see below, the accuracy of the approximation is controlled by the \textbf{truncated singular value weight}
	\begin{equation}
	\epsilon^2(\chiMPO') := \sum_{a=\chiMPO'+1}^{\chiMPO} s_a^2.
		\label{eq:truncated_singular_value_weight}
	\end{equation}

	To \textit{compute} almost-Schmidt forms and compress operators, we work in the framework of matrix-product operators. We now recall their definitions to set notation and a few essential properties. Suppose we have a space of on-site operators with an orthonormal basis $\{\1, \widehat{O}_2,\ldots,\widehat{O}_d\}$ where $\IP{\widehat{O}_\alpha,\widehat{O}_\beta}= \delta_{\alpha\beta}$
	 is an inner product.\footnote{For example, a spin-$\frac{1}{2}$ chain has an on-site basis of the Pauli matrices $\{\1, \widehat{X},\widehat{Y},\widehat{Z}\}$.} Any translation-invariant operator, with unit cell of size $N$, can be written as\footnote{This is notation for multiplication along the virtual indices: $[\widehat{W}^{(1)} \widehat{W}^{(2)}]_{ac} := \sum_{b=1}^{\chiMPO^{(2)}} \left[W^{(1)}\right]^\alpha_{ab} \left[W^{(2)}\right]^\beta_{bc} \widehat{O}_\alpha \otimes \widehat{O}_\beta$.}
	\begin{equation}
		\widehat{H} = \cdots \left[\W^{(1)} \W^{(2)} \cdots \W^{(N)}\right]  \left[\W^{(1)} \cdots \W^{(N)}\right]\cdots	
		\label{eq:H_UCMPO_form}
	\end{equation}
	where each $\W^{(n)} = \sum_{\alpha=1}^d [W^{(n)}]^\alpha \widehat{O}_\alpha$ is an operator-valued matrix of size $\chiMPO^{(n)} \times \chiMPO^{(n+1)}$. We require each $\widehat{W}^{(n)}$ to have blocks of size $(1,\chiMPO^{(n)}-2,1) \times (1,\chiMPO^{(n+1)}-2,1)$
	\begin{equation}
\W = 
\left(
\begin{array}{c|c|c}
\1 & \CC & \DD \\ \hline
& \widehat{A}  & \BB \\ \hline
 & & \1
\end{array}
\right)
	\label{eq:block_structure}
\end{equation}
This ensures that each operator is a sum of terms that are the identity far enough to the left or right --- a physical and mathematical necessity for a local operator.

	An MPO is said to be in \textbf{left-canonical form} if all but the last column of each $\W^{(n)}$ are mutually orthonormal:\footnote{More explicitly, if $\sum_{\alpha=1}^d \sum_{a =1}^{\chiMPO-1} [ W^{(n)} ]_{ab}^{\alpha*} [W^{(n)}]_{ac}^\alpha = \delta_{bc}$ for each $\W^{(n)}$.}
	\begin{equation}
		\sum_{a=1}^{\chiMPO-1} \IP{\W^{(n)}_{ab},\W^{(n)}_{ac}} = \delta_{bc}, \quad \forall n \in \Z/N\Z, 1 \le b,c \le D-1.
		\label{eq:left_canonical_form}
	\end{equation}
Similarly, an MPO is \textbf{right-canonical} if all the rows except the first are mutually orthonormal.	
	
	The representation \eqref{eq:H_UCMPO_form}  is not unique, which is a manifestation of gauge freedom. Two MPOs $\{\widehat{W}^{(n)}\}$ and $\{\widehat{W}^{(n)'}\}$ are \textbf{gauge equivalent} if there gauge exist matrices $\{G_n\}$ such that
	\begin{equation}
		G_{n-1} \widehat{W}^{(n)'} = \widehat{W}^{(n)} G_n, \quad n \in \Z/N\Z.
	\end{equation}
	So long as $\widehat{H}$ is sufficiently local, one can show\cite{parker2019local} there exist a gauge where the $\W^{(n)}$'s are left-canonical (and another gauge for right canonical).

	Now that we have set definitions, the next section describes the compression algorithm for unit cell MPOs.

	\subsection{The Unit Cell Compression Algorithm}
	
	The rough idea of the compression algorithm for unit cell MPOs $\{\W^{(n)}\}$ is as follows. 
	\begin{enumerate}
		\item Compute the right-canonical form $\{\W_R^{(n)}\}$.
		\item Find the gauge transform $\{G_n\}$ needed to transform to left-canonical form $\{\W_L^{(n)}\}$.
		\item Take the SVD decomposition of the gauge transformation matrix: $G_n~=~U_n S_n V^\dagger_n$ and absorb the unitaries into the $\W$'s. This realizes the almost-Schmidt decomposition of Eq.~\eqref{eq:H_almost_Schmid_form}.
		\item Truncate the number of singular values in $S_n$ from $\chiMPO^{(n)}$ to $\chiMPO^{(n)'}$ and correspondingly reduce the bond dimensions of the $\W$'s, producing the compressed Hamiltonian.
	\end{enumerate}
	The rest of this section is devoted to showing the correctness of this procedure and filling in the details. It turns out that the most subtle part by far is canonicalization --- the algorithm for putting an MPO into left/right canonical form. We therefore delay the discussion of canonicalization to Appendix \ref{sec:UCMPO_canonicalization} below and for now simply assume it can be done.

We specialize to the case of $N=2$ sites in the unit cell for concision, as larger unit cells are a direct generalization. Suppose that 
\begin{equation}
	R_{n-1} \W_R^{n} = \W^{(n)} R_n
	\label{eq:left_can}
\end{equation}
is a gauge transformation so that the $\W_R^{(n)}$'s are right-canonical. Then
\begin{align}
	\widehat{H} 
	&= \cdots \Wa \Wb \Wa \Wb \cdots\\
	&= \cdots \Wa_R \Wb_R \Wa_R \Wb_R \cdots\\
\end{align}
by introducing $R_2$ at $\infty$ and sweeping to the right. We can then impose a further gauge transformation to make the first row of each $\W_R^{(n)}$ simultaneously identity-free. This is done by 
\begin{equation}
	R_n' := \begin{pmatrix}
		1 & \v{t}_n & 0\\
		0 & I & 0\\
		0 & 0 & 1
	\end{pmatrix},
\end{equation}
where the $1\times (\chiMPO^{(n)}-2)$-dimensional vectors $\v{t}_n$  are chosen such that $0 = \v{c}_0^{(n)} + \v{t}_n - \v{t}_{n+1} A_0^{(n)}$ where $A_0^{(n)}$ and $\v{c}_0^{(n)}$ are the $\1$-components of $\AA^{(n)}$ and $\v{c}^{(n)}$ respectively.
\begin{equation}
\begin{aligned}
	&\begin{pmatrix}
		\v{c}_0^{(1)} & \cdots & \v{c}_0^{(N)}
	\end{pmatrix}
	=\\
	&\begin{pmatrix}
	\v{t}_1 & \cdots & \v{t}_N	
\end{pmatrix}
\begin{pmatrix}
	I & & \cdots & -A_0^{(n)}\\
	-A_0^{(1)} & I & & &\\
	& \ddots & \ddots\\
	& & -A_0^{(n-1)} & I
\end{pmatrix}.
\end{aligned}
\label{eq:first_row_id_free_condition}
\end{equation}
In practice, one should solve this by imposing the identity free condition column-by-column. For each new column, this requires solving a linear equation of size $N$. Using the same technique from Eq.~(71) of \cite{parker2019local}, the total operation can be performed in $O(N \chiMPO^2)$ operations.

Imposing this gauge we may assume $\W_R^{(n)}$ has no identity components in its first row. This implies that the first column of $\W_R^{(n)}$ is already orthogonal to all the other columns, such that the gauge transformation to the left-canonical form can be written as
\begin{equation}
	C_{n-1} \W_R^{(n)} = \W_L^{(n)} C_n\,,
	\label{eq:center_gauge_transform}
\end{equation}
with block-diagonal gauge transformation matrices $C_n = \diag(1 \; \mathsf{C}_n \; 1)$. Putting in a $C_2$ matrix at $-\infty$ and sweeping it to the center, we arrive at a mixed canonical form
\begin{align}
	\widehat{H} &= \cdots \Wa_L \Wb_L C_2 \Wa_R \Wb_R \cdots\\
	 &= \cdots \Wb_L \Wa_L C_1 \Wb_R \Wa_R \cdots
\end{align}
As $C_n$ are block diagonal, we can compute their SVD's
\begin{equation}
	C_n = U_n S_n V_n^\dagger 
	\label{eq:SVD_block_diagonal}
\end{equation}
which will also be block-diagonal. Define
\begin{align}
	\Q^{(n)} &:= U_{n-1}^\dagger \W_L^{(n)} U_n\\
	\P^{(n)} &:= V_{n-1}^\dagger \W_R^{(n)} V_n
\end{align}
for $n \in \Z/N\Z$. Then, since $U_n U_n^\dagger = I = V_n V_n^\dagger$, we have
\begin{align}
	\widehat{H} &= \cdots \Qa \Qb S_2 \Pa \Pb \cdots\\
	&= \cdots \Qb \Qa S_1 \Pb \Pa \cdots
\end{align}
which is analogous to center-canonical form for MPS. The center bond can be swept back and forth via the gauge relation
\begin{equation}
	S_{n-1} \P^{(n)} = \Q^{(n)} S_n, \quad n \in \Z/N\Z.
	\label{eq:center_canonical_gauge_condition}
\end{equation}

\begin{figure}
\begin{algorithm}[H]
	\caption{Unit Cell iMPO Compression}
	\label{alg:iMPO_compression}
	\setstretch{1.35}
	\begin{algorithmic}[1]
		\Require{$\{\W^{(n)}\}$ is a first-order (see \cite{parker2019local}) unit cell iMPO.}
		\Procedure{UnitCellCompress}{$\W^{(n)},\eta$} \Comment{Cutoff $\eta$}
		\State $\W_R^{(n)} \gets \textsc{RightCan}[\W^{(n)}]$ 
		\State $\W_R^{(n)} \gets R_{n-1}' \W_R^{(n)} R_n^{'-1}$ so that $\CC_0^{(n)} = 0$  \Comment{Use $\v{t_n}$ from Eq.~\eqref{eq:first_row_id_free_condition}.}
		\State $\W_L^{(n)}, C_n \gets
		\textsc{LeftCan}[\W_R^{(n)}] $
		\State $(U_n, S_n, V_n^\dagger) \gets \textsc{SVD}[C_n]$
		\State $\Q^{(n)}, \P^{(n)} \gets U_{n-1}^{\dagger} \W_L^{(n)} U_n \,,\, V_{n-1}^{\dagger} \W_R  V $
		\State $\chiMPO^{(n)} \gets  \max_a \{1 \le a \le \chiMPO^{(n)}]: S^{(n)}_{aa} > \eta \}$ \Comment{Defines the projector $\PP_n$.}
		\State{$\Q^{(n)} \gets \PP_{n-1}^\dagger \Q_n  \PP_n$}
		\State{$S_n\hspace{0.8em} \gets \PP_{n}^\dagger S^{(n)} \PP_n$}
		\State{$\P^{(n)} \gets \PP_{n-1}^\dagger \P^{(n)}  \PP_n$}
		\State \textbf{return} $\Q^{(n)}$
		\EndProcedure	
	\end{algorithmic}
\end{algorithm}
\end{figure}

To see how compression works, we adopt the technique of assuming that the operator can be represented \textit{exactly} by an MPO of lower bond-dimension, i.e. that a number of the singular values vanish exactly. Finding the lower bond dimension MPO uses the same algorithm as compression when the small singular values are truncated, so this shows the correctness of the algorithm. 

Thus we assume, temporarily, that only $\chiMPO^{(n)'}$ of the $\chiMPO^{(n)}$ singular values of $\mathsf{S}_n$ are non-zero. Hence there are projection operators $\PP_n$ from bond dimension $\chiMPO^{(n)}$ to bond dimension $\chiMPO^{(n)'}$ with $\PP_n \PP_n^\dagger$ a projector and $\PP_n^\dagger \PP_n = I_{1+\chi^{(n)'}+1}$ and
\begin{align}
	S_n = S_n \PP_n \PP_n^\dagger = \PP_n S_n' \PP_n^\dagger = \PP_n \PP_n^\dagger S_N
	\label{eq:projection_of_singular_values}
\end{align}
where $S_n'$ is the projected diagonal matrix of non-zero singular values. 

We can then introduce pairs of projectors on each bond:
\begin{align}
	\widehat{H}
	&= \cdots \Qa \Qb S_2 \PP_2 \PP_2^\dagger \Pa \Pb \cdots\\
	\nonumber
	&= \cdots \Qa S_1 \Pb \PP_2 \PP_2^\dagger \Pa \Pb \cdots\\
	\nonumber
	&= \cdots \Qa S_1 \PP_1 \PP_1^\dagger \Pb \PP_2 \PP_2^\dagger \Pa \Pb \cdots\\
	\nonumber
	&= \cdots \PP_2^\dagger \Qa \PP_1 \PP_1^\dagger \Pb \PP_2 S_2' \PP_2^\dagger \Pa \PP_1 \PP_1^\dagger \Pb \PP_2\cdots\\
	\nonumber
	&= \cdots \PP_1^\dagger \Qb \PP_2 \PP_2^\dagger \Pa \PP_1 S_1' \PP_1^\dagger \Pb \PP_2 \PP_2^\dagger \Pa \PP_1\cdots
\end{align}
It is now clear how to define a new representation for $\widehat{H}$ with a reduced bond dimension:
\begin{align}
	\P^{(n)'} := \PP_{n-1}^\dagger \P^{(n)} \PP_n\\
	\Q^{(n)'} := \PP_{n-1}^\dagger \Q^{(n)} \PP_n\\
\end{align}
whereupon
\begin{equation}
	\widehat{H} = \cdots \Q^{(1)'} \Q^{(2)'} S_2 \P^{(1)'} \P^{(2)'} \cdots
	\label{eq:compressed_H}
\end{equation}
is a representation of $\widehat{H}$ with lower bond dimension. If we now relax the requirement that the truncated singular values were exactly zero, the strict equality of the new representation becomes approximate. 

\subsection{Canonicalization \& Topological Sorting for Unit Cell MPOs}
\label{sec:UCMPO_canonicalization}

\begin{figure}
	\includegraphics{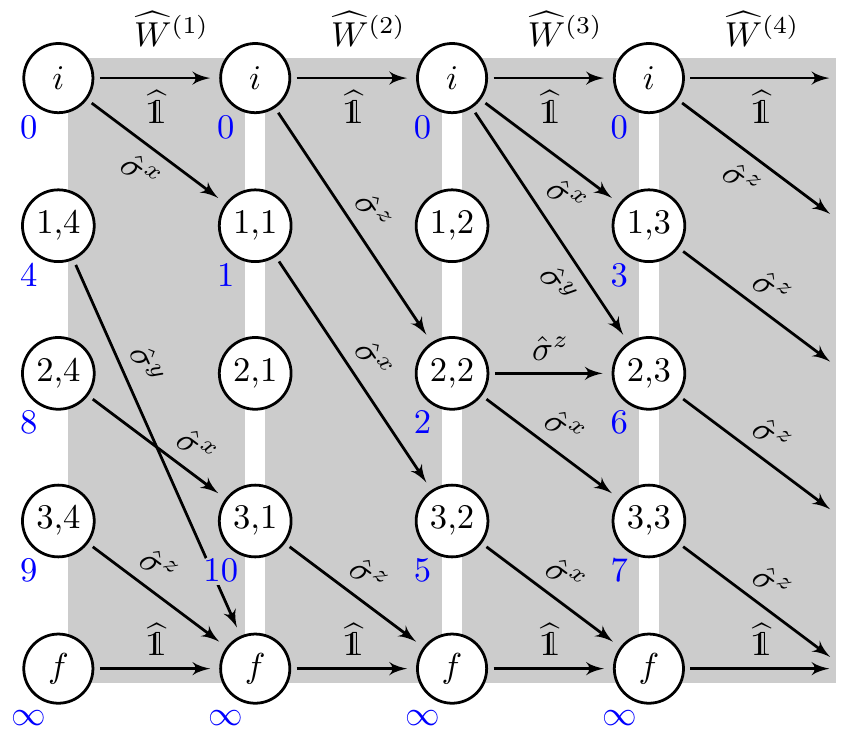}
	\caption{An example of the finite state machine for an MPO with unit cell size $4$. One on-site operator is placed for each arrow, and the arrows wrap around from right to left. Each gray box represents the data stored in one tensor. The UCMPO has bond dimension $(\chiMPO_4,\chiMPO_1,\chiMPO_2,\chiMPO_3) = (5,4,4,5)$ and is loop free. The blue numbers are a (non-unique) topological ordering for the nodes.}
	\label{fig:UCMPO_FSM_example}
\end{figure}

	In this section we provide the ``missing link'' needed to complete the compression procedure: a canonicalization algorithm. Any unit cell MPO (UCMPO) can be put into left or right canonical form using QR iteration~\cite{parker2019local} with cost $O(N \chiMPO^3)$. As many as $40$ iterations can be necessary to reach high precision, making this quite slow in practice. However, the MPOs for Hamiltonians in DMRG have a special property, a ``topological ordering'', which enables canonicalization to be performed with cost $O(N \chiMPO^3)$ but  \textit{without iteration}. For large MPOs such as the one for BLG with $\chiMPO\sim 100,000$, this is a crucial speed-up. We first define a ``topological ordering,'' then provide the canonicalization algorithm and a proof of its correctness and runtime. We conclude the section with a few remarks on practical implementation details.

	An MPO can be thought of as a finite state machine (FSM) for placing on-site operators in a certain order \cite{crosswhite2008finite}. For MPOs with $N$ tensors in a unit cell, the FSM gains an additional structure: the FSM has $N$ parts, with the nodes of part $n$ corresponding to the bond between $\W^{(n-1)}$ and $\W^{(n)}$ and edges between parts $n-1$ and $n$ corresponding to tensor elements $\W_{ab}^{(n)}$. See Fig.~\ref{fig:UCMPO_FSM_example} for an example.

	When one writes down an (non-unit cell) MPO $\W$ for a Hamiltonian ``by hand", then the MPO generally has a special structure: $\W$ is upper-triangular as a matrix. In Ref.~\cite{parker2019local}, the upper triangular structure was shown to permit a fast canonicalization algorithm. However, this does not immediately generalize to a unit cell MPO, for a simple reason:  if a unit cell MPO $\UCMPO$ has bond dimensions $\chiMPO^{(1)}, \chiMPO^{(2)} \dots \chiMPO^{(n)}$, not all equal, then the matrices are rectangular and cannot all be upper triangular. To find a good generalization of triangularity, we must look to the finite state machine.

	A UCMPO $\UCMPO_{n=1}^N$ is said to be \textbf{loop free} if its finite state machine contains no loops \textit{after} the initial and final nodes are removed from the graph.  For $N=1$, then an upper-triangular MPO is always loops free, and a loop-free MPO is always upper-triangular (up to permutation). We stress that both the upper triangular and loop free conditions are gauge-\textit{dependent}. Furthermore, for any $N$, if each $\W^{(n)}$ is square and upper-triangular, then the UCMPO is loop free. The converse is almost true as well; any loop free UCMPO is an upper-triangular MPO ``in disguise''. To see this, we need a definition, which will be at the heart of this section.

	\begin{defn}
		A \textbf{topological ordering} for a UCMPO $\UCMPO_{n=1}^N$ is an ordering of the nodes of the FSM (excluding the initial and final nodes)
		\begin{equation}
			O = \st{(a_1,n_1) \prec (a_2,n_2) \prec \cdots \prec (a_\chiMPO,n_\chiMPO)}
			\label{eq:top_ordering}
		\end{equation}
		such that
		\begin{equation}
			\W^{(n)}_{ab} = 0 \text{ whenever } (a,n) \succeq (b,n+1),	
			\label{eq:top_ordering_condition}
		\end{equation}
		where $n \in \Z/N\Z$ indexes the bonds (explicitly, bond $n+1$ connects $\W^{(n)}$ to $\W^{(n+1)}$), $a \in \N$ indexes the node within the bond, and $\chiMPO = \sum_n \chiMPO^{(n)}$ is the total number of nodes. 
	\end{defn}
If an MPO is loop free, then its finite state machine (excluding the initial and final nodes) is a directed acyclic graph, and thus contains at least one topological ordering. This is easily computed by Kahn's algorithm (a standard result in graph theory) with cost linear in the number of nodes plus edges in the FSM. With this, we can show that loop free UCMPOs are upper triangular ones ``in disguise'' and then use this ordering as the basis for an efficient canonicalization algorithm.

\begin{lemma}
	Suppse $\UCMPO_{n=1}^N$ is a loop free MPO. Then, by inserting rows and columns of zeros and permuting the rows and columns of the matrices (which is a gauge transform), $\UCMPO$ can be made upper triangular.
\end{lemma}
\begin{proof}
	Suppose the bond dimension of $\W^{(n)}$ is $\chiMPO^{(n)}$ on the left and $\chiMPO^{(n+1)}$ on the right, with $\chiMPO = \sum_n \chiMPO^{(n)}$. Let $O$ be a topological ordering for $\UCMPO$ of the form \eqref{eq:top_ordering}. Define a gauge matrix $\mathbb{P}_n$ of dimension $\chiMPO^{(n)}\times \chiMPO$ with matrix elements
	\begin{equation}
		\left[\mathbb{P}_{n}\right]_{b,i} 
		= \begin{cases}
			1 & \text{ if $(b,n+1) = O_i$}\\
			0 & \text{ otherwise.}
		  \end{cases}
	\end{equation}
	This ``blows up" $\W^{(n)}$ on the right to bond dimension $\chiMPO > \chiMPO^{(n+1)}$ by inserting zeros, and puts the indices into topological order. One can check that $\mathbb{P}_n^\dagger \mathbb{P}_n = I_{\chiMPO^{(n+1)}}$, so we may define $\W^{(n)'} := \mathbb{P}_{n-1}^{\dagger} \W^{(n)} \mathbb{P}_n$ of size $\chiMPO \times \chiMPO$ (which obeys $\mathbb{P}_{n-1} \W^{(n)'} = \W^{(n)} \mathbb{P}_n$, making it a gauge transformation). 
	
	The new MPO $\W^{(n)'}$ is upper-triangular. To see this, take $i \ge j$. Then either $\W^{(n)'}_{ij} = W_{ab}^{(n)}$ for $O_i = (a,n)$ and $O_j = (b,n+1)$, or $\W^{(n)'}_{ij} = 0$. But $O_i = (a,n) \succeq (b,n+1) = O_j$, so $W_{ab}^{(n)} = 0$ regardless. Therefore $\W^{(n)'}$ is upper triangular.
\end{proof}

We present the algorithm for (left) canonicalization of loop free UCMPOs in Alg \ref{alg:UCMPO_UT_canonicalization}. We now prove its correctness, then analyze its cost.

\begin{figure}
\begin{algorithm}[H]
	\caption{Unit Cell iMPO (Left) Canonicalization}
	\label{alg:UCMPO_UT_canonicalization}
	\setstretch{1.35}
	\begin{algorithmic}[1]
		\Require{$\{\W^{(n)}\}_{n=1}^N$ is a loop free UCMPO.}
		\Procedure{UnitCellLeftCanonical}{$\W^{(n)},\eta$}
		\State $O =$ \textsc{Kahn'sAlgorithm}[FSM[$\{\W^{(n)}\}$]]
		\For{$(b,n+1) \in O$}
		\State $P \gets \setc{a}{ O \ni (a,n+1) \prec (b,n+1)}$
		\State $r_a \gets \sum_c \IP{\W_{ca}, \W_{cb}},\quad \forall a \in P$
		\State $R \gets I_{\chiMPO_n}$, $R_{ab} \gets r_a, \quad \forall a \in P$
		\State $\W^{(n)} \gets \W^{(n)}R$, $\W^{(n+1)} \gets R^{-1} \W^{(n+1)}$
		\State $R \gets I_{\chiMPO_n}$, $R_{bb} \gets \left( \sum_c \IP{\W_{cb}^{(n)},\W_{cb}^{(n)}} \right)^{-1/2}$ 
		\State $\W^{(n)} \gets \W^{(n)}R$, $\W^{(n+1)} \gets R^{-1} \W^{(n+1)}$
		\EndFor
		\State \textbf{return} $\UCMPO$
		\EndProcedure
	\end{algorithmic}
\end{algorithm}
\end{figure}

\begin{prop}
	Suppose $\UCMPO_{n=1}^N$ is loop free. The output of Alg. \ref{alg:UCMPO_UT_canonicalization} is a left canonical UCMPO.
\end{prop}

\begin{proof}
The main idea is to iterate over the columns of $\UCMPO$ in topological order, orthogonalizing each column against all the previous ones as in Gram-Schmidt.

Let $O$ be a topological order for the nodes as in \eqref{eq:top_ordering}. As each $\W^{(n)}$ has the form \eqref{eq:block_structure}, the first column of each is already orthonormal. We proceed by induction. Suppose that we have orthogonalized columns up to $(d,n+1) \in O$. Then for $(a,m), (b,m) \prec (d,n+1)$,
\begin{equation}
	\sum_{c}	\IP{\W_{ca}^{(m)},\W_{cb}^{(m)}} = \delta_{ab}.
	\label{eq:induction_assumption}
\end{equation}
Let the \textit{predecessor} nodes be $P := \setc{(a,n+1)}{(a,n+1) \prec (d,n+1)}$ and for each $(a,n)~\in~P$, define the inner products with all previous columns as
\begin{subequations}
\begin{align}	
	r_a &:=  \sum_{c}	\IP{\W_{ca}^{(n)},\W_{cd}^{(n)}},\\
	R &:= I_{\chiMPO^{(n+1}} - \sum_{a \in P} r_a \v{e}_{ad}
\end{align}
\end{subequations}
where $\v{e}_{ad}$ is the elementary matrix where entry $ad$ is $1$ and the rest are zero: $(\v{e}_{ad})_{ij} = \delta_{ai} \delta_{dj}$. Here $R$ is only non-identity in column $d$, and it performs elementary column operations when acting to the left and elementary row operations acting to the right. In particular, we have chosen it to perform one Gram-Schmidt step, orthogonalizing column $d$ against previous columns of $\W^{(n)}$. It is easy to invert $R$,  $R^{-1} =  I_{\chiMPO^{(n)}} + \sum_{a \in P} r_a \v{e}_{ad}$, so we can cast this Gram-Schmidt step as a gauge transform with a single non-identity gauge matrix:
\begin{subequations}
	\begin{align}
		\W^{(n)'} &:= \W^{(n)} R,\\
		\W^{(n+1)'} &:= R^{-1} \W^{(n+1)}.
	\end{align}
	\label{eq:gram_schmidt_step_gauge}
\end{subequations} We then have two things to show: (i) that this gauge transform really does orthogonalize column $d$ of $\W^{(n)}$ against previous columns and (ii) that the gauge transform does not ruin the orthogonality condition of Eq.~\eqref{eq:induction_assumption}. Both are easy computations.

For (i), the effect of $R$ acting on $\W^{(n)}$ on the right is to add column $c$ to column $d$ with coefficient $r_c$:
\[
	\W^{(n)}R = \W^{(n)} - \sum_{c\in P} r_c \W^{(n)} \v{e}_{cd}. 
\]
The matrix $\W^{(n)} \v{e}_{cd}$ is the matrix with only column $d$ non-zero, and whose values are those from column $c$ of $\W^{(n)}$. Let $\W_{:,a}$ denote the $a$th column vector of $\W$, as usual. Then
\[
	\IP{\W^{(n)}_{:,a}, [\W^{(n)} \v{e}_{cd}]_{:,d}} = \IP{\W^{(n)}_{:,a}, \W^{(n)}_{:,c}} = \delta_{ac}
\]
by \eqref{eq:induction_assumption}. Therefore, for any $(a,n) \prec (d,n)$,
\begin{align*}
	&\IP{\W^{(n)'}_{:,a},\W^{(n)'}_{:,d}}
	\\
	&\hspace{1em} =\IP{\W^{(n)}_{:a},\W^{(n)}_{:d}}
	-\sum_{c \in P} r_c  \IP{\W^{(n)}_{:,a},[\W^{(n)} \v{e}_{cd}]_{:,d}}\\
	&\hspace{1em} =r_a - \sum_{c \in P} r_c \delta_{ca} = 0.
\end{align*} Therefore column $d$ of $\W^{(n)}$ is orthogonal to each previous column.

For (ii), the effect of $R^{-1}$ acting to the left on $\W^{(n+1)}$ is to add row $d$ to row $c$ with coefficient $r_c$: 
\[
	\W^{(n+1)'} = \W^{(n+1)} + \sum_{c \in P} r_c \sum_e \left(\W^{(n+1)}_{de}\right) \v{e}_{ce} =: \W + \delta\W.
\]
Take $(a,n+1), (b,n+1) \prec (d,n)$. Then
\[
	[\delta\W]_{:,a} = \sum_{c \in P} r_c \left( \W^{(n+1)}_{da} \right) \v{e}_{ca} = 0 
\]
since $\W_{da}^{(n+1)} = 0$ as $(d,n) \succeq (a,n+1)$ by \eqref{eq:top_ordering_condition}, and similarly $[\delta\W]_{:,b} = 0$.
Therefore,
\begin{align*}
	\IP{\W^{(n+1)'}_{:,a}, \W^{(n+1)'}_{:,b}} 
	\ &=\ \IP{\W_{:,a}+\delta \W_{:,a}, \W_{:,b} + \delta\W_{:,b}}\\
	\ &=\ \IP{\W_{:,a},\W_{:,b}} + 0 = \delta_{ab},
\end{align*}
so the induction hypothesis \eqref{eq:induction_assumption} holds for $\{\W^{(n)'}\}$.

As the gauge transform $R$ adds previous columns to column $d$ of $\W^{(n)}$ and adds row $d$ of $\W^{(n+1)}$ to previous rows, the transformed UCMPO is also loop free. Thus after this gauge transform, column $d$ of $\W^{(n)'}$ is orthogonal to all previous columns and all of the structure of the UCMPO is preserved.

A similar, simpler gauge transform
\[
	R = I_{\chiMPO^{(n+1)}} + \Big( \sum_c \IP{\W_{cd}^{(n)'},\W_{cd}^{(n)'}} \Big)^{-1/2} \v{e}_{dd}
\]
can then be used to normalize column $d$ of $\W^{(n)'}$. Repeating the previous arguments, one can show that this similarly does not disrupt the orthogonality of $\W^{(n+1)'}$ or the loop free condition. Therefore we have made one more column orthonormal to the previous ones, completing the proof.
\end{proof}

Algorithm \ref{alg:UCMPO_UT_canonicalization} is quite efficient, with cost that scales as $O(\sum_n [\chiMPO^{(n)}]^3)$. This is somewhat surprising, as it seems we are doing a total of $\chiMPO = \sum_n \chiMPO_n$ gauge transformations, each of which is a matrix multiplication. However, the $R$ matrices are particularly simple: they only differ from the identity in a single column. The transformations $\W^{(n)'} = \W^{(n)}R$ and $\W^{(n+1)'} = R \W^{(n+1)}$ to orthogonalize a column may be performed with rank-$1$ matrix updates whose cost is only $O([\chiMPO^{(n)}]^2)$. Similarly, the gauge transform to normalize a column, which simply scales a row or column, costs only $O(\chiMPO^{(n)})$. As we must iterate over every column of every tensor, the total cost is then $O(\sum_n [\chiMPO^{(n)}]^3)$. However, each iteration requires only elementary matrix operations, for which highly optimized libraries are available, and a low constant factor on the algorithm. One can also employ these algorithms with charge-conserving MPOs, which vastly decreases the runtime in practice.

\subsection{Properties of Compressed Hamiltonians}
\label{sec:properties_of_compressed_Hamiltonians}

We now show that compressed Hamiltonians are accurate approximations to the original Hamiltonian. This will give us guarantees that the (ground state) physics we are interested in is unchanged by compression. In fact, just as with matrix product \textit{states}, the error is controlled by the weight of the truncated singular values. We demonstrate three properties of the compressed Hamiltonian $H'$ \textit{when we truncate a single bond}:

\begin{equation}
	\widehat{H}_{\mathrm{BLG}} \to \widehat{H}'(\epsilon)
\end{equation}
where $\widehat{H}'$ satisfies the following:
\begin{enumerate}
	\item $\widehat{H}'$ is Hermitian
	\item The ground state energy is accurate: $\delta E \le 4^2\, \epsilon$ 
	\item Observables are accurate: $\Delta \braket{\widehat{O}} \le \frac{4^3}{\Delta E} \dn{\widehat{O}} \, \epsilon$
	\item The ground state wavefunction is accurate: $\n{1 - \mathcal{F}} \le \frac{4^4}{\Delta E^2}\, \epsilon^2$,
\end{enumerate}

We reiterate that these are \textit{local} bounds, corresponding to truncating a single bond. It is reasonable to expect that these results can be generalized to \textit{global} bounds which apply when all bonds are truncated simulataneously, just as they can for matrix product states. However, such generalizations are often highly technical and therefore beyond the scope of this work. As a practical matter, Fig.~\ref{fig:MPOfidelity} demonstrates that the global errors in the ground state energy, fidelity, and expectation values are small and decrease as $\epsilon(D)\to 0$.

	\subsubsection{Compressed Hamiltonians are Hermitian}

All Hamiltonians in quantum mechanics are Hermitian. We now show that Hermiticity is preserved by dropping singular values of an Hamiltonian. There is just one caveat: if the spectrum contains a set of degenerate singular values, then one must drop either all of them or none of them: 

\begin{prop}
    Let $\widehat{H}$ be a Hermitian operator with the following almost-Schmidt form:
    \begin{equation}
		\widehat{H} = \widehat{H}_L \1_R + \1_L \widehat{H}_R + \sum_{a=1}^{N_s} \sum_{i=1}^{\chiMPO_a}  \widehat{O}_{L}^{a, i} s_{a} \widehat{O}_R^{a, i}
    \end{equation}
    where $a$ labels degenerate singular values, $N_s$ is the number of distinct singular values,  and $\chiMPO_a$ is the degeneracy of the $a$th Schmidt value. Then, for any subset $\mathcal{A} \subset \{1, 2,...,N_s\}$, the compressed operator
    \begin{equation}
		\widehat{H}' = \widehat{H}_L \1_R + \1_L \widehat{H}_R + \sum_{a \in \mathcal{A}} \sum_{i=1}^{\chiMPO_a}  \widehat{O}_{L}^{a, i} s_{a} \widehat{O}_R^{a, i}
    \end{equation}
    is Hermitian.
\end{prop}

For concision, we sketch the proof. Due to the orthonormality condition, $\widehat{H}_L \1_R$, $\1_L \widehat{H}_R$, and $\sum_{a=1}^{N_s} \sum_i^{\chiMPO_a}  \widehat{O}_{L}^{a, i} s_{a} \widehat{O}_R^{a, i}$ must be independently Hermitian. This implies $\sum_{a=1}^{N_s} \sum_i^{\chiMPO_a}  \widehat{O}_{L}^{a, i} s_{a} \widehat{O}_R^{a, i} = \sum_{a=1}^{N_s} \sum_i^{\chiMPO_a}  (\widehat{O}_{L}^{a, i})^\dagger s_{a} (\widehat{O}_R^{a, i})^\dagger$. Note that both the LHS and the RHS of this equation can be regarded as a singular value decomposition in operator space. It follows from the uniqueness of singular value decomposition that $\sum_i^{\chiMPO_a}  \widehat{O}_{L}^{a, i} s_{a} \widehat{O}_R^{a, i} =  \sum_i^{\chiMPO_a}  (\widehat{O}_{L}^{a, i})^\dagger s_{a} (\widehat{O}_R^{a, i})^\dagger$ for any $a$. The proposition follows.

In practice, this means that one should always drop singular values by imposing a minimum value to retain rather than a maximum number.

We note that this result readily generalizes to a case when we truncate all bonds at the same time. To see this, we note that the proof above shows the action of Hermitian conjugation commutes with the singular value matrix. Then, just like in the case of symmetric MPS, singular values can be dropped without ruining Hermiticity.

	\subsubsection{Compressed Hamiltonians are Accurate}

	The accuracy of a compressed Hamiltonian is controlled by the weight of the truncated singular values in almost-Schmidt form, Eq.~\eqref{eq:truncated_singular_value_weight}. Conceptually, one should think of truncation as introducing a small perturbation to the Hamiltonian. If the truncated weight is small, then the perturbation is small, and its effects to the ground state energy, the fidelity, and other observables are also small. 

	To quantify these effects, we employ the sup norm, an operator norm well-suited for ground state properties. If $\widehat{H}$ is an operator, then its \textbf{sup norm} $\dn{\widehat{H}}$ is given by
\begin{equation}
	\dn{\widehat{H}}^2 := \sup_{\ket{\psi}} \frac{\braket{\psi|\widehat{H}\widehat{H}|\psi}}{\braket{\psi|\psi}}.
	\label{eq:sup_norm}
\end{equation}
As the sup norm is extensive, we work with the sup norm per unit cell, so that it is finite.

	We first quote a result from \cite{parker2019local}: the change in the ground state energy is small under truncation. 
	\begin{prop}[Prop. 5 of \cite{parker2019local}]
	Suppose $\widehat{H}$ is a $k$-body Hamiltonian with on-site dimension $d$.\footnote{e.g.~$d=4$ for spin-$\frac{1}{2}$'s or spinless fermions.} If $\widehat{H}$ is compressed from bond dimension $\chiMPO$ to $\widehat{H}'$ with $\chiMPO'$ with truncated weight
	\begin{equation}
		\epsilon^2 := \sum_{a=\chiMPO+1}^{\chiMPO'} s_a^2,
		\label{eq:epsilon_truncated_weight}
	\end{equation}
	Then the change in the ground state energy is bounded by
	\begin{equation}
		\delta E \le \dn{\widehat{H}-\widehat{H}'} \le d^{\frac{k}{2}} \epsilon.
		\label{eq:GS_energy_bound}
	\end{equation}
	\end{prop}
	In practice, the singular values for an Hamiltonian fall off quite quickly --- often exponentially, or as a power law at worse. So retaining only a small number of singular values can produce a highly accurate approximation for the ground state energy. 

It is natural to assume that if the ground state energy is accurate, then the other ground state properties --- such as expectation values of observables and even the entire ground state wavefunction --- are accurate as well. Unfortunately, there is a rare but severe failure of this assumption. Near a first-order phase transition, a tiny perturbation to a Hamiltonian can push the system across the phase transition, changing the properties of the ground state in a discontinuous manner (except for the energy). However, as long as the competing states have large energy difference away from the transition, this will only cause infinitesimal shift of critical parameters. We can therefore understand the generic case by simply assuming we are far from a phase transition and the ground state changes continuously.

To do this, we work in first order perturbation theory. Suppose $\widehat{H}$ is a $k$-body Hamiltonian with a unique ground state with gap $\Delta E$. Suppose we write $\widehat{H} = \widehat{H}' + \delta \widehat{H}$ with truncated weight $\epsilon^2$ as in \eqref{eq:epsilon_truncated_weight}, and consider an observable of interest $\widehat{O}$. Then we can write the new ground state as
\begin{align*}
	\ket{E_0(\delta)'} 
	&= \ket{E_0} + \ket{\delta E_0} + O(\epsilon^2),\\
	\ket{\delta E_0} 
	&= \sum_{\lambda \neq 0} \frac{\braket{E_\lambda|\delta \widehat{H}|E_0}}{E_\lambda - E_0} \ket{E_\lambda}.
\end{align*}
Then
\begin{align*}
	\Delta O &:= \n{\braket{E_0(\delta)'|\widehat{O}|E_0(\delta)'} - \braket{E_0|\widehat{O}|E_0}}\\
	& \; = 2 \n{\mathrm{Re} \braket{E_0|\widehat{O}| \delta E_0}} + O(\epsilon^2), 
\end{align*}
so
\begin{align*}
	&\n{\braket{E_0|\widehat{O}|\delta E_0}}\\
	& \hspace{1em} \le \n{\sum_{\lambda\neq 0} \frac{\braket{E_0|\delta\widehat{H}|E_\lambda}\braket{E_\lambda|\widehat{O}|E_0}}{E_\lambda-E_0}}\\
	& \hspace{1em}\le \frac{1}{\Delta E}  \n{\sum_{\lambda\neq 0} \braket{E_0|\widehat{O}|E_\lambda}\braket{E_\lambda|\delta\widehat{H}|E_0}}\\
	& \hspace{1em} \le \frac{1}{\Delta E}  \n{
		\braket{E_0|\widehat{O}\, \delta \widehat{H}|E_0}
		- \braket{E_0|\widehat{O}|E_0} \braket{E_0|\delta \widehat{H}|E_0}
}\\
& \hspace{1em} \le \frac{2}{\Delta E} \dn{\widehat{O}}  \cdot \dn{\delta \widehat{H}}
\end{align*}
where we have used $\sum_{\lambda\neq 0} \ket{E_\lambda} \bra{E_\lambda} = I - \ket{E_0}\bra{E_0}$ and submultiplicativity of the norm. Using \eqref{eq:GS_energy_bound}, the change in the expectation value is bounded by
\begin{equation}
	\Delta O \le \frac{4 d^{\frac{k}{2}}}{\Delta E} \dn{\widehat{O}} \; \epsilon.
	\label{eq:expectation_value_bound}
\end{equation}
We may therefore conclude that the error in expectation values should be small, provided that the uncompressed Hamiltonian is sufficiently far from a first-order phase transition. The condition of a gapped ground state may be relaxed, in which case the error will be controlled by the matrix elements of $\widehat{O}$ between the ground state and low-lying excited states.

	\subsubsection{Compressed Hamiltonians have High Fidelity}

	We have now seen that the ground state energy and expectation values of observables are accurately captured by the approximate, compressed Hamiltonian. In fact, the entire ground state wavefunction $\ket{E_0'}$ of $\widehat{H}'$ is very close to the original ground state wavefunction $\ket{E_0}$ of $\widehat{H}$. This allows us to use structural properties of $\ket{\psi'}$, such as its correlation length as a function of MPS bond dimension, as an accurate stand-in for the true ones and use them to e.g.~diagnose the scaling properties of phase transitions.

To see this, we again work in perturbation theory, this time to second order. Let $\widehat{H} = \widehat{H}' + \delta \widehat{H}$ and take the same assumptions as above. Then we write
\[
	\ket{E_0(\delta)'} = \ket{E_0} + \ket{\delta E_0} + \ket{\delta^2 E_0}  + O (\epsilon^3). 
\]
so
\[
	\braket{E_0|E_0(\delta)'} = 1 + 0 -\frac{1}{2} \sum_{\lambda \neq 0} \frac{\braket{E_0|\delta \widehat{H}| E_\lambda} \braket{E_\lambda|\delta \widehat{H}|E_0}}{(E_\lambda - E_0)^2}.
\]
By the same argument as above the error in the ground state fidelity is bounded as
\begin{equation}
	\n{1 - \braket{\psi'|\psi}} \le \frac{d^{k}}{\Delta E^2} \; \epsilon^2. 
	\label{eq:fidelity_bound}
\end{equation}
In conclusion, we have now seen that the compressed MPOs should accurately reproduce the true ground state physics and provided error bounds on the precision. This justifies our use of compressed Hamiltonians to study twisted bilayer graphene.

\section{Numerical cross checks}
\label{app:numerical_details}

In addition to the analytic error bounds from the previous section, we also performed extensive numerical checks to verify that our computations were correct in practice as well as in principle. As mentioned above, we used the standard \texttt{TeNPy} library \cite{hauschild2018efficient}, written by one of us, for all DMRG calculations. The MPO compression code was carefully verified by unit testing, benchmarking, and a variety of cross-checks. The two primary cross-checks, which we now describe, verify the accuracy of the compression algorithm and the accuracy of the transformations between the various representations of the Hamiltonian.

\begin{table}
    \begin{center}
    
    \begin{ruledtabular}
    \begin{tabular}{rccc}
           & $\v{k}$ space & $xk$ space & MPO \\ \colrule
    $E_\text{kin}$ (\si{\milli\electronvolt}) &        -69.099        &         -69.095         &   -69.095           \\
    $E_\text{int}$ (\si{\milli\electronvolt}) &         32.263      &       32.257           &         32.257 \\
    $\Delta E_\text{kin}$ (\si{\milli\electronvolt}) &         -      &       $4.2 \times 10^{-3}$           &         $4.2 \times 10^{-3}$ \\
    $\Delta E_\text{int}$ (\si{\milli\electronvolt}) &         -      &       $7.0 \times 10^{-3}$         &         $6.4 \times 10^{-3}$ \\
    \end{tabular}
    \end{ruledtabular}
    \end{center}
    \caption{Energy of $\ket{\psi_{\text{kin}}}$ per momentum per band at $N_y = 2$, $w_0/w_1 = 0.825, d = 30 \si{nm}$. MPO compression was performed with singular value truncation cutoff at $10^{-3} \si{\milli\electronvolt}$. The energy difference is calculated against the $k$ space result.}
    \label{tab:correctness}
\end{table}

\subsection{Gauge Transform Verification}

It is crucial that the compression algorithm is not only precise, as we have shown in previous sections, but also accurate. That is, the output of the implementation of the compression algorithm is indeed the compressed MPO described analytically. To verify this, we use the fact that Algorithm \ref{alg:iMPO_compression} is a gauge transformation, up until the truncation step. This gauge transformation obeys gauge relations given in Sec.~\ref{sec:UCMPO_canonicalization}, and reproduced here for convenience:
\begin{equation}
    \begin{aligned}
        	R_{n-1} \W_R^{n} &= \W^{(n)} R_n,\\
C_{n-1} \W_R^{(n)} &= \W_L^{(n)} C_n,\\
S_{n-1} \P^{(n)} &= \Q^{(n)} S_n.\\
    \end{aligned}
\end{equation}
Due to the large number of small matrix elements, many indexing errors and other accuracy problems only manifest as small errors in the gauge relations. We therefore verified the gauge relations to precision $10^{-13}$, nearly the floating point limit. Together with checks for canonicality of $\W_{R, L}$, this constitutes a sufficient check for the correctness of the algorithmic implementation.

\subsection{Cross checks for \tBLG{} Hamiltonian}
In order to perform HF and DMRG calculations, one needs to represent the Hamiltonian in a variety of ways, as shown in Fig.~\ref{fig:flowchart}, and it is imperative to make sure there is no error when we transform one representation to another. This section reviews a list of numerical checks we performed to guarantee correctness.

We first start from the momentum space representation, which is suitable for HF calculations. This interaction is specified by $V_{\v{q}}$, together with the form factors $\Lambda_{\v{q}}$. This representation can be Wannier localized (i.e. Fourier transformed) to obtain the mixed-$xk$ space representation for iDMRG. Finally we use the mixed-$xk$ space representation to construct an MPO, and compress it down to a smaller bond dimension. The first transformation is a unitary transformation, and is in principle exact up to numerical precision, whereas the precision of the compression is limited by MPO singular value cutoff.

In order to check if these transformations are accurate, we calculate a physical observable using each representation. The kinetic part $H_{\mathrm{kin}}$ of the BM Hamiltonian $h(k)$ is gapped for a wide range of parameters for small $N_y$, and the ground state $\ket{\psi_{\text{kin}}}$ is easy to calculate in each representation. Therefore, we can easily evaluate the following energies in momentum space, mixed-$xk$ space, and DMRG.

\begin{equation}
\begin{aligned}
    E_{\text{kin}} & = \bra{\psi_{\text{kin}}} H_{\text{kin}} \ket{\psi_{\text{kin}}}\\
    E_{\text{int}} & = \bra{\psi_{\text{kin}}} H_{\text{int}} \ket{\psi_{\text{kin}}}
\end{aligned}
\end{equation}
where $H_\text{kin}$ and $H_\text{int}$ are the kinetic and interaction part of the Hamiltonian, respectively. The comparison for $N_y = 2$ is shown in Table~\ref{tab:correctness}. We see that the energy error is very small. We further check that these errors decrease as the accuracy of each calculation is increased (e.g.~by increasing the cutoff range for MPO creation). We may conclude that the transformations between Hamiltonian representations are sufficiently accurate.
	
\end{document}